\documentclass[11pt]{article}
\usepackage{amsmath, amssymb, amsfonts, amsthm,latexsym}
\usepackage[all,cmtip]{xy}
\usepackage{graphicx}
\usepackage{pstricks}
\usepackage{verbatim}

\newtheorem*{theorem*}{Theorem}
\newtheorem{theorem}{Theorem}[section]

\newtheorem{lem}[theorem]{Lemma}
\newtheorem{lemma}[theorem]{Lemma}

\newtheorem{propos}[theorem]{Proposition}

\newtheorem{defin}[theorem]{Definition}
\newtheorem{cor}[theorem]{Corollary}

\DeclareMathOperator{\hm}{Hom}

\DeclareMathOperator{\LC}{LC}
\DeclareMathOperator{\cp}{CP}
\DeclareMathOperator{\LL}{LC}

\DeclareMathOperator{\col}{Col}
\newcommand{\ind}{1{\hskip -2.5 pt}\hbox{I}}
\newcommand{\bicon}{co-connected}
\newcommand{\Bicon}{Co-connected}

\DeclareMathOperator{\mt}{Mod_3}

\newcommand{\dist}{\mathrm{dist}}

\def\Z{\mathbb Z}
\def\ZZ{\mathbb Z^d}
\def\RR{\mathbb R^d}
\def\N{\mathbb N}
\def\C{\mathcal C}
\def\W{\mathcal W}
\def\V{\mathcal V}
\def\U{\mathcal U}
\def\T{\mathcal T}
\def\TT{\mathbb T_n^d}
\def\L{\mathcal L}

\DeclareMathOperator\Ty{Type}
\DeclareMathOperator\QP{QP}
\def\E{\mathbb E}
\renewcommand\P{\mathbb P}

\def\zero{\mathbf{0}}

\def\PC{\text{PC}}

\def\intb{\partial_\bullet}
\def\extb{\partial_\circ}

\newcommand\heading[1]{\vspace{5pt}\noindent{\bf #1.}}

\newcommand{\ronthanks}{School of Mathematics, Raymond and Beverly Sackler Faculty of Exact
Sciences, Tel Aviv University, Tel Aviv, Israel. E-mail: {\tt peledron@post.tau.ac.il}.
Research Supported by an ISF grant and an IRG grant.}

\newcommand{\ohadthanks}{School of Mathematics, Raymond and Beverly Sackler Faculty of Exact
Sciences, Tel Aviv University, Tel Aviv, Israel. E-mail:
{\tt ohad\_f@netvision.net.il}. Research supported by an ERC advanced grant.}

\title{Rigidity of 3-colorings of the discrete torus}
\author{Ohad N. Feldheim\thanks{\ohadthanks} \and Ron Peled\thanks{\ronthanks}}

\begin{document}
\maketitle
\begin{abstract}
We prove that a uniformly chosen proper $3$-coloring of the
$d$-dimensional discrete torus has a very rigid structure when the
dimension $d$ is sufficiently high. We show that with high
probability the coloring takes just one color on almost all of
either the even or the odd sub-torus. In particular, one color
appears on nearly half of the torus sites. This model is the zero
temperature case of the $3$-state anti-ferromagnetic Potts model
from statistical physics.

Our work extends previously obtained results for the discrete torus
with specific boundary conditions. The main challenge in this
extension is to overcome certain topological obstructions which
appear when no boundary conditions are imposed on the model.
Locally, a proper $3$-coloring defines the discrete gradient of an
integer-valued height function which changes by exactly one between
adjacent sites. However, these locally-defined functions do not
always yield a height function on the entire torus, as the gradients
may accumulate to a non-zero quantity when winding around the torus.
Our main result is that in high dimensions, a global height function
is well defined with high probability, allowing to deduce the rigid
structure of the coloring from previously known results. Moreover,
the probability that the gradients accumulate to a vector $m$,
corresponding to the winding in each of the $d$ directions, is at
most exponentially small in the product of $\|m\|_\infty$ and the
area of a cross-section of the torus.

In the course of the proof we develop discrete analogues of notions
from algebraic topology. This theory is developed in some generality
and may be of use in the study of other models.

\let\thefootnote\relax\footnotetext{\emph{MSC2010 Subject classification.} 82B20, 82B26, 82B41, 60C05, 60D05, 60K35, 05A16.}
\let\thefootnote\relax\footnotetext{\emph{Keywords.} 3-colorings, Potts model, rigidity, discrete topology, discrete cohomology, 3-states.}
\end{abstract}
\section{Introduction}\label{sec: intro}
We study proper $3$-colorings of $\TT$, the d-dimensional discrete
torus $(\Z/n\Z)^d$, whose side length $n$ is even. Our main theorem
is that in high dimensions, a uniformly chosen proper $3$-coloring
of $\TT$ is nearly constant on one of the two bipartition classes of
$\TT$. Precisely, denote the partite classes of $\TT$ by $V^0$ and
$V^1$. A \emph{proper $3$-coloring} of $\TT$ is a function
$f\colon\TT\to\{0,1,2\}$ satisfying $f(v)\neq f(w)$ whenever $v$ and
$w$ are adjacent in $\TT$. Denote by $\cp_{i,k}(f)$ the proportion
of color $k$ on $V^i$, that is,
\begin{equation*}
\cp_{i,k}(f)\,:=\frac{|\{v\in V^i\ :\ f(v)= k\}|}{|V^i|}.
\end{equation*}
\begin{theorem}\label{thm: main}
There exist $d_0, c>0$ such that for every integer $d\ge d_0$ and every
even integer $n$, a uniformly chosen proper $3$-coloring $f\colon\TT\to\{0,1,2\}$
satisfies
\begin{equation*}
\E \bigg(\min_{i\in\{0,1\}} \cp_{i,k}(f)\bigg)\le
\exp\left(-\frac{cd}{\log^2 d}\right)\quad\text{for all
}k\in\{0,1,2\}.
\end{equation*}
\end{theorem}
Thus, the theorem asserts that typically in high dimensions, for
each color there is a partite class on which the color hardly
appears. Equivalently, one of the partite classes is dominated by a
single color.

The next section describes the main idea of the proof. More precise
definitions are given in Section~\ref{sec: prelim}.

\subsection{Relation with height functions}\label{sec:idea_of_proof}
Our proof of Theorem~\ref{thm: main} exploits a connection between
proper $3$-colorings and height functions, which we now describe. It
is convenient to introduce the required notions on a general graph.
Suppose $G$ is a connected, bipartite graph with a fixed vertex
$v_0\in V(G)$. Let $\col(G,v_0)$ be the set of all proper
$3$-colorings of $G$ taking the value $0$ at $v_0$. That is,
\begin{equation}\label{eq:col_G_def}
  \col(G,v_0):=\{f:V(G)\to\{0,1,2\}\,:\, f(v_0)=0,\, f(v)\neq
  f(w)\text{ when } (v,w)\in E(G)\}.
\end{equation}
An integer-valued function on $V(G)$ is called a \emph{homomorphism
height function} on $G$, or simply height function or HHF, if it
differs by exactly one between adjacent vertices of $G$. Let
$\hm(G,v_0)$ be the set of all homomorphism height functions on $G$
which take the value $0$ at $v_0$. Precisely,
\begin{equation}\label{eq:hom_G_def}
  \hm(G,v_0):=\{f:V(G)\to\Z\,:\, f(v_0)=0,\, |f(v)-f(w)|=1\text{ when } (v,w)\in E(G)\}.
\end{equation}

In this paper, we always take $G$ to be either $\TT$ or $\ZZ$ for
some $n$ and $d$. We consider both $\TT$ and $\ZZ$ to come with a
fixed coordinate system and denote by $\zero$ the vector
$(0,0,\ldots,0)$ in that system. For these graphs, we abbreviate
$\col(G,\zero)$ to $\col(G)$ and $\hm(G,\zero)$ to $\hm(G)$.

The connection we need between proper colorings and height functions
is summarized by the following two facts:
\begin{enumerate}
\item For any graph $G$, $v_0\in V(G)$ and $h\in\hm(G,v_0)$, the
function $g:V(G)\to\{0,1,2\}$ defined by
\begin{equation*}
  g(v):= h(v)\bmod 3
\end{equation*}
belongs to $\col(G,v_0)$.
\item When $G=\ZZ$, the above correspondence defines a \emph{bijection}
    between $\hm(\ZZ)$ and $\col(\ZZ)$.
\end{enumerate}
The first fact is straightforward and the second fact appears to be
folklore in the field (see Proposition~\ref{prop:col height
bijection}).

Our goal in this work is to use the above correspondence to transfer
known results on height functions, proved in \cite{PHom}, to results
on colorings, thereby obtaining Theorem~\ref{thm: main}. Our task
is, however, made complicated by the following obstruction. The
above correspondence is not a bijection when $G=\TT$. In other
words, there exist colorings in $\col(\TT)$ which are not the modulo
3 of any height function in $\hm(\TT)$. For instance, the coloring
$012012$ of $\mathbb{T}_6^1$ provides one such example. The source
of this problem is of a \emph{topological} nature, stemming from the
fact that the torus has non-contractible cycles. This poses a major
difficulty, preventing a direct use of the known results on height
functions. The following theorem, whose proof occupies most of this
paper, provides a way around this difficulty. It shows that the
above correspondence is, nonetheless, close to being bijective when
the dimension $d$ is sufficiently high.
\begin{theorem}\label{thm: almost bijection}
There exist $d_0$ and $c>0$ such that for every integer $d\ge d_0$ and every
even integer $n$, a uniformly chosen proper $3$-coloring of $\TT$ satisfies
\begin{equation*}
\P(f\text{ is not the modulo $3$ of some HHF on }\TT)\le \exp(-c_d
n^{d-1}),
\end{equation*}
with $c_d=\frac{c}{d\log^2 d}$.
\end{theorem}
In the next section we explain how Theorem~\ref{thm: main} follows
from the above theorem and a result on height functions proved in
\cite{PHom}. In Section~\ref{sec: BG} we present some background.
The rest of the paper is devoted to the proof of Theorem~\ref{thm:
almost bijection}. Section~\ref{sec: prelim} contains the first part
of the proof and a proof overview. The proof is inspired by ideas
from algebraic topology but the necessary tools are developed
completely in the discrete setting. We believe that some of these
tools could prove useful in other models as well, especially the
trichotomy theorems of Section~\ref{sec: structure of sets},
Theorem~\ref{thm: pair_trich} and Theorem~\ref{thm: main
trichotomy}, which deal with discrete counterparts of manifolds of
codimension one. The connection between our work and algebraic
topology is expounded upon in Section~\ref{subs: topology}.
Section~\ref{sec: Rem Op} is dedicated to remarks and open problems.

\subsection{Remarks and
extensions}\label{sec:remarks_and_extensions}
We point out that the bound presented in Theorem~\ref{thm: main} is
near optimal. Perhaps surprisingly, Theorem~\ref{thm: main} itself
implies the following claim.
\begin{propos}\label{prop:rigidity_lower_bound}
There exist $d_0, c>0$ such that for every integer $d\ge d_0$ and
every even integer $n$, a uniformly chosen proper $3$-coloring
$f\colon\TT\to\{0,1,2\}$ satisfies
\begin{equation*}
\E \bigg(\min_{i\in\{0,1\}} \cp_{i,k}(f)\bigg)\ge
\exp\left(-cd\right)\quad\text{for all }k\in\{0,1,2\}.
\end{equation*}
\end{propos}
This proposition is proved in Section~\ref{sec:near_optimality}.

We also emphasize that Theorem~\ref{thm: almost bijection} serves as
a bridge between results on uniformly sampled homomorphism height
functions on $\TT$ and uniformly sampled proper $3$-colorings. Thus,
results on the former may be transferred easily to the latter, as is
illustrated by the deduction of Theorem~\ref{thm: main} in
Section~\ref{sec: proof of main}. One expects it to be possible to
upgrade Theorem~\ref{thm: main} by showing that the quantity
$\min_{i\in\{0,1\}} \cp_{i,k}(f)$ is not only small on average,
but also small with high probability as $n$ tends to infinity. To
use Theorem~\ref{thm: almost bijection} to this end would require
extending the corresponding results on height functions. While we
believe such extensions are possible, we do not delve further in
this direction as our main concern in this paper is to establish the
relation between the models.

As explained in Section~\ref{sec: prelim} below, we approach
Theorem~\ref{thm: almost bijection} by identifying the set of proper
$3$-colorings with a set of quasi-periodic height functions. Each
such height function has a well-defined \emph{slope}, a vector which
measures the amount by which it changes when going around the torus
in each direction. Homomorphism height functions on $\TT$ can be
identified with quasi-periodic height functions with zero slope. The
proof of Theorem~\ref{thm: almost bijection} proceeds by finding a
one-to-one map between quasi-periodic functions of a given non-zero
slope, and a tiny subset of the quasi-periodic functions with zero
slope, see Theorem~\ref{thm: QPm QP0 embedding} below. In fact, more
can be deduced from our techniques. As we show in Theorem~\ref{thm: QP steep QP0 embedding},
the size of the set of quasi-periodic functions with a given slope
may be estimated in terms of this slope, yielding stronger bounds
for steeper slopes. For instance, the chance of sampling a proper
$3$-coloring whose corresponding height function changes by a linear
amount when going around the torus, is exponentially small in $n^d$
rather than the $n^{d-1}$ appearing in Theorem~\ref{thm: almost
bijection}.

While Theorem~\ref{thm: almost bijection} is proved in high
dimensions, the main ingredient in its proof, the above-mentioned one-to-one mapping of quasi-periodic height functions with a
given slope to quasi-periodic height functions with zero slope, is developed in all dimensions. The
part which is missing in low dimensions is a counterpart of
\cite[Theorem~2.8]{PHom}, which would show that the probability
that a low-dimensional HHF on $\TT$ has a long level line is exponentially small in this length. This result is not expected in two dimensions (see discussion in Section~\ref{sec: BG} below), but may be valid already in dimensions $d\ge 3$. Theorem~\ref{thm: almost bijection}
would immediately extend to any dimension in which this result is established. Appropriate analogs of Theorem~\ref{thm: main} in dimensions $d\ge 3$ may also be valid, as the proof of Theorem~\ref{thm: main} relies on Theorem~\ref{thm: almost bijection} and input on the fluctuations of homomorphism height functions provided in \cite{PHom} in high dimensions (see Section~\ref{sec: proof of main} below).

\subsection{Background and related works}\label{sec: BG}
Our work is not the first to establish rigidity of proper
$3$-colorings in high dimensions. Previously, a result analogous to
Theorem~\ref{thm: main} in which the proper $3$-coloring is sampled
from the set of colorings with `zero boundary conditions' was
established in \cite{PHom}, and also by Galvin, Kahn, Randall and
Sorkin in \cite{GKRS}. The restriction to such `zero boundary
conditions' makes the problem simpler from a topological point of
view since it essentially removes the non-trivial cycles of $\TT$,
rendering the correspondence described in
Section~\ref{sec:idea_of_proof} into a bijection of height functions
and proper $3$-colorings with these boundary conditions. The results
of \cite{PHom} and \cite{GKRS} imply Roman Koteck\'y's conjecture
(see \cite{K85} for context and \cite{GKRS} for additional details),
that the proper $3$-coloring model admits at least 6 different Gibbs
states in high dimensions.

Galvin and Randall \cite[theorem 2.1]{GR} established a related
result in the same setting as Theorem~\ref{thm: main}. They showed
that for each color $k$, with probability at least $1 -
\exp(-c_dn^{d-1} / \log^2 n)$, the proportions of the color on the
two bipartite classes differ by at least $\rho$, where $\rho\approx
0.22$. In terms of the quantities $\cp_{i,k}(f)$ used in
Theorem~\ref{thm: main}, this means that $|\cp_{0,k}(f) -
\cp_{1,k}(f)|\ge \rho$ with high probability. Taking into account
that each color may appear on at most half of the vertices of the
torus, this implies that $\min_{i\in\{0,1\}} \cp_{i,k}(f)\le
\frac{1 - \rho}{2}\approx 0.39$ with high probability. In contrast,
Theorem~\ref{thm: main} shows that $\E\left(\min_{i\in\{0,1\}}
\cp_{i,k}(f)\right)\le\exp(-c d / \log^2(d))$, a bound which is
near optimal by Proposition~\ref{prop:rigidity_lower_bound}. As
discussed in Section~\ref{sec:remarks_and_extensions}, we believe
this bound may be shown to hold not only on average but with high
probability as $n$ tends to infinity by extending the corresponding
results on height functions. The techniques of \cite{GR} are rather
different from ours. While we proceed by developing the topological
theory of discrete height functions, the work \cite{GR} stays fully
in the realm of $3$-colorings.

Other related results include torpid mixing of the Glauber dynamics
for proper $3$-colorings of $\TT$ \cite{GR} and the fact that
homomorphism height functions have bounded range on the hypercube
graph $\{0,1\}^d$, as proved by Kahn \cite{K01} and Galvin
\cite{G03}.

In statistical physics terminology, the proper $3$-coloring model is
the same as the zero temperature case of the antiferromagnetic
3-state Potts model. It is expected that the analog of our result
continues to hold for small, positive temperature, but this remains
unproven. In two dimensions, the model is equivalent to the uniform
six-vertex, or square ice, model (this was pointed out by Andrew
Lenard, see \cite{L67}). It is expected that the analog of
Theorem~\ref{thm: main} fails in two dimensions, as the square ice
model is conjectured to be in a disordered phase, in the sense that
the model should have a unique Gibbs state when $d=2$. However, it
may well be that multiple Gibbs states exist already for any $d\ge
3$. Investigating other graphs, Koteck\'y, Sokal and Swart
\cite{KSS12} have shown that the model has multiple Gibbs states on
certain \emph{planar} lattices. This result was extended by Huang
et. al. \cite{HCD+} who have shown that for every $q\ge 3$, there
are planar lattices on which the proper $q$-coloring model has
multiple Gibbs states.

The fact that a uniformly chosen $3$-coloring on the torus is the
modulo $3$ of a height function with high probability
(Theorem~\ref{thm: almost bijection}) is also expected to fail in
two dimensions. Some evidence for this phenomenon is provided by the
study of the dimer model. In the dimer model, one samples uniformly
a perfect matching of an underlying graph. On suitable graphs, the
perfect matching defines locally the gradient of an integer-valued
height function and one may study similar questions to those studied
here. Boutillier and de Tili\`ere \cite{BdT09} (see also Kenyon
\cite[Section 4.17]{K09}) considered the dimer model on a piece of
the hexagonal lattice wrapped around a torus. They showed that the
random height differences accumulated when winding around the torus
tend to a non-degenerate limit distribution (a discrete
Gaussian-type distribution) as the side length of the torus
increases.

It is conjectured that the rigidity phenomenon described by
Theorem~\ref{thm: main} has an analog for proper colorings with more
than $3$ colors. Specifically, that for any $q\ge 4$ there exists a
$d_0(q)$ such that a uniformly sampled proper $q$-coloring of $\TT$,
$d\ge d_0(q)$, has the following structure with high probability.
The colors split into two sets of sizes $\lfloor q/2\rfloor$ and
$\lceil q/2\rceil$, with the even sublattice colored predominantly
by colors from one set and the odd sublattice colored predominantly
by colors from the other set. While this conjecture remains open,
several related results have appeared. Galvin and Tetali
\cite{GT04}, following work of Kahn \cite{K01-2}, gave approximate
counts for the number of \emph{graph homomorphisms} from $d$-regular
graphs to arbitrary finite graphs. Specializing to proper
$q$-colorings of $\TT$, their results support the above conjecture.
Meyerovitch and Pavlov \cite{MP14} analyzed, so called, axial
products of shifts of finite type, a more general model than graph
homomorphisms on $\Z^d$, and found explicit expressions for the
limiting topological entropy of such models as $d$ tends to
infinity. Their results are also in agreement with the above
conjecture. Galvin and Engbers \cite{EG12} established the analog of
the conjecture, and more general rigidity results for graph
homomorphisms, in the limit when $n$ is fixed and $d$ tends to
infinity. Similar rigidity results on expander and tree graphs are
established in \cite{PSY13, PSY13-2, PSY13-3}.

Of related interest is the hard-core model in $\TT$. In this model,
one samples an independent set $I$ of $\TT$ with probability
proportional to $\lambda^{|I|}$. It is expected that there exists
some $\lambda_c = \lambda_c(d)$ satisfying that, with high
probability, if $\lambda>\lambda_c$ the sampled independent set
resides predominantly in one of the two sublattices, whereas if
$\lambda<\lambda_c$ no such structure appears. While the existence
of $\lambda_c$ is still open (and there are examples of graphs for
which it does not exist, see \cite{BHW99}) one may still define
$\lambda_c'=\lambda_c'(d)$ as the infimum over $\lambda$ for which
the model admits multiple Gibbs states. Dobrushin \cite{D68} proved
that $\lambda_c'<\infty$ in every dimension $d\ge 2$, with an upper
bound growing to infinity with $d$. Galvin and Kahn \cite{GK04}
significantly improved this result by showing that $\lambda_c'$
tends to zero with $d$. The quantitative bound obtained in
\cite{GK04} was further improved in \cite{PS13}. The main technical
ingredient in both \cite{GK04, PS13}, as well as the aforementioned
\cite{PHom, GKRS}, is a careful analysis of the structure of certain
special cutsets in $\TT$, when the dimension $d$ is sufficiently
high. This is in contrast to this work, in which discrete analogs of
topological considerations constitute the bulk of the argument.

\subsection{Proof of Theorem 1.1}\label{sec: proof of main} We end the introduction by explaining how to deduce Theorem~\ref{thm: main} from Theorem~\ref{thm:
almost bijection} and a result of \cite{PHom} on the fluctuations of
typical homomorphism height functions on $\TT$.

We start with the following lemma, which states the required result
on the typical behavior of height functions.
\begin{lemma}\label{lem: ron}
There exist $c>0$ and $d_0$ such that in all dimensions $d\ge d_0$,
if $h$ is uniformly sampled from $\hm(\TT)$ then
\begin{equation*}
  \P(|h(u)-h(v)|\ge 3)\le \exp\left(-\frac{cd}{\log^2 d}\right)
  \qquad\forall u,v\in\TT.
\end{equation*}
\end{lemma}
\begin{proof}
Theorem~2.1 in \cite{PHom} gives, in particular, that there exist
$c>0$ and $d_0$ such that in all dimensions $d\ge d_0$ and for every
$u,v\in\TT$, if $h$ is uniformly sampled from $\hm(\TT,u)$, then
\begin{equation*}
  \P(|h(v)|\ge 3)\le \exp\left(-\frac{cd}{\log^2 d}\right).
\end{equation*}
The lemma follows from this by using the fact that the mapping
$T_u\colon\hm(\TT)\to\hm(\TT, u)$ defined by $T_u(h)(v):=h(v)-h(u)$
is a bijection.
\end{proof}

We are now ready to prove Theorem~\ref{thm: main}. First, observe
that by symmetry, it suffices to prove the theorem for a uniformly
chosen coloring in $\col(\TT)$, i.e., a coloring normalized at
$\zero$.

Let $f$ be uniformly chosen from $\col(\TT)$. Recall that
$$\cp_{i,k}(f)=\frac{|\{v\in V^i\ :\ f(v)= k\}|}{|V^i|},$$
where $V^0$ and $V^1$ are the partite classes of $\TT$. Fix
$k\in\{0,1,2\}$ and let
\begin{equation*}
X:=\min_{i\in\{0,1\}}\cp_{i,k}.
\end{equation*}
We need to show that $\E(X)\le\exp(-cd/\log^2d)$ for some $c>0$ and
all sufficiently high $d$.

Fix $d$ sufficiently high and $c>0$ sufficiently small for the
following arguments. Define the event
\begin{equation*}
  A:=\{\text{$f$ is the modulo $3$ of some HHF in
$\hm(\TT)$}\}.
\end{equation*}
By symmetry again, Theorem~\ref{thm: almost bijection} implies that
\begin{equation*}
  \P(A^c)\le \exp\left(-\frac{c}{d\log^2d}n^{d-1}\right).
\end{equation*}
Hence,
\begin{equation}\label{eq: expectancy argu}
\E(X)=\E(X\ind_A)+\E(X\ind_{A^c})\le\E(X|A)+\exp\left(-\frac{c}{d\log^2d}n^{d-1}\right).
\end{equation}
Thus we focus on estimating $\E(X|A)$. Conditioning on $A$, there
exists some $h\in\hm(\TT)$ for which $f\equiv h\pmod3$. Moreover,
since distinct functions in $\hm(\TT)$ give rise to distinct
colorings in $\col(\TT)$ under the modulo 3 operation, it follows
that, conditioned on $A$, $h$ is uniformly distributed in
$\hm(\TT)$.

Now note that if $u,v\in\TT$ are vertices in different partite
classes of $\TT$ then $h(u)$ and $h(v)$ have different parity. Thus,
for such vertices, we have the following containment of events,
\begin{equation*}
\{f(u)=f(v)\} = \{h(u)\equiv h(v)\!\pmod 3\} \subseteq
\{|h(u)-h(v)|\ge 3\}.
\end{equation*}
We conclude that $X$ satisfies the following relation.
\begin{align*}
X^2&=\frac{1}{|V^0|^2}\min_{i\in\{0,1\}}|\{v\in V^i\ :\
f(v)=k\}|^2\le\frac{1}{|V^0|^2} |\{v\in V^0\ :\
f(v)=k\}|\cdot|\{v\in V^1\ :\ f(v)=k\}|\le\\
&\le \frac{1}{|V^0|^2}\sum_{u\in V^0,\,v\in
V^1}\hspace{-3pt}\ind_{(f(u)=f(v))}\le \frac{1}{|V^0|^2}\sum_{u\in
V^0,\,v\in V^1}\hspace{-3pt}\ind_{(|h(u)-h(v)|\ge3)}.
\end{align*}
Hence, we may use Lemma~\ref{lem: ron} to deduce that
\begin{equation}\label{eq: expectency calc}
\E(X|A)\le\sqrt{\E(X^2|A)}\le\frac{1}{|V^0|}\sqrt{\sum_{\substack{u\in
V^0,\, v\in V^1}}\P(|h(u)-h(v)|\ge3)}\le \exp\left(-\frac{cd}{\log^2
d}\right).
\end{equation}
Together with \eqref{eq: expectancy argu}, this establishes
Theorem~\ref{thm: main}.

\section{Preliminaries and Overview}\label{sec: prelim}
This section is divided into an introduction to the objects and notation of
the paper, and to a reduction of Theorem \ref{thm: almost bijection} to a
statement concerning quasi-periodic functions on the integer lattice. At the
end of the section we give a glimpse into the ideas of the proof, and discuss
the relation between our work and algebraic topology.

\subsection{Preliminary definitions}\label{subs: predef}

\heading{Lattice and Torus} We write $\ZZ$ for the nearest-neighbor
graph of the standard $d$-dimensional integer lattice, and
$\TT=(\Z/n\Z)^d$ for the graph of the $d$-dimensional discrete torus
with side length $n$. We assume $n$ is an even integer greater or
equal than 4, fixing it throughout the paper. We also assume both
graphs come with a fixed coordinate system, letting $e_i\in\ZZ$ be
the $i$th standard basis vector for $1\le i\le d$. In both graphs,
two vertices are adjacent if they differ by one in exactly one
coordinate. As $n$ is even, both graphs are bipartite. In both we
thus refer to the vertices in the bipartition class of
$\zero=(0,\dots,0)$ as \emph{even}, and to the rest of the vertices
as \emph{odd}. For a vector $v\in\ZZ$, and a set $U\in\ZZ$ we write
$U+v$ to denote $\{u+v :\ u\in U\}$.

\heading{Distance and boundary} Let $G$ be a connected graph. We
write $u\sim v$ to denote that a pair of vertices $u,v\in V(G)$ are
adjacent. For a set of vertices $U\subseteq V(G)$ we define the
\emph{boundary} of $U$ to be the set of edges
\begin{equation*}
\partial U:=\{ e\in E(G) :  e\cap U\neq \emptyset\text{ and } e\cap U^c\neq \emptyset\}.
\end{equation*}
We use $\dist(u,v)$ for the shortest-path distance between $u$ and
$v$, and extend this notion to non-empty sets $U,V\subseteq V(G)$,
defining
\begin{equation*}
\dist(U,V):=\min\{\dist(u,v)\ :\ u\in U, v\in V\}.
\end{equation*}
If one of the sets $U,V$ is empty, we write $\dist(U,V)=\infty$. For
a set of vertices $U$, we denote
\begin{align*}
U^+&:=\{u \in V(G) : \dist(\{u\},U)\le 1\},\\
U^-&:= \{u \in V(G) : \dist(\{u\},U^c)>1\}.
\end{align*}

Note that $U^- = ((U^c)^+)^c$. We also abbreviate $U^{++}:=(U^+)^+$
and $U^{--}:=(U^-)^-$. The following simple relations hold for any
two sets $U,V\subseteq V(G)$:
\begin{gather}\label{eq: Basic dist eq1}
U^+ \subseteq V \Longleftrightarrow U\subseteq V \text{ and } \partial U \cap
\partial V = \emptyset,\\
\dist(U^+,V)=\max(\dist(U,V)-1,0),\label{eq: Basic dist eq2}\\
U \subseteq V \Longleftrightarrow \forall \,W\subset
V(G),\,\dist(U,W)\ge \dist(V,W).\label{eq: Basic dist eq3}
\end{gather}

For a set of vertices $U$, we define the \emph{internal vertex
boundary} of $U$ to be
\begin{equation*}
\intb U:=U\setminus U^-.
\end{equation*}
Similarly we define the \emph{external vertex boundary} of $U$ to be
\begin{equation*}
\extb U:=U^+\setminus U.
\end{equation*}
In both $\ZZ$ and $\TT$, we call a set of vertices $U$ \emph{odd} if
all the vertices of $\intb{U}$ have the same parity (in \cite{PHom}
a different convention is used, calling a set $U$ odd if all
vertices of $\intb{U}$ are odd). The internal and external vertex
boundaries of an odd set of vertices $U\subsetneq T_{10}^2$, as well
as $U^+$ and $U^-$, are depicted in Figure~\ref{fig: U stuff}.

\begin{figure}[htb!]
\centering%
    \rput(9.25,4.3){$U^+$}
    \rput(5.4,4.3){$U$}
    \rput(1.58,4.3){$U^-$}
      \rput(3.55,0.15){$\intb U$}
      \rput(7.3,0.15){$\extb U$}
\includegraphics[scale=0.25]{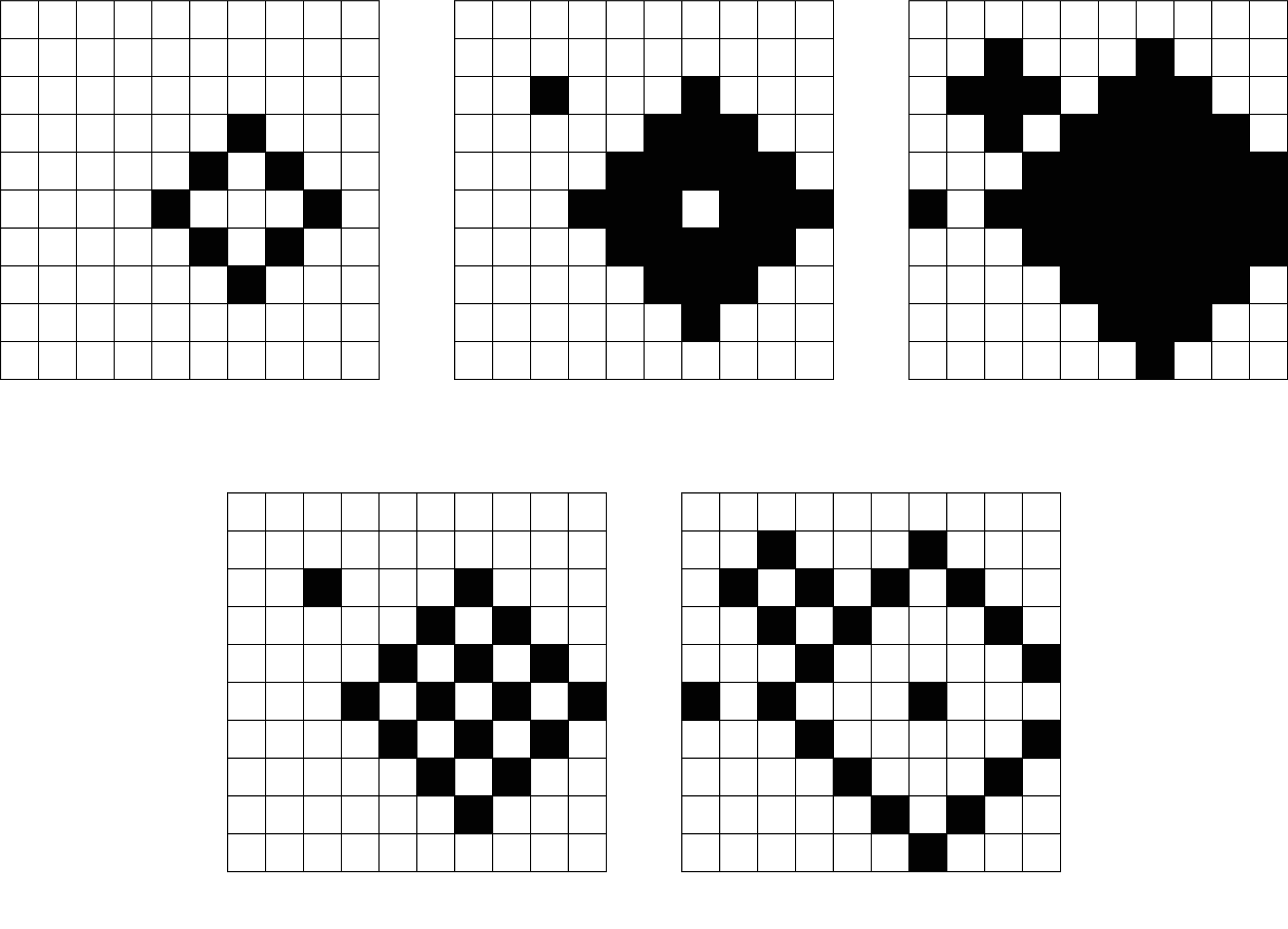}\\
\caption{Boundary operations on some odd set $U$ in $T_{10}^2$.} \label{fig: U stuff}
\end{figure}
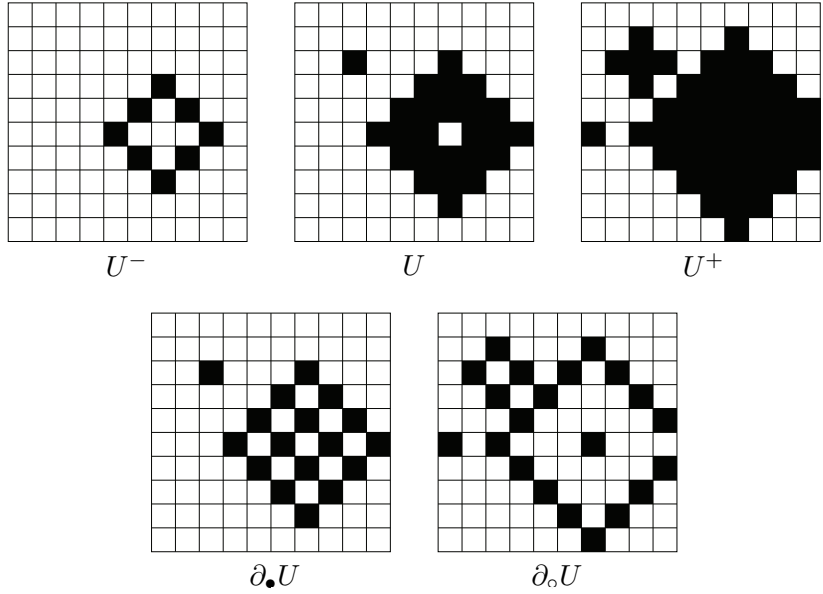

\heading{Homomorphism height functions, $3$-colorings and
quasi-periodic functions} A proper $3$-coloring of a graph $G$ is a
function $f:V(G)\to\{0,1,2\}$ satisfying $f(v)\neq f(w)$ when
$(v,w)\in E(G)$. An integer-valued function on $V(G)$ is called a
\emph{homomorphism height function} on $G$, or simply height
function or HHF, if it differs by exactly one between adjacent
vertices of $G$. We usually work with $\col(G, v_0)$ and $\hm(G,
v_0)$, the sets of colorings and height functions normalized to take
the value $0$ at the vertex $v_0$, as defined in
\eqref{eq:col_G_def} and \eqref{eq:hom_G_def}. When $G=\TT$ or $\ZZ$
we abbreviate $\col(G,\zero)$ to $\col(G)$ and $\hm(G,\zero)$ to
$\hm(G)$.

Let $V$ be either $\Z$ or $\{0,1,2\}$. We say that a function
\begin{equation*}
  \text{$f\colon\ZZ\to V$ is \emph{periodic} if $f(v)=f(w)$ whenever
$v-w=n e_i$ for some $i$}.
\end{equation*}
We denote by $\PC$ the set of periodic proper $3$-colorings in
$\col(\ZZ)$. Similarly, for an integer vector
$m=(m_1,\dots,m_d)\in\ZZ$, we say that a function
\begin{equation*}
  \text{$h\colon\ZZ\to \Z$ is \emph{quasi-periodic} with slope $m$ if $f(v)=f(w)+m_i$ whenever $v-w=n e_i$ for some $i$}.
\end{equation*}
We write $\QP_m$ for the set of quasi-periodic HHFs with slope $m$
in $\hm(\ZZ)$. Note that for an HHF, being periodic is equivalent to
being quasi-periodic with slope $\zero$. We remark that our
definition of slope is not completely standard and it may be equally
natural to say that a quasi-periodic function with slope $m$,
according to our definition, has, in fact, slope $\frac{1}{n}\cdot
m$. Our definition is chosen as it is convenient to work with
integer vectors, keeping in mind that $n$ is fixed throughout the
paper.

Observe that, in fact,
\begin{equation}\label{eq:QP_m_restrictions}
  \text{$\QP_m=\emptyset$ if $m\notin 2\ZZ$ or if $|m_i|>n$ for some
  $i$}.
\end{equation}
To see this, note that any $h\in\hm(\ZZ)$ must take even values on
even vertices, and satisfy $|h(v)|\le\dist(v,0)$, since $h$ changes
by one between adjacent vertices. Thus, we must have that $m_i=h(n
e_i)$ is even and $|h(n e_i)|\le n$ for all $i$. The quasi-periodic
functions whose slope is not a multiple of $6$ will not play a role
in our work, as we show in Proposition~\ref{prop: QP2}. Thus we
define
\begin{equation}\label{eq: bound on slopes}
\QP:= \bigcup_{m\in 6\ZZ\cap [-n,n]^d}\QP_m.
\end{equation}

Denote by $\pi\colon\ZZ\to\TT$ the natural projection from the integer
lattice to the torus, defined by
\begin{equation*}
\pi((v_1,\ldots, v_d)) = (v_1\bmod n,\ldots, v_d\bmod n)
\end{equation*}
(where we identify the coordinate system of the torus with
$\{0,\ldots,n-1\}^d$). Observe that $\pi$ extends naturally to a
bijection between periodic proper $3$-colorings (of $\ZZ$) and
proper $3$-colorings of $\TT$, as well as to a bijection between
periodic HHFs (on $\ZZ$) and HHFs on $\TT$. With a slight abuse of
notation we also denote these extensions by $\pi$.

\heading{Relations between HHFs and $3$-colorings} It is not
difficult to see that the mapping $\mt$, which takes an HHF $h$ to
the function defined by
\begin{equation*}
\mt(h)(v) := h(v)\bmod 3,
\end{equation*}
maps every HHF to a proper $3$-coloring. As mentioned in the
introduction, it is a known fact that $\mt$ defines a bijection
between $\hm(\ZZ)$ and $\col(\ZZ)$, that is between the set of HHFs
on $\ZZ$ normalized at $\zero$ and the set of proper $3$-colorings
of $\ZZ$ normalized at $\zero$. As we could not locate a reference
for this fact, we provide a short proof now.

\begin{propos}\label{prop:col height bijection}
The map $\mt$ defines a bijection between $\hm(\ZZ)$ and
$\col(\ZZ)$.
\end{propos}
\begin{proof}
We first check that $\mt$ is an injective map. Suppose $h_1,
h_2\in\hm(\ZZ)$ are two distinct height functions with $\mt(h_1) =
\mt(h_2)$. As $h_1(\zero) = h_2(\zero) = 0$, it follows that there
exist two adjacent vertices $v,w\in\ZZ$ satisfying that
$h_1(v)=h_2(v)$ but $h_1(w)\neq h_2(w)$. However, as $|h_1(v) -
h_1(w)| = |h_2(v) - h_2(w)| = 1$, this contradicts our assumption
that $\mt(h_1)(w) = \mt(h_2)(w)$.

We proceed to show that $\mt$ is onto. Let $f\in\col(\ZZ)$. Our goal
is to define an $h\in\hm(\ZZ)$ satisfying that $\mt(h) = f$. First,
define a spanning tree $\T$ of $\Z^d$, rooted at $\zero$, as
follows: Given $v=(v_1,\ldots, v_d)\in\ZZ\setminus\{\zero\}$ let
$k(v)$ equal the minimal $k$ for which $v_k\neq 0$. Define the
parent $v^*$ of $v$ in $\T$ by setting $v^*_j = v_j$ for all $j\neq
k(v)$ and setting $v^*_{k(v)} = v_{k(v)} - 1$ if $v_{k(v)}>0$ or
$v^*_{k(v)} = v_{k(v)} + 1$ if $v_{k(v)}<0$, noting that $v^*\sim v$
and $\dist(v^*,\zero) = \dist(v,\zero)-1$. Now define $h(v)$ by
induction on $\dist(v,\zero)$. Set $h(\zero) := 0$ and, for
$v\in\ZZ\setminus\{\zero\}$,
\begin{equation}\label{eq:h_from_f_def}
  \text{set $h(v)$ to be the unique integer satisfying $|h(v) -
  h(v^*)|=1$ and $h(v)\equiv f(v)\pmod 3$}.
\end{equation}
As we clearly have $\mt(h) = f$, it remains to verify that
$h\in\hm(\ZZ)$.

Let $v,w\in\ZZ$ be adjacent vertices. We need to show that
\begin{equation}\label{eq:hom_condition_for_bijection}
|h(v) - h(w)|=1.
\end{equation}
Assume without loss of generality that
$\dist(v,\zero)=\dist(w,\zero)+1$. We proceed again by induction on
$\dist(v,\zero)$. If $w=\zero$ or $v_j\neq w_j$ for some $j\le k(w)$
then necessarily $v^* = w$, whence
\eqref{eq:hom_condition_for_bijection} follows from
\eqref{eq:h_from_f_def}. Otherwise, observe that $v^*\sim w^*$. By
the induction assumption, $|h(v^*) - h(w^*)| = 1$. Using also the
fact that $|h(v) - h(v^*)| = |h(w) - h(w^*)| = 1$ and
$h(v)\not\equiv h(w)\pmod 3$ by \eqref{eq:h_from_f_def}, it follows
that \eqref{eq:hom_condition_for_bijection} holds, as required.
\end{proof}

This bijection does not extend to $\TT$, as there are colorings in
$\col(\TT)$ which are not the image of any HHF through $\mt$.
Nonetheless, $\col(\TT)$ is still in bijection with a subclass of
quasi-periodic functions, as the following proposition states.

\begin{propos}\label{prop: QP2}
The mapping
$\pi\circ\mt\colon\QP\to\col(\TT)$ is a bijection.
\end{propos}
\begin{proof}

We first show that the mapping is well-defined. Let $h\in\QP_m$ for
some $m\in 6\ZZ$. By quasi-periodicity, $h(v)\equiv h(v+n e_i) \pmod
3$, for all $1\le i\le d$ and $v\in\ZZ$. Consequently $\mt(h)\in\PC$
and hence $\pi$ may be applied to $\mt(h)$ to produce an element of
$\col(\TT)$.

Since $\mt$ is a bijection between $\hm(\ZZ)$ and $\col(\ZZ)$ and
$\pi$ is a bijection between $\PC$ and $\col(\TT)$, we deduce that
$\pi\circ\mt$ is one-to-one on $\QP$. All that remains in order to
show that this mapping is a bijection, is to prove that it is onto.

Let $f\in\col(\TT)$. Define $g:=\pi^{-1}(f)\in\PC$ and an HHF $h$ by
$h:=\mt^{-1}(g)$. We need to show that $h\in \QP_m$ for some $m\in
6\ZZ\cap [-n,n]^d$. We first show that for any $v,w\in \ZZ$ and
$1\le i\le d$,
\begin{equation*}
  h(v + ne_i) - h(v) = h(w + ne_i) - h(w).
\end{equation*}
For this it suffices to show that for any $v\in\ZZ$ and $1\le i,j\le
d$,
\begin{equation}\label{eq:quasi_periodic_cond}
  h(v + ne_i) - h(v) = h(v + e_j + ne_i) - h(v + e_j).
\end{equation}
Since $h(v+e_j) - h(v)$ and $h(v+e_j + ne_i) - h(v+ne_i)$ are both
in $\{-1,1\}$ by the definition of homomorphism height function, the
equality \eqref{eq:quasi_periodic_cond} follows upon recalling that
$g = \mt(h)$ and noting that
\begin{equation*}
g(v + e_j) - g(v) = g(v + e_j + ne_i) - g(v + ne_i),
\end{equation*}
since $g$ is periodic. Thus $h\in \QP_m$ for some $m\in \ZZ$.

It remains to show that $m\in 6\ZZ\cap [-n,n]^d$. By
\eqref{eq:QP_m_restrictions} it suffices to show that $m\in 3\ZZ$.
This follows from the fact that
\begin{equation*}
m_i=h(n e_i)\equiv g(n e_i) = g(\zero) = 0 \pmod3.\qedhere
\end{equation*}
\end{proof}

Proposition~\ref{prop: QP2} enables us to define the following
partition of $\col(\TT)$,
\begin{equation}\label{eq:col_m_def}
\col_m(\TT):= (\pi\circ\mt)(\QP_m).
\end{equation}
It also implies the important fact that $\col_\zero(\TT)$ and
$\hm(\TT)$ are in bijection via $\pi \circ
\mt^{-1}\circ\pi^{-1}$. In other words,
\begin{equation}\label{eq:col_0_characterization}
  \col_\zero(\TT) = \{f\in \col(\TT)\,:\, f\text{ is the modulo $3$ of some }
h\in\hm(\TT)\}.
\end{equation}

The relations between $\col(\TT), \hm(\TT), \QP$ and $\PC$ are
summarized in Figure~\ref{fig: QP COL}.

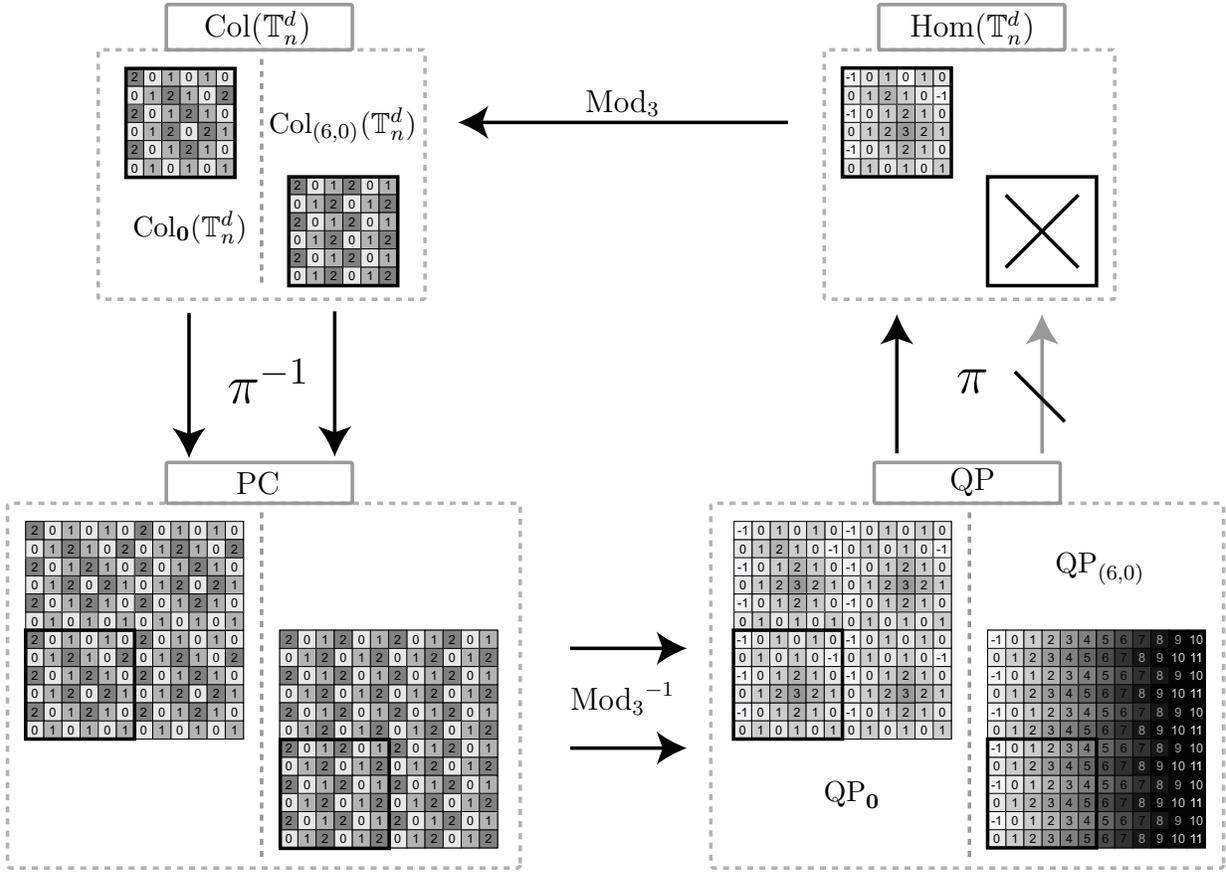
\begin{figure}[htb!]
\centering%
    \rput(4.48,9.9){$\col_{(6,0)}(\TT)$}
    \rput(2.45,8.55){$\col_{\zero}(\TT)$}
    \rput(14.55,4){$\QP_{(6,0)}$}
    \rput(11.25,1){$\QP_{\zero}$}
    \rput(3.35,11.2){\large$\col(\TT)$}
    \rput(12.85,11.2){\large$\hm(\TT)$}
    \rput(3.35,5.15){\large$\PC$}
    \rput(12.85,5.15){\large$\QP$}
    \rput(8.2,10.2){\large$\mt$}
    \rput(8.2,2.28){\large$\mt^{-1}$}
    \rput(3.49,6.5){\huge$\pi^{-1}$}
    \rput(12.85,6.5){\huge$\pi$}
\includegraphics[scale=0.19]{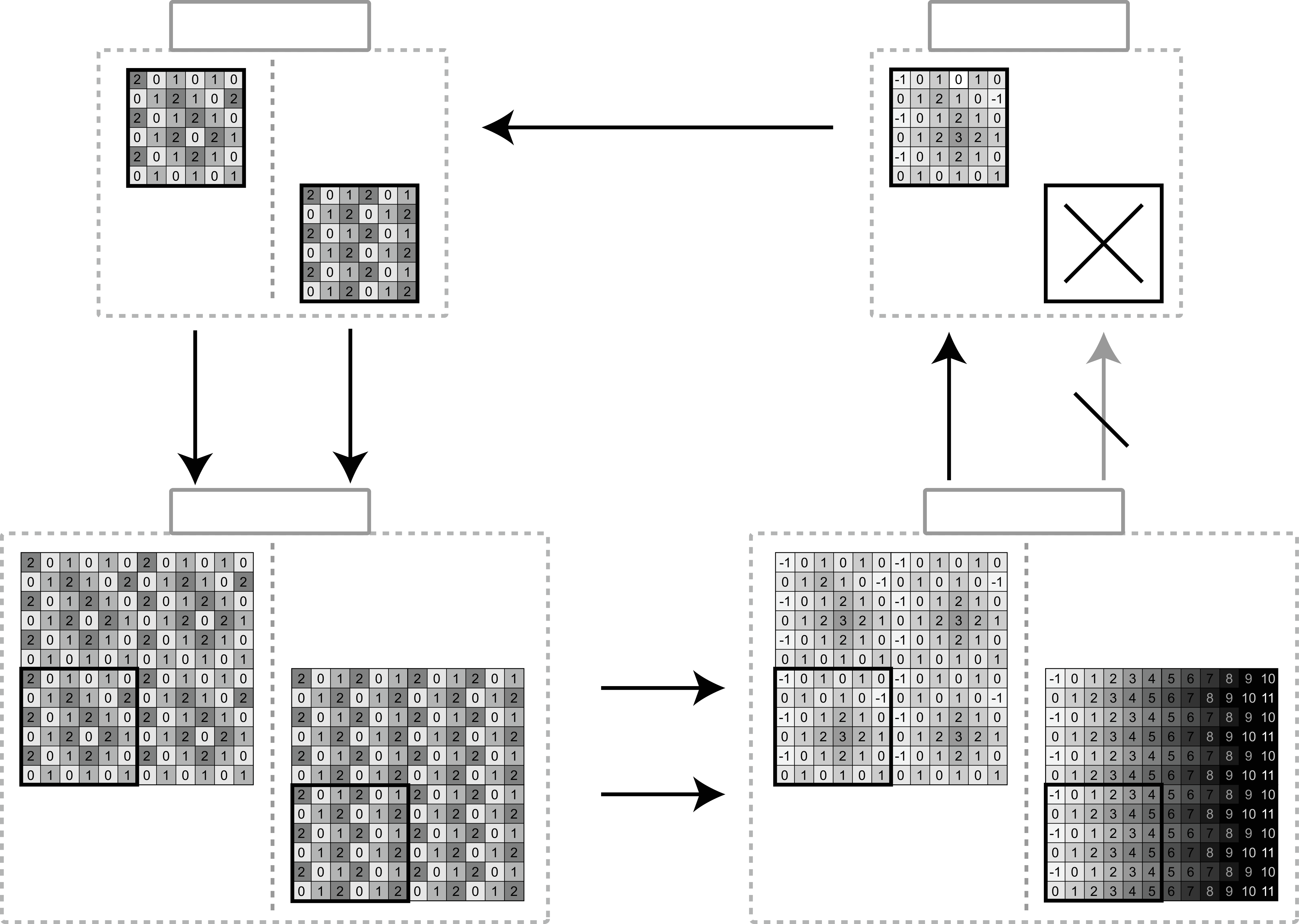}\\
\caption{The relations between $\col(\TT)$ and $\hm(\TT)$ through
periodic colorings and quasi-periodic HHFs on $\ZZ$. Notice that for
$\PC$ and $\QP$ only a small region of the infinite lattice is
illustrated. All functions are normalized at $\zero$, at the lower
left corner of the displayed region. The illustrations depict the
case $n=6$, $d=2$.} \label{fig: QP COL}
\end{figure}

\subsection{Most elements of $\QP$ are in $\QP_{\zero}$}

The following theorem states that most elements of $\QP$ have slope
$\zero$. This is equivalent to stating that most elements of
$\col(\TT)$ are in $\col_\zero(\TT)$.

\begin{theorem}\label{thm: QPm QP0 embedding}
There exist $d_0$ and $c>0$ such that in all dimensions $d\ge d_0$,
for every $m\in6\ZZ\setminus\{\zero\}$ we have
\begin{equation}\label{eq:QP_m_QP_0_rel}
\frac{|\QP_m|}{|\QP_\zero|}\le \exp(-c_d n^{d-1}),
\end{equation}
with $c_d = \frac{c}{d\log^2 d}$.
\end{theorem}

Our techniques also allow us to obtain a stronger version of
Theorem~\ref{thm: QPm QP0 embedding}. This version is not required
for the proof of Theorem~\ref{thm: almost bijection}, but is of
independent interest as it significantly improves the bound on the
size of $\QP_m$ when $m$ has a large coordinate. For clarity of
presentation, most of the paper is devoted to the proof of
Theorem~\ref{thm: QPm QP0 embedding} and the necessary modifications
required to obtain Theorem~\ref{thm: QP steep QP0 embedding} are
then detailed in Section~\ref{sec: Steep slopes}.

\begin{theorem}\label{thm: QP steep QP0 embedding}
There exist $d_0$ and $c>0$ such that in all dimensions $d\ge d_0$,
for every $m\in6\ZZ\setminus\{\zero\}$ we have
\begin{equation}\label{eq:QP_m_steep_QP_0_rel}
\frac{|\QP_m|}{|\QP_\zero|}\le \exp\left(-c_d n^{d-1}\cdot\max_{1\le
i\le d} |m_i|\right),
\end{equation}
with $c_d = \frac{c}{d\log^2 d}$.
\end{theorem}

In thermodynamic terms, a consequence of this theorem is that the
surface tension is non-differentiable at $\zero$ as a function of
the normalized slope $s=\frac{1}{n}\cdot m$. More precisely, for
each $s\in\RR$, the theorem implies that
\begin{equation*}
\limsup_{n\to\infty} \frac{1}{|\TT|} \log\left(\frac{|\col_{\lfloor
s\cdot n\rfloor}(\TT)|}{|\col(\TT)|}\right) = \limsup_{n\to\infty}
\frac{1}{|\TT|} \log\left(\frac{|\QP_{\lfloor s\cdot
n\rfloor}|}{|\QP|}\right)\le -c_d \max_{1\le i\le d} |s_i|
\end{equation*}
while
\begin{equation*}
\lim_{n\to\infty} \frac{1}{|\TT|}
\log\left(\frac{|\col_{\zero}(\TT)|}{|\col(\TT)|}\right) =
\lim_{n\to\infty} \frac{1}{|\TT|}
\log\left(\frac{|\QP_{\zero}|}{|\QP|}\right)=0.
\end{equation*}

Given \eqref{eq: bound on slopes}, we observe that Theorem~\ref{thm:
QPm QP0 embedding} and Theorem~\ref{thm: QP steep QP0 embedding} are
trivial for $n\le 4$, as in those cases $\QP_m$ is empty for
$m\neq\zero$. Thus, we shall assume $n\ge 6$ in the proofs of these
theorems.

Theorem~\ref{thm: almost bijection} is an immediate consequence of
(and is, in fact, equivalent to) Theorem~\ref{thm: QPm QP0
embedding}.

\begin{proof}[Proof of Theorem~\ref{thm: almost bijection} from Theorem~\ref{thm: QPm QP0 embedding}]
By symmetry, it is enough to prove Theorem~\ref{thm: almost
bijection} for colorings normalized at $\zero$. That is, to
establish that for sufficiently large $d$, if $f$ is uniformly
sampled from $\col(\TT)$ then
\begin{equation}\label{eq:almost_bijection_reformulation}
  \P\left(f\text{ is not the modulo $3$ of some }h\in\hm(\TT)\right)\le \exp\left(-\frac{c}{d\log^2 d}
n^{d-1}\right).
\end{equation}

Suppose then that $f$ is uniformly sampled from $\col(\TT)$. By
Proposition~\ref{prop: QP2}, \eqref{eq: bound on slopes},
\eqref{eq:col_m_def} and \eqref{eq:col_0_characterization},
\begin{multline*}
\P\left(f\text{ is not the modulo $3$ of some }
h\in\hm(\TT)\right)=\frac{\Big|\bigcup_{m\in (6\ZZ\cap
[-n,n]^d)\setminus\{\zero\}}
\col_m(\TT)\Big|}{|\col(\TT)|} =\\
= \frac{\Big|\bigcup_{m\in (6\ZZ\cap [-n,n]^d)\setminus\{\zero\}}
\QP_m\Big|}{|\QP|}\le (2n+1)^d \max_{m\in 6\ZZ\setminus \{\zero\}}
\frac{|\QP_m|}{|\QP_\zero|}.
\end{multline*}
Thus \eqref{eq:almost_bijection_reformulation} follows from
Theorem~\ref{thm: QPm QP0 embedding}.
\end{proof}

\subsection{Proof overview}\label{subs: Prem: Proof overview}

Most of the remainder of the paper is dedicated to proving
Theorem~\ref{thm: QPm QP0 embedding}. Our proof can be divided into
two parts. First we construct a set of one-to-one mappings, $\Psi_m
: \QP_m \to \QP_\zero$ for $m\in6\ZZ\setminus\{\zero\}$. We then apply
results from \cite{PHom} to show that the image of $\QP_m$ under
$\Psi_m$ is relatively small. Theorem~\ref{thm: QPm QP0 embedding}
follows. In this section we present for the reader a rough sketch of
the idea behind the construction of $\Psi_m$.

Let us first explain (a minor variant of) the construction of
$\Psi_m$ in dimension $d=1$, where it is rather simple. Suppose that
$h$ is a $1$-dimensional quasi-periodic HHF with slope $6\cdot
\ell>0$ (the case that the slope is negative is treated
analogously). One can look for the minimal $w\ge 0$ such that
$h(w)=2$ and for the maximal $u\le 0$ such that $h(u)=-3 \ell+2$.
Since $h$ has slope $6\ell$ it follows that $w-u<n$. Thus, we may
partition $\mathbb{Z}$ to segments of the form $(u+in, w+in]$ and
$(w+in, u+(i+1)n]$, $i\in\mathbb{Z}$. We may then define, for $v\in
\mathbb{Z}$,
$$\Psi_{6\ell}(h)(v)=\begin{cases}
h(v) - 6 i \ell & \text{$u+in \le v \le w+in$ for some $i\in \mathbb{Z}$} \\
4 - h(v) - 6 i \ell & \text{$w+in \le v \le u+(i+1)n$ for some $i\in
\mathbb{Z}$}.
\end{cases}$$
An example is shown in Figure~\ref{fig: 1 dim examp}.

It is not difficult to check that $\Psi_{6\ell}(h)$ is still an HHF,
noting that the action of $\Psi_{6\ell}$ can be seen as reversing
the gradient of $h$ between $w$ and $u+n$ and each of their
translations by multiples of $n$. Moreover, the resulting HHF will
be periodic in the sense that $\Psi_{6\ell}(h)(v+n) =
\Psi_{6\ell}(h)(v)$ for all $v\in \Z$. To see that $\Psi_{6\ell}$ is
one-to-one, one may check that $w$ is the minimal in $\mathbb{Z}_+$
satisfying $\Psi_{6\ell}(h)(w)=2$ and $u$ is the maximal in
$\mathbb{Z}_-$ satisfying $\Psi_{6\ell}(h)(u)=-3 \ell+2$. Given
$\ell$, one can thereby recover $u$ and $w$ from $\Psi_{6\ell}(h)$
and use them to recover $h$.

\begin{figure}[htb!]
\centering%
    \rput(12.4,4.3){$\Psi_6(h)$}
    \rput(3.5,4.3){$h$}
    \rput(1.1,0.8){$u$}
    \rput(2.07,0.8){$w$}
    \rput(10,0.8){$u$}
    \rput(10.93,0.8){$w$}
\includegraphics[scale=0.25]{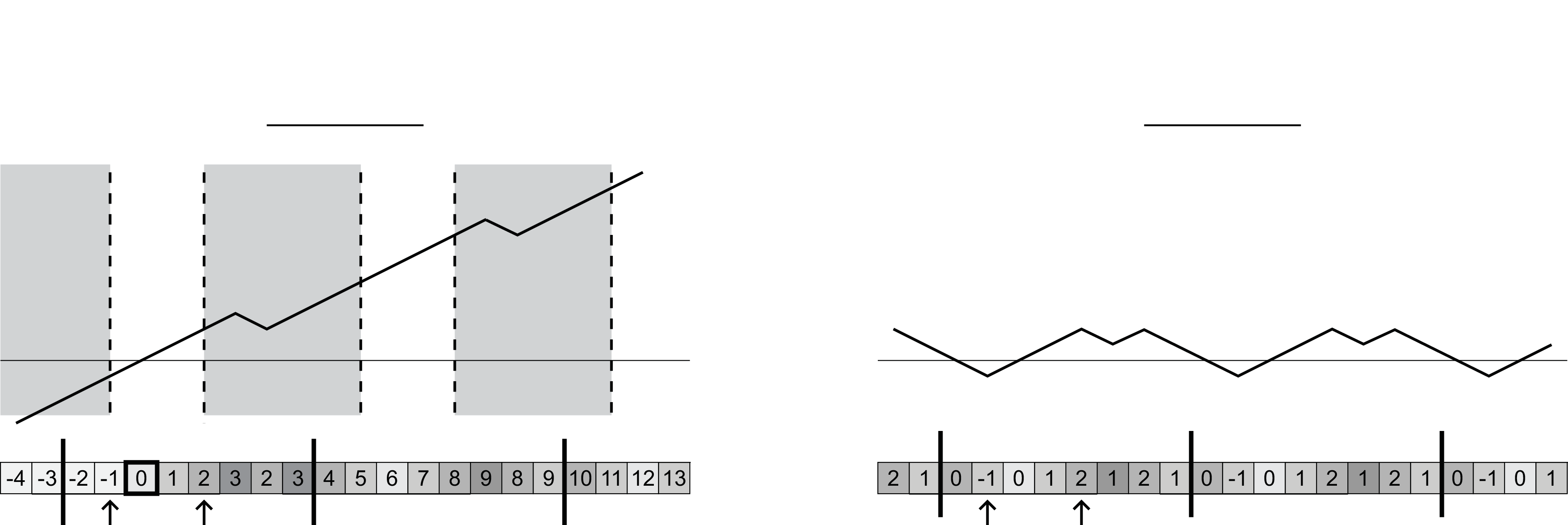}\\
\caption{On the left - an example of a one-dimensional quasi
periodic HHF with $n=8$ and slope $6$. The gray regions are the
regions where $\Psi_6$ reverses the gradient of the function. On the
right - the image of the same HHF through $\Psi_6$.} \label{fig: 1
dim examp}
\end{figure}
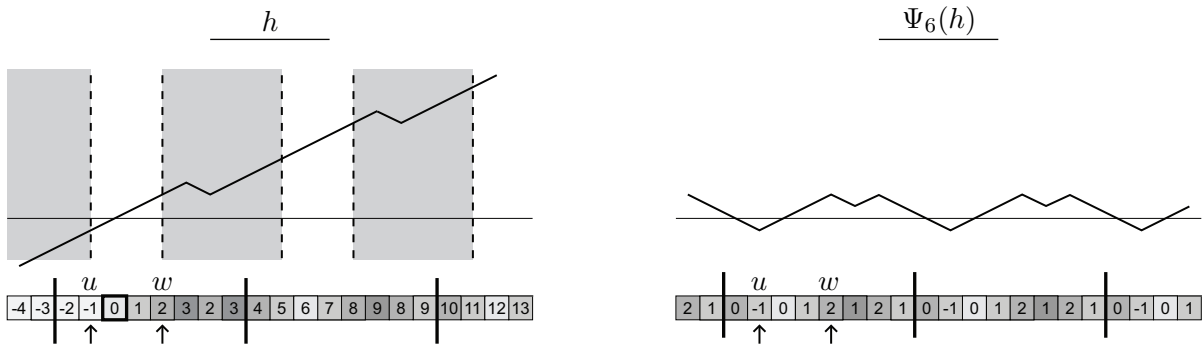

Generalizing this technique to higher dimensions is not immediate.
The general idea is to use the given HHF $h$ to carefully define two
sets $U,W\subseteq\ZZ$ and a vector $\Delta\in n\ZZ$ suitable for
our purposes. The set $U$ is the analog of the interval $(-\infty,
u]$ and the set $W$ is the analog of the interval $(-\infty, w]$.
Among the properties which these sets satisfy is the fact that if we
define $U_i:=U+i\Delta$ and $W_i:=W + i\Delta$ then the sets
$(W_i\setminus U_i)$ and $(U_{i+1}\setminus W_i)$ form a partition of
$\Z^d$. We then define $\Psi_m$, analogously to the above
one-dimensional case, by reversing the gradient of $h$ in the
regions $(U_{i+1}\setminus W_i)$, see \eqref{eq:Psi_m_def}. The main
difficulty is to find such sets $W,U$, and vector $\Delta$, for
which this operation yields a \emph{periodic} HHF, and is moreover invertible given $m$.

To show that that the size of the image of $\Psi_m$ is small compared to
$|\QP_\zero|$, we find an additional set $V$, sandwiched between $U$ and $W$ such that $\partial V$
is a level set of both $h$ and $\Psi_m(h)$.
We recall that $\pi$ defines a bijection between
\emph{periodic} HHFs on $\ZZ$ and HHFs
on $\TT$, and show that $\pi \circ \Psi_m(h)$
has a level set which contains $\pi(\partial V)$ ($\pi$ extends naturally to a mapping of the edges of $\ZZ$ to the edges of $\TT$).
After proving that $|\pi(\partial V)|\ge n^{d-1}$, we use a result from \cite{PHom}
to show that the probability that an HHF on $\TT$ contains such a long level
set is exponentially small in $n^{d-1}$. It follows that $|\QP_m|/|\QP_\zero|$ is tiny for all $m\neq \zero$.

The sets $U,V,W$ which we define are closely related to the level sets
of the function $h$ in the sense that $h$ is constant on $\intb U,\intb V,
\intb W, \extb U,\extb V$ and $\extb W$. In addition, they satisfy special
topological properties. The boundaries $\partial U$, $\partial V$ and $\partial
W$, regarded as a collection of plaquettes in $\mathbb{R}^d$, are
analogs of continuous hypersurfaces. Furthermore, the projection of
these boundaries to the torus are analogs of hypersurfaces whose
removal does not disconnect the torus.

The existence of sets $U,V,W$ satisfying all the required properties
is far from obvious. The intuition for it comes from algebraic
topology, specifically de Rham cohomology theory, and some of the
connections are explained in the next section. However, our proof
proceeds by developing the theory fully in the discrete setup. This
is achieved in sections~\ref{sec: structure of sets} and~\ref{sec:
QP properties}. This theory is then applied in Section~\ref{sec:
proof of 2.1} to define $\Psi_m$ and prove that it satisfies the
required properties.

To get a feeling of why the sets $U$ and $W$ exist, it may help to
think first of continuous linear functions on $\mathbb{R}^d$. A
multidimensional linear function is always simply a projection on
its gradient vector. Such a linear function could be made periodic
by periodically reversing its gradient between two hyperplanes which
are perpendicular to the gradient vector. These hyperplanes are the
analogs of $\partial W$ and $\partial U$. This case is therefore
very similar to the one-dimensional case. Algebraic topology tells
us that every continuous function is a deformation of a linear
function. Thus, a guiding intuition may be that for more general
functions, the above hyperplanes are deformed into some
hypersurfaces, and hence should still exist.

\subsection{Relation with topology}\label{subs: topology}

The proof of Theorem~\ref{thm: QPm QP0 embedding} is motivated by
ideas from algebraic topology. One element of the proof that might
puzzle a reader who lacks topological background is our ability to
find a domain, bounded by two hypersurfaces, such that reversing the
gradient in translated copies of this domain suffices to make our
HHF periodic. We dedicate this short section to highlight some of
the analogies between concepts of the proof and their continuous
topological counterparts and shed some light on this particular
point.

We begin with a brief review of concepts from de Rham cohomology
theory. A $0$-form on a manifold is simply a smooth function. A
$1$-form is a differential form which can be integrated against
paths. On Riemannian manifolds a $1$-form can be identified with a
vector field through the Riemannian metric. A $1$-form is called
\emph{closed} if it satisfies that its integral over contractible
loops is $0$. The gradient of a $0$-form is always a closed
$1$-form, and, locally, the converse is also true. Globally,
however, on non-contractible manifolds such as the torus, there are
many closed $1$-forms which are not the gradient of any $0$-form.
The group of closed $1$-forms modulo the gradients of the $0$-forms
is called the first de Rham cohomology group of the manifold.

In the context of our work, $0$-forms correspond to HHFs on the
torus. Closed $1$-forms correspond to proper $3$-colorings of the
torus, in the sense that, locally, they describe the discrete
gradient of an HHF. In the continuous torus every closed $1$-form is
locally the gradient of a $0$-form. Similarly, in the discrete
torus, every 3-coloring is locally the gradient of an HHF. However,
the local information does not always add up to form the global
structure of an HHF.

Algebraic topology tells us that the first de Rham cohomology
measures this global obstruction, in the sense that a $1$-form
corresponds to the zero class of the cohomology group if and only if
it is globally the gradient of a $0$-form. The first de Rham
cohomology of the $d$-dimensional torus is $\RR$. The class of any
given $1$-form can be identified by the integral of the form
over a loop in each of the standard basis directions. In the
terminology of this paper, this vector of integrals is called the
\emph{slope} of the form. Another way to represent the slope of a
$1$-form is to look at its pullback to what is called
the \emph{universal cover} of our
space. In the case of the torus we look at quasi-periodic functions
over $\RR$. Taking this point of view, the slope is the vector of
differences between the quasi-periodic function at standard basis
points and at $0$.

Poincar\'{e} duality identifies $H^1$, the first cohomology group of
the torus, with $H_{d-1}$, the $(d-1)$-th homology group of the
torus, which corresponds, if the slope consists of integers, to a
class of hypersurfaces of codimension~1. The duality further tells
us that for every nice enough $1$-form in a class of $H^1$, there
exist hypersurfaces in the dual class in $H_{d-1}$, orthogonal to
the gradient of the form and with the following property. Cutting
the torus along such a hypersurface leaves the torus connected, but
nullifies the cohomology class, i.e., on the cut torus the $1$-form
becomes the gradient of a $0$-form.

Much of the above description carries over to the discrete case.
Here too, we match proper 3-colorings with quasi-periodic HHFs, and
classify them according to their slope. We find ``level sets'',
corresponding to the above hypersurfaces, along which one may cut
the torus, that is, remove the corresponding edges, to make the
coloring the gradient of an HHF. We consider two such level sets
with a specific height difference. Deleting the edges of these level
sets splits the torus into two connected components such that on
each component, the coloring is the gradient of an HHF. Since the
height of the HHF is constant along each boundary of the cut torus
(as we have cut along level sets), we may reverse the gradient of
the coloring on one of the connected components of the cut torus to
obtain a coloring which is globally the gradient of an HHF (here,
our specific choice of the height difference of the level sets
enters). This illustrates the operation of $\Psi_m$. In practice, we
transfer most of the topological part of the proof to statements
involving HHFs on $\ZZ$, the universal cover of the torus. This
gives us more direct access to the level sets.

The main difficulties in our task are to define the level sets in
the discrete setup and to do so in such a way that would allow their
recovery after applying the gradient-reversal operation. As
mentioned above, the topological arguments are applicable to nice
functions, with nice level sets. In the discrete setting the level
sets are made out of plaquettes that can have complicated
intersections, of various dimensions. Proving that discrete level
sets still possess a nice structure requires the theory developed in
sections~\ref{sec: structure of sets} and~\ref{sec: QP properties}.

It remains unclear whether it is possible to avoid any combinatorial
argument in our proof, and use only topology. One can hope to
achieve this either by defining a clever discrete variant of the de
Rham cohomology, or by mapping the discrete problem to an analogous
question in $\RR$ with the hope of tackling it there. This, however,
is a path we did not pursue.

\section{Closed Hypersurfaces in $\ZZ$}\label{sec: structure of sets}

In this section we introduce a class of subsets of $\ZZ$ and discuss the
topological properties of its members. The definitions and results are
inspired by continuous topological analogs in $\RR$ but are given directly in
the discrete setting without requiring knowledge of the continuous notions
(see Section~\ref{subs: topology} for more on the connection). We make no
mention of neither colorings nor height functions here and thus the section
may be read using only the definitions regarding set operations in
Section~\ref{sec: prelim}. The tools developed here are applied to the study
of colorings and height functions in the following section, but we believe
that they are also of independent interest and may be of use for other
purposes.

The ultimate conclusion of the discussion here, Theorem~\ref{thm:
main trichotomy} below, is a certain trichotomy for systems of
translates in $\ZZ$. This trichotomy is later applied to level sets
of quasi-periodic HHFs.

We remind the reader that in the beginning of section~\ref{sec:
prelim} we fixed an even integer $n$ for the remainder of the paper.
This integer plays the role of the side length of the torus $\TT$ in
later sections. In this section $n$ will also play a role, though
the torus $\TT$ will not be explicitly mentioned. We point out that, unlike
the rest of the paper, the results and proofs presented in this
section remain valid regardless of whether $n$ is even or odd.

The structure of the section is as follows. In Section~\ref{subs:
trich: Topology} we present the fundamental properties of the sets
that we investigate and state our two main results, in the form of
certain trichotomies. Section~\ref{subs: Trich-corr} describes
corollaries of the main results, which will be of use in our
application. The proofs of the main results are given in
Sections~\ref{sec: proof of pair tric} and~\ref{sec: proof of trans
tric}.

\subsection{Topology of $\ZZ$}\label{subs: trich: Topology}
We begin by defining three properties of sets in $\ZZ$: \emph{\bicon
ness}, \emph{boundary disjointness}, and \emph{translation
respecting}. These are repeatedly used throughout the paper.

\heading{\Bicon ness} A set $U\subseteq\ZZ$ is called \emph{\bicon}
if $U\neq \emptyset$, $U\neq \ZZ$ and $U$ and $U^c$ are connected.

A useful property of \bicon{ }sets is that their boundaries are, in a sense,
connected. Namely,

\begin{propos}\label{prop: Timar II}
If $A$ is a {\bicon} set in $\ZZ$ then $\intb A \cup \extb A$,
$A^{++}\setminus A$ and $A\setminus A^{--}$ are all connected sets.
\end{propos}
We delay the proof of this proposition to Section~\ref{sec: proof of pair
tric}, as it requires the tools developed there.

In order to get a more intuitive grasp of the theorems and
definitions of this section the reader might find it useful to
regard $\ZZ$ as a lattice of $d$-dimensional cubes where the edges
between adjacent vertices represent plaquettes of codimension $1$.
Taking this continuous view, {\bicon} sets are analogous to
continuous sets whose boundary is a connected, oriented, closed
hypersurface. A set and its complement should be thought of as
defining opposite orientations on the same surface.

\heading{Boundary disjointness} Two sets $U_1,U_2\subseteq\ZZ$ are called
\emph{boundary disjoint}
    if
    \begin{enumerate}
    \item $\partial U_1\cap\partial U_2=\emptyset$,
    \item there is no $4$-cycle in $\ZZ$ whose vertices, in order, are
    $(v_{00},v_{01},v_{11},v_{10})$ such that $v_{00}\in U_1^c\cap U_2^c$, $v_{01}\in U_1^c\cap U_2$,
        $v_{11}\in U_1\cap U_2$ and $v_{10}\in U_1\cap U_2^c$.
    \end{enumerate}
Here and below, by a cycle in $\ZZ$ we mean a finite set $\{(u_1,
v_1),\ldots, (u_k,v_k)\}$ of distinct edges of $\ZZ$ satisfying that
$u_{i+1}=v_i$, $1\le i\le k-1$, and $u_1=v_k$. A \emph{$4$-cycle} is
a cycle with $k=4$, and by its vertices, in order, we mean $(u_1,
u_2, u_3, u_4)$.

Continuing the analogy with hypersurfaces, two sets are boundary
disjoint if their boundaries neither overlap nor intersect
transversally.

When both $U_1$ and $U_2$ are odd, as will always be the case from
Section~\ref{sec: QP properties} and on, the second condition for
boundary disjointness is trivially fulfilled, yielding the simpler
relation:
\begin{equation}\label{eq: odd disjoint}
\text{odd }U_1,U_2\text{ are boundary disjoint iff }\partial
U_1\cap\partial U_2=\emptyset.
\end{equation}

Observe that, by definition, boundary disjointness is preserved under taking
complements, i.e., if $U_1,U_2$ are boundary disjoint sets, then each of the pairs $\{U_1^c,U_2\}$,
$\{U_1,U_2^c\}$ and $\{U_1^c,U_2^c\}$ are also boundary disjoint.

The containment relations between two \bicon{ }boundary disjoint
sets are restricted by the following theorem.
\begin{theorem}\label{thm: pair_trich}\emph{(Pair trichotomy)}
If $U_1,U_2\subseteq\ZZ$ are \bicon{ }and boundary disjoint sets,
then exactly one of the following alternatives holds:
\begin{itemize}
\item $U_1 \cap U_2=\emptyset$,
\item $U_1^c \cap U_2^c=\emptyset$,
\item $U_1 \subsetneq U_2$ or $U_2 \subsetneq U_1$.
\end{itemize}
\end{theorem}
The proof of this theorem is postponed to Section~\ref{sec: proof of
pair tric}.

The following proposition relates containment of boundary disjoint
sets and their distance from a third set.
\begin{propos}\label{prop: dist prop}
If $U_1,U_2\subseteq\ZZ$ are non-empty, boundary disjoint sets
satisfying $U_1\subset U_2$ then for every non-empty set $V$
satisfying $V\cap U_2=\emptyset$ we have
$\dist(U_1,V)>\dist(U_2,V)$.
\end{propos}
\begin{proof}
Using boundary disjointness and \eqref{eq: Basic dist eq1}, we have
$U_1^+\subseteq U_2$. By \eqref{eq: Basic dist eq2} and \eqref{eq:
Basic dist eq3} we thus have $\dist(U_1,V)>\dist(U_2,V)$ as
required.
\end{proof}

\heading{Translation respecting sets} For a set $U\subseteq\ZZ$, we
define $T_U = T_U^n$, the set of \emph{translates} of $U$ by
multiples of $n$ in each of the coordinate directions, as
\begin{equation*}
T_U:=\{U + x :\ x\in n \ZZ\},
\end{equation*}
recalling that $U+v:=\{u+v\,:\,u\in U\}$. We note that it may well be the
case that different translations of $U$ yield the same set.

A set $U\subseteq\ZZ$ is called \emph{translation respecting} if $U$ is
\bicon{ }and every \emph{distinct} $U_1,U_2\in T_U$ are boundary disjoint.
Observe that, by definition, if $U$ is translation respecting, then so is
$U^c$.

Continuing the analogy with hypersurfaces, a translation respecting set is analogous to a hypersurface in $\RR$,
which satisfies that the projection of $\RR$ to the continuous torus maps its boundary to a closed hypersurface.

The main result of this section is that the trichotomy of Theorem~\ref{thm:
pair_trich} extends to translation respecting sets in the following strong
sense.

\begin{theorem}\emph{(Translation trichotomy)} \label{thm: main trichotomy}
If $U\subseteq\ZZ$ is translation respecting and $|T_U|>1$, then exactly one
of the following alternatives holds:
\begin{itemize}
\item \emph{[Type 1]} If $U_1,U_2\in T_U$ and $U_1\neq U_2$ then $U_1\cap
    U_2 = \emptyset$.
\item \emph{[Type -1]} If $U_1,U_2\in T_U$ and $U_1\neq U_2$ then
    $U_1^c\cap U_2^c = \emptyset$.
\item \emph{[Type 0]} If $U_1,U_2\in T_U$ then $U_1\subseteq U_2$ or
    $U_2\subseteq U_1$.
\end{itemize}
Moreover, if $U$ satisfies the Type $0$ alternative of the theorem, then
    there exists a unique order-preserving \emph{bijection} $o\colon
    T_U\to\Z$ such that $o(U)=0$. Here, order preserving means that $o(U_1)< o(U_2)$ if and only if $U_1\subsetneq U_2$. Furthermore, there
    exists a $\Delta\in n\ZZ$ such that $o^{-1}(i+1)=o^{-1}(i)+\Delta$
    for all $i\in\Z$. We call any such $\Delta$ a \emph{minimal translation} of $U$.
\end{theorem}

The proof of this theorem is postponed to Section \ref{sec: proof of
trans tric}.

We remark regarding the assumption that $|T_U|>1$ that while in dimension $d=2$ any \bicon{ }set $U$ has $|T_U|>1$ (recalling that a \bicon{ }set is assumed to be different from $\emptyset$ and $\ZZ$), there do exist \bicon{ }sets $U$ in dimensions $d\ge 3$ having $|T_U|=1$ (for instance,
the set of vertices in $\Z^d$ having at most one coordinate
which is not a multiple of $n$).

Theorem \ref{thm: main trichotomy} allows us to assign a type to
every translation respecting set $U$ satisfying $|T_U|>1$. For
$i\in\{-1,0,1\}$, we write $\Ty(U)=i$ if $U$ satisfies the Type $i$
alternative of the theorem. The case $|T_U|=1$ has little bearing on
our application. However, for completeness, we say in this case,
with a slight abuse of notation, that both $\Ty(U)=1$,
$\Ty(U)=-1$ and $\Ty(U)\neq0$ hold. An illustration of sets of the various types is
given in Figure~\ref{fig: Trich}.

\begin{figure}[htb!]
\centering%
    \rput(1.6,-0.15){$n$}
\includegraphics[scale=0.25]{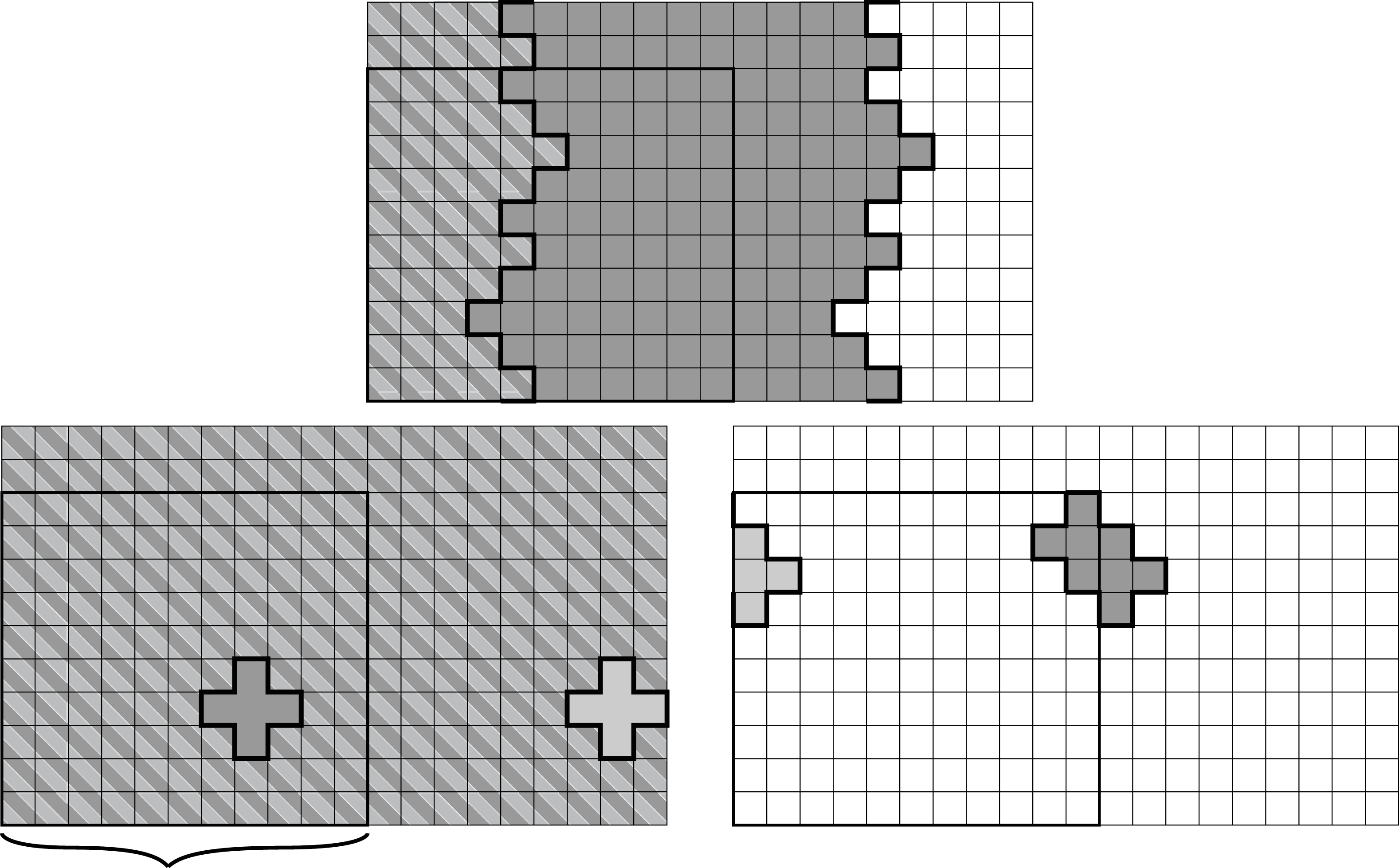}\\
\caption{Examples of translation respecting sets of the three types.
In each image a portion of the plane is depicted, on which a set $U$
and its translation $U+ne_1$ are emphasized in light gray and in
dark gray respectively. Vertices contained in both sets are striped.
In each image a different alternative of Theorem~\ref{thm: main
trichotomy} holds: At the top type $0$, at the bottom-left
 type $-1$ and at the bottom-right type $1$.}
\label{fig: Trich}
\end{figure}

\subsection{Corollaries of the trichotomy}\label{subs: Trich-corr}

In this section we state several useful corollaries of Theorem~\ref{thm: main
trichotomy}. The next proposition discusses how the type of translation
respecting sets is affected by taking complements.

\begin{propos}\label{prop: trans resp comp}
If $U$ is translation respecting of type $i$ then:
\begin{itemize}
\item $U^c$ is translation respecting of type $-i$.
\item If $U$ is of type 0 with minimal translation $\Delta$, then $-\Delta$ is a minimal translation of $U^c$.
\end{itemize}
\end{propos}
The proof of this proposition is straightforward from Theorem~\ref{thm: main
trichotomy}.

The following proposition investigates the possible containment
relations between translation respecting sets.

\begin{propos}\label{prop: Type rules}
Let $U,V$ be two translation respecting sets satisfying that
$|T_U|,|T_V|>1$ and $U\subseteq V$. Then $\Ty(U)\ge\Ty(V)$.
\end{propos}
\begin{proof}
Our goal is to show that
$(\Ty(U),\Ty(V))\notin\{(-1,0),(-1,1),(0,1)\}$. Equivalently, we
need to show that
\begin{align}
&\text{if $\Ty(V) = 1$ then $\Ty(U) = 1$},\label{eq:V_1_U_1}\\
&\text{if $\Ty(U) = -1$ then $\Ty(V) = -1$}.\label{eq:U_-1_V_-1}
\end{align}

Suppose first that $\Ty(V) = 1$. Let $\Delta\in n\ZZ$ be such that
$V+\Delta\neq V$ (which exists as $|T_V|>1$). As $\Ty(V)=1$, $V\cap
(V+\Delta) = \emptyset$. Thus, as $U\subseteq V$ and
$U+\Delta\subseteq V + \Delta$ we deduce that
\begin{equation}\label{eq:U_disjoint_U_Delta}
U\cap (U + \Delta) = \emptyset.
\end{equation}
In particular, $U\neq U + \Delta$ whence $U$ and $U+\Delta$ are
boundary disjoint (as $U$ is translation respecting) and the pair
trichotomy, Theorem~\ref{thm: pair_trich}, implies that
\begin{equation}\label{eq:U_comp_intersect_U_Delta_comp}
U^c\cap (U+\Delta)^c\neq \emptyset.
\end{equation}
The translation trichotomy, Theorem~\ref{thm: main trichotomy}, and
the relations \eqref{eq:U_disjoint_U_Delta} and
\eqref{eq:U_comp_intersect_U_Delta_comp} imply that $\Ty(U) = 1$,
establishing \eqref{eq:V_1_U_1}.

Now observe that $U^c,V^c$ are also translation respecting and
satisfy $V^c\subseteq U^c$. Thus, we may apply \eqref{eq:V_1_U_1}
with $(U,V)$ replaced by $(V^c, U^c)$ and deduce from
Proposition~\ref{prop: trans resp comp} that \eqref{eq:U_-1_V_-1}
holds.
\end{proof}

\heading{Translation respecting sets of type $0$} These have a
unique structure, as the following proposition indicates.

\begin{propos}\label{prop: type 0 union is all} If $U$ is
translation respecting of type $0$ then:
\begin{itemize}
\item $\displaystyle\bigcup_{V\in T_U} V=\ZZ$.
\item There exists $1\le i\le d$ such that for every $v\in\ZZ$,
    $\{v+k e_i\ :\ k\in\Z\}$ intersects both $U$ and $U^c$.
\item If $U+ne_1\neq U$ then for every $v\in\ZZ$,
    $\{v+k e_1\ :\ k\in\Z\}$ intersects both $U$ and $U^c$.
\end{itemize}
\end{propos}
\begin{proof}
Let $v\in \ZZ$ and let $\Delta$ be a minimal translation of $U$.
Observe that by definition, $U\subsetneq U+\Delta$, and $U,
U+\Delta$ are \bicon{ }and boundary disjoint. Applying
Proposition~\ref{prop: dist prop} we get $\dist(U + \Delta
,\{v\})\le\max(\dist(U,\{v\})-1,0)$. Iterating, we obtain that there
exists some $k$ such that $v\in U+k\Delta$. We deduce the first item
of the proposition.

The second item follows from the third by symmetry, as the fact that
$U+\Delta\neq U$ (using that $U$ is of type $0$) implies that there
exists some $1\le i\le d$ for which $U+ne_i\neq U$. We proceed to
prove the third item. Observe that by the last part of
Theorem~\ref{thm: main trichotomy}, there exists some $\ell\in
\Z\setminus\{0\}$ such that $U+ne_1=U+\ell\Delta$. It follows also
that $U^c+ne_1=U^c+\ell\Delta$. Notice that both $U$ and $U^c$ are
translation respecting of type $0$ with $-\Delta$ being a minimal
translation for $U^c$ (by Proposition~\ref{prop: trans resp comp}).
Thus, the first item of the proposition and the last part of
Theorem~\ref{thm: main trichotomy} show that for every $v\in\Z^d$
there exist $k_1,k_2\in\mathbb{Z}$ such that
\[v\in (U+k_1 \ell \Delta)\cap (U^c + k_2 \ell \Delta).\]
Equivalently $v-k_1ne_1 \in U$ while $v-k_2ne_1 \notin U$, as
required.
\end{proof}

Recall that $\pi$ is the projection of $\ZZ$ onto $\TT$. It naturally extends to a mapping of the edges of $\ZZ$ to the edges of $\TT$.
The projection of the boundary of translation respecting sets of type $0$ through $\pi$ is very long, as the following lemma shows.

\begin{lemma}\label{lem: type 0 steep has long boundary in increasing direction}
If $V$ is a translation respecting set of type $0$ with minimal translation $\Delta$ satisfying $V+\ell\Delta=V+n e_1$ then
$$|\{(w_0,w_1)\in\pi(\partial {V})\,:\, w_0-w_1=e_1\}|\ge \ell n^{d-1}.$$
\end{lemma}

\begin{proof}
The lemma holds trivially if $\ell=0$. Assume without loss of generality that $\ell>0$.
We write
\begin{equation*}
X:=\{x\in\ZZ\ :\ \forall j\in\{2,\dots,d\}\ \; 0\le x_j<n\}.
\end{equation*}
Observe that $\pi(\partial  V)=\pi(\partial V+k\Delta)$ for all $k\in\Z$.
Thus, to obtain the lemma it would suffice to show the following two claims:
\begin{equation}\label{eq: many level lines}
\text{$\{\pi(X\cap (\intb V+k\Delta))\}_{k\in\{0,\dots,\ell-1\}}$
are disjoint}.
\end{equation}
\begin{equation}\label{eq: level lines are long}
\text{For each $k\in\{0,\dots,\ell-1\}$ we have $|\pi(E_1(X)\cap
(\partial V+k\Delta))|\ge n^{d-1}$},
\end{equation}
where $E_1(X):=\{(x,x+e_1)\,:\,x\in X\}.$

We begin by showing \eqref{eq: many level lines}. Since $V$ is translation respecting,
$\{X\cap (\intb V+k\Delta)\}_{k\in\{0,\dots,\ell-1\}}$ are disjoint.
    Thus, to obtain \eqref{eq: many level lines}, all that remains is to
show that for all pairs of distinct $k_1,k_2\in\{0,\dots,\ell-1\}$, there are no two
elements $x_1\in \intb V+k_1\Delta$, $x_2\in \intb V+k_2\Delta$, such that
$x_1+ane_1=x_2$ for some $a\in\N$.
Indeed, in such a case, we would have $x_2=x_1+ane_1\in \intb V+k_1\Delta+ane_1=\intb
V+(k_1+a\ell)\Delta$, which would imply, by boundary disjointness, that
$V+k_2\Delta = V+(k_1+a\ell)\Delta$, and hence $k_2=k_1+a\ell$ which contradicts our assumption.
\eqref{eq: many level lines} follows.

To see \eqref{eq: level lines are long}, observe that by the third
item of Proposition~\ref{prop: type 0 union is all}, for every
$x\in\ZZ$, $k\in \{0,1,\dots,\ell-1\}$, there exists a
$q\in\Z$ such that
\begin{equation}\label{eq:t0 steep boundary intersection}
(x+q e_1,x+(q+1) e_1)\in \partial V+k\Delta.
\end{equation}
Using \eqref{eq:t0 steep boundary intersection} for all $x\in X$
satisfying that $x_1=0$, we obtain \eqref{eq: level lines are long}.
\end{proof}

\subsection{Proof of the pair trichotomy}\label{sec: proof of pair tric}
In this section we prove Proposition~\ref{prop: Timar II} and
Theorem~\ref{thm: pair_trich} using the approach of Tim\'ar in
\cite{Tim}. To do so, we make use of the well-known fact that
$4$-cycles span the cycles of $\ZZ$, i.e., every cycle $\sigma$ in
$\ZZ$ can be written as
\begin{equation}\label{eq: cycles decomposition}
\sigma = \displaystyle \sum_{c\in\C} c,
\end{equation}
where $\C$ is a set of $4$-cycles, and we interpret the sum as
meaning that an edge is in $\sigma$ if it appears in an odd number
of cycles in $\C$.

To aid our proof we introduce the following family of graphs.
\begin{defin}\label{def: G_u connected} Given $U\subseteq \ZZ$, a
set of vertices, we define a graph $G_U$ as follows. The vertices of
$G_U$ are the vertices of $\ZZ$. Two vertices $u,v$ are adjacent in
$G_U$ if there exist $e_u,e_v\in \partial U$ and a $4$-cycle $c$,
such that $u\in e_u, v\in e_v$, and $e_u,e_v\in c$.
\end{defin}
The following lemma connects this definition with \bicon{ }sets.
\begin{lemma}\label{lem: edge connectivity}
If $U\subset\ZZ$ is a \bicon{ }set of vertices, then $\intb U$ is connected
in $G_U$.
\end{lemma}
\begin{proof}
The proof is heavily based on ideas developed in \cite{Tim}. It
suffices to show that for any non-trivial partition $S_1, S_2$ of
$\intb U$ there exists an edge of $G_U$ connecting $S_1$ and $S_2$.
Here, a non-trivial partition means that $S_1, S_2\neq \emptyset$,
$S_1\cap S_2=\emptyset$ and $S_1\cup S_2 = \intb U$. Let $S_1, S_2$
be such a partition. We set
\begin{align*}
   E_1&:=\{e\in\partial U \ :\ e\cap S_1 \neq \emptyset\},\\
   E_2&:=\{e\in\partial U \ :\ e\cap S_2 \neq \emptyset\}.
\end{align*}
By the connectedness of $U$ and $U^c$ in $\ZZ$, there exists some
cycle $\sigma$ in $\ZZ$ which contains exactly one edge of $E_1$ and
one edge of $E_2$ (in fact, we can even pick those boundary edges
arbitrarily). As $4$-cycles span the cycles of $\ZZ$, we write
$\sigma$ as a sum of such cycles
\begin{equation}
\sigma = \displaystyle \sum_{c\in\C} c,
\end{equation}
as in \eqref{eq: cycles decomposition}. We notice that as $\sigma$
contains an odd number of $E_1$ edges (in fact, just one), there
must also be a $4$-cycle $c_0\in\C$ containing an odd number of
$E_1$ edges. However as every cycle contains an even number of edges
from the boundary $\partial U=E_1\uplus E_2$, $c_0$ must contain an
edge of $E_2$ as well. Thus $S_1$ and $S_2$ are connected by an edge
of $G_U$, concluding the proof.
\end{proof}

Lemma \ref{lem: edge connectivity} allows us to prove
Proposition~\ref{prop: Timar II} and Theorem~\ref{thm: pair_trich}.
In this proof we will make use of \cite[Theorem 4]{Tim}. For
convenience, we state a special case of this theorem in the context
of our work.

\begin{theorem*}[Tim\'ar]
For any {\bicon} $A\subsetneq \ZZ$, the set $$\{y\in A^c\ :\ y\text{
differs from some point in }A\text{ by $\pm1$ in each of exactly one
or two coordinates}\}$$ is connected in $\ZZ$.
\end{theorem*}

To see that this is a special case of \cite[Theorem 4]{Tim}, take
$G=\ZZ$, and let $G^+$ be $G$ with an edge between every two
vertices who differ by $\pm1$ on each of exactly one or two
coordinates. Also, take $C=A$, and let $x$ be some arbitrary point
in $A^c$.

\begin{proof}[Proof of Proposition~\ref{prop: Timar II}]
Let $A$ be a {\bicon} set in $\ZZ$. The first part of the
proposition is an immediate result of Lemma~\ref{lem: edge
connectivity}, as connectivity of $\intb A \cup \extb A$ in $\ZZ$ is
weaker than connectivity of $\intb A$ in $G_A$. The proof of the
second part uses the above stated version of \cite[Theorem 4]{Tim}.
By the theorem,
$$B:=\{y\in A^c\ :\ y\text{ differs from some point in }A\text{ by $\pm1$ in
each of exactly one or two coordinates}\}$$ is connected in $\ZZ$.
In addition $B$ satisfies that $B\subset A^{++}\setminus A$ and that
every vertex in $A^{++}\setminus A$ has a neighbor in $B$ (as
$A^+\setminus A\subset B$). We therefore have that $A^{++}\setminus
A$ is connected in $\ZZ$ as required. To get the third part of the
proposition, we recall that if $A$ is \bicon, then so is $A^c$, and
that $A\setminus A^{--}=(A^c)^{++}\setminus A^c$. We can therefore
derive the third part of the proposition by applying the second part
to $A^c$.
\end{proof}

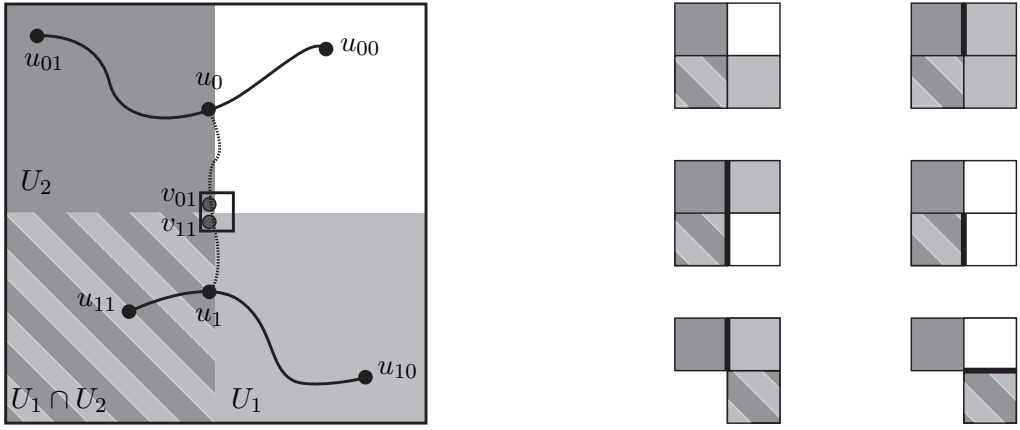
\begin{figure}[htb!]
\centering%
\includegraphics[scale=0.55]{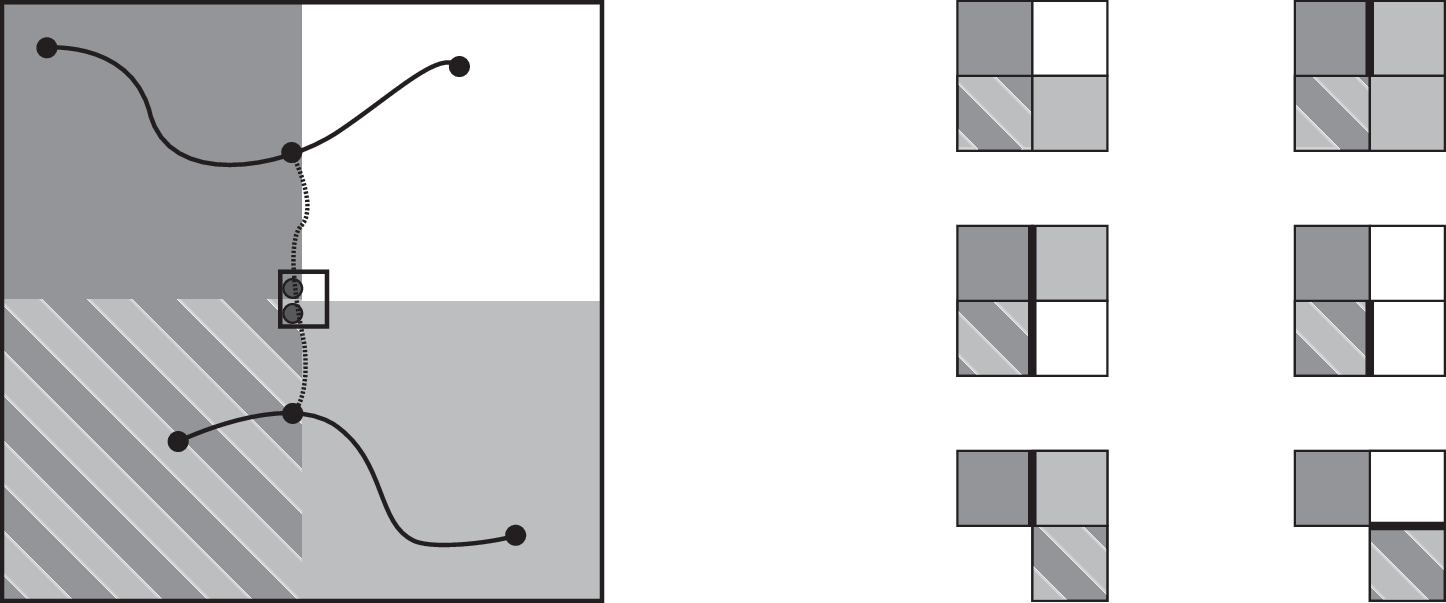}\\
    \rput(-6.3,3.8){$U_2$}
    \rput(-3.5,0.9){$U_1$}
    \rput(-6,0.9){$U_1\cap U_2$}
    \rput(-2,5.6){$u_{00}$}
    \rput(-1.5,1.3){$u_{10}$}
    \rput(-5.5,2.2){$u_{11}$}
    \rput(-6.2,5.4){$u_{01}$}
    \rput(-4,5.15){$u_{0}$}
    \rput(-4,2){$u_{1}$}
    \rput(-4.4,3.2){$v_{11}$}
    \rput(-4.4,3.6){$v_{01}$}
\caption{Illustration accompanying the proof of Theorem~\ref{thm:
pair_trich}. On the left - the roles of $u_{00}$, $u_{10}$, $u_{11}$
and $u_{01}$ are illustrated, as well as those of $u_0$, $u_1$,
$v_{01}$ and $v_{11}$. On the right - all the possible
configurations of the $4$-cycle $c$, up to rotation and reflection,
are illustrated. Observe that $c$ must contain a dark vertex
$v_{01}$ and a striped vertex $v_{11}$ and that each of these must
be adjacent to a vertex in $c$ which is neither dark nor striped. In
the top four configurations $v_{01}$ and $v_{11}$ are next to each
other while in the bottom two they are in opposite corners of the
cycle. When the boundary disjointness is ruled out due to the
existence of an edge violating $\partial U_1\cap\partial
U_2=\emptyset$, this edge is marked. When no edge is marked, the
alternative is ruled out due to the existence of a ``forbidden
cycle'' (as in the definition of boundary disjointness).}
\label{fig: pair proof}
\end{figure}

\begin{proof}[Proof of Theorem~\ref{thm: pair_trich}]
We accompany the proof with Figure~\ref{fig: pair proof}. Assume to
the contrary all the alternatives in the theorem do not hold. We can
therefore pick $u_{11} \in U_1 \cap U_2$, $u_{10} \in U_1 \cap
U_2^c$, ${u}_{01} \in U_1^c \cap U_2$ and ${u}_{00} \in U_1^c \cap
U_2^c$. As $U_1$ is connected, there exists a path inside $U_1$
between ${u}_{10}$ and ${u}_{11}$. This path must contain a vertex
$u_1\in U_1 \cap \intb U_2 $. Similarly there exists a path outside
$U_1$ between ${u}_{00}$ and ${u}_{01}$ which contains a vertex
$u_0\in U_1^c\cap \intb U_2$.

By Lemma~\ref{lem: edge connectivity}, $\intb U_2$ is connected in
$G_{U_2}$. In particular, if we partition $\intb U_2$ into $U_1\cap
\intb U_2$ and $U_1^c\cap \intb U_2$, we must have an edge in
$G_{U_2}$ crossing this partition. In other words, there exists a
$4$-cycle $c$ which contains two edges $e_0,e_1\in\partial U_2$, and
two vertices $v_{01}\in e_0$ and $v_{11}\in e_1$ such that $v_{01}
\in
 U_1^c \cap \intb U_2$ and $v_{11}\in U_1\cap \intb U_2$. A careful case study of all the possible
configurations of such a cycle (see Figure~\ref{fig: pair proof})
yields that its existence must contradict the boundary disjointness
for $U_1$ and $U_2$. We conclude that at least one of the
alternatives in the theorem must hold.

Next we show that \emph{exactly} one of the alternatives holds. The
third alternative cannot co-exist with either of the first two
alternatives as a co-connected set is non-empty and has non-empty
complement. For the first two alternatives to hold together it must
be the case that $U_1 = U_2^c$, contradicting the boundary
disjointness of $U_1$ and $U_2$. The theorem follows.
\end{proof}

\subsection{Proof of the translation trichotomy}\label{sec: proof of trans tric}
This section is dedicated to the proof of Theorem~\ref{thm: main trichotomy}.

We begin by showing the trichotomy itself. The pair intersection
trichotomy, Theorem~\ref{thm: pair_trich}, guarantees that every two
sets $U_1,U_2\in T_U$ satisfy one of the three alternatives of the
theorem. Thus it is sufficient to show that for any three distinct
sets $U_1,U_2,U_3\in T_U$, the same alternative holds for both pairs
$U_1,U_2$ and $U_1,U_3$. In particular, the theorem is immediate if
$|T_U|=2$. Fix distinct $U_1, U_2, U_3\in T_U$. We shall rule out
three cases.
\begin{enumerate}
\item Alternatives 0 and 1 cannot coexist. Let $\delta,\Delta\in n\ZZ$ be such
    that $U_2 = U_1+\delta$ and $U_3=U_1+\Delta$. Assume, WLOG, that
    $U_1\cap U_3=\emptyset$ and $U_1\subsetneq U_2$. As $U_1$ and $U_2$
    are boundary disjoint, by Proposition~\ref{prop: dist prop} we get
    that $\dist(U_1,U_3)>\dist(U_2,U_3)$. We note that,
    $U_1+\Delta\subseteq U_1+\Delta+\delta$, as $U_1\subseteq
    U_1+\delta$. We deduce, using \eqref{eq: Basic dist eq3}, that
    $\dist(U_1+\delta, U_1+\Delta) \ge \dist(U_1+\delta,
    U_1+\Delta+\delta)$. Putting all of this together, we get:
\begin{equation*}
\dist(U_1,U_1+\Delta) > \dist(U_1+\delta, U_1+\Delta) \ge
\dist(U_1+\delta, U_1+\Delta+\delta) = \dist(U_1,U_1+\Delta),
\end{equation*}
which is a contradiction.
\item Alternatives 0 and -1 cannot coexist. The argument follows similarly to
    the previous part by passing from $U_1,U_2,U_3$ to
    $U_1^c,U_2^c,U_3^c$.
\item Alternatives 1 and -1 cannot coexist. To see this, assume, WLOG, that
    $U_1\cap U_2=\emptyset$ and $U_1^c\cap U_3^c=\emptyset$. It follows
    that $U_1\cup U_3=\ZZ$ and hence $U_2\subseteq U_3$. A contradiction
    follows since alternatives 0 and 1 cannot coexist.\qedhere
\end{enumerate}
Next, we show the second part of the theorem, i.e., that if
$\Ty(U)=0$, then there exists a translation $\Delta \in n \ZZ$ and
an order-preserving bijection $o\colon T_U\to\Z$, such that
$o^{-1}(i+1)=o^{-1}(i)+\Delta$ for all $i\in\Z$. Assume $\Ty(U)=0$.
Define $o(U):=0$ and for any $V\in T_U$ let
\begin{equation*}
o(V):=\begin{cases} \big|\{W\in T_U\ :\ U\subsetneq W\subseteq
V\}\big|&
U\subseteq V\\
-\big|\{W\in T_U\ :\ V\subseteq W\subsetneq U\}\big|& V\subseteq U
\end{cases}.
\end{equation*}
To see that this is well defined, let us explain why $\{W\in T_U\ :\
U\subsetneq W\subseteq V\}$ is finite. A similar argument will show
that $\{W\in T_U\ :\ V\subseteq W\subsetneq U\}$ is finite. Since
$T_U$ is ordered by inclusion, applying Proposition~\ref{prop: dist
prop} to the complements of two distinct sets in $\{W\in T_U\ :\
U\subsetneq W\subseteq V\}$, taking the $V$ of the proposition to be
our $U$, shows that each set $W$ in $\{W\in T_U\ :\ V\subseteq
W\subsetneq U\}$ is uniquely characterized by $\dist(W^c,U)$. Since
$\dist(W^c,U)\le\dist(V^c,U)$ we conclude that $\{W\in T_U\ :\
U\subsetneq W\subseteq V\}$ is finite, as we wanted to show.

To show that $o$ is one-to-one, suppose $V_1,V_2\in T_U$ satisfy
$o(V_1)=o(V_2)$. Assume WLOG that $o(V_1) \ge 0$ and $V_1\subseteq V_2$. This
implies that
$$\{W\in T_U\ :\ U\subsetneq W\subseteq V_1\} \subseteq \{W\in T_U\ :\ U\subsetneq W\subseteq V_2\}.$$
However, as $o(V_1)=o(V_2)$, we get
$$\{W\in T_U\ :\ U\subsetneq W\subseteq V_1\} = \{W\in T_U\ :\ U\subsetneq W\subseteq V_2\}$$
and, in particular, $V_2\subseteq V_1$. Thus $V_1 = V_2$.

Finally, we show that there is a $\Delta\in n\ZZ$ such that
$o^{-1}(i+1)=o^{-1}(i)+\Delta$ for all $i\in\Z$. We begin by
observing that $o^{-1}(1)$ is nonempty. To see this recall that
$|T_U|>1$ and therefore $U \subsetneq U+z$ for some $z \in n\ZZ$.
This implies that $o(U+z)\ge1$ and therefore there must exist some
$\Delta\in n\ZZ$ such that $o(U+\Delta)=1$. Equivalently, there is
no $W\in T_U$ for which $U\subsetneq W\subsetneq U+\Delta$. Since
this situation is preserved under translations it follows that
$o^{-1}(i) = U+i\Delta$ for all $i\in\Z$. \qed

\section{Sublevel Sets of HHFs}\label{sec: QP properties}

In this section we establish the theoretical basis for dealing with
quasi-periodic HHFs. Much of the intuition behind the theorems of this
section stems from algebraic topology, viewing quasi-periodic HHFs as a
discrete analogue of co-cycles on the torus, and periodic HHFs as a discrete
analogue of co-boundaries. Nonetheless, we avoid making any direct reference
to topology, and restrict ourselves to purely combinatorial proofs. The results of this section are central to our construction in Section~\ref{sec: proof of 2.1} of the one-to-one mapping $\Psi_m : \QP_m \to \QP_\zero$ and the analysis of its properties.

We begin by introducing the notions of \emph{sublevel sets} and
\emph{sublevel components} of HHFs. Sublevel sets are discrete
counterparts to sublevel sets of continuous functions. A sublevel
component augments a sublevel set to a {\bicon} set.

Let $G$ be either $\ZZ$ or $\TT$. Let $k\in\Z$, $h\in \hm(G)$ and
let $u,v\in V(G)$ satisfy
\begin{equation}\label{eq:h_u_k_h_v}
   h(u)\le k<h(v).
\end{equation}
We define the \emph{$k$-sublevel set} of $u$,
\begin{equation*}
  \LL_h^{k+}(u)\text{ is the connected component of $u$ in $G\setminus
\{w\in V(G)\ :\ h(w)=k+1 \}$}.
\end{equation*}
While the sublevel set is itself connected, by definition, its
complement may be disconnected. We wish to isolate a single
connected component of the complement and do this by enlarging the
sublevel set. Precisely, we define the \emph{$k$-sublevel component
from $u$ to $v$},
\begin{equation*}
  \LC_h^{k+}(u,v)\text{ is the complement of the connected
component of $v$ in $G\setminus \LL_h^{k+}(u)$}.
\end{equation*}
Figure~\ref{fig: LC} illustrates a sublevel component and a sublevel
set in $\ZZ$. In our applications sublevel sets are mostly used as a
part of the definition of sublevel components, without a significant
role of their own. To simplify our notation we write $\LL_h^+(u)$
for $\LL_h^{h(u)+}(u)$ and $\LC_h^+(u,v)$ for $\LC_h^{h(u)+}(u,v)$.

\begin{figure}[htb!]
\centering%
\includegraphics[scale=0.35]{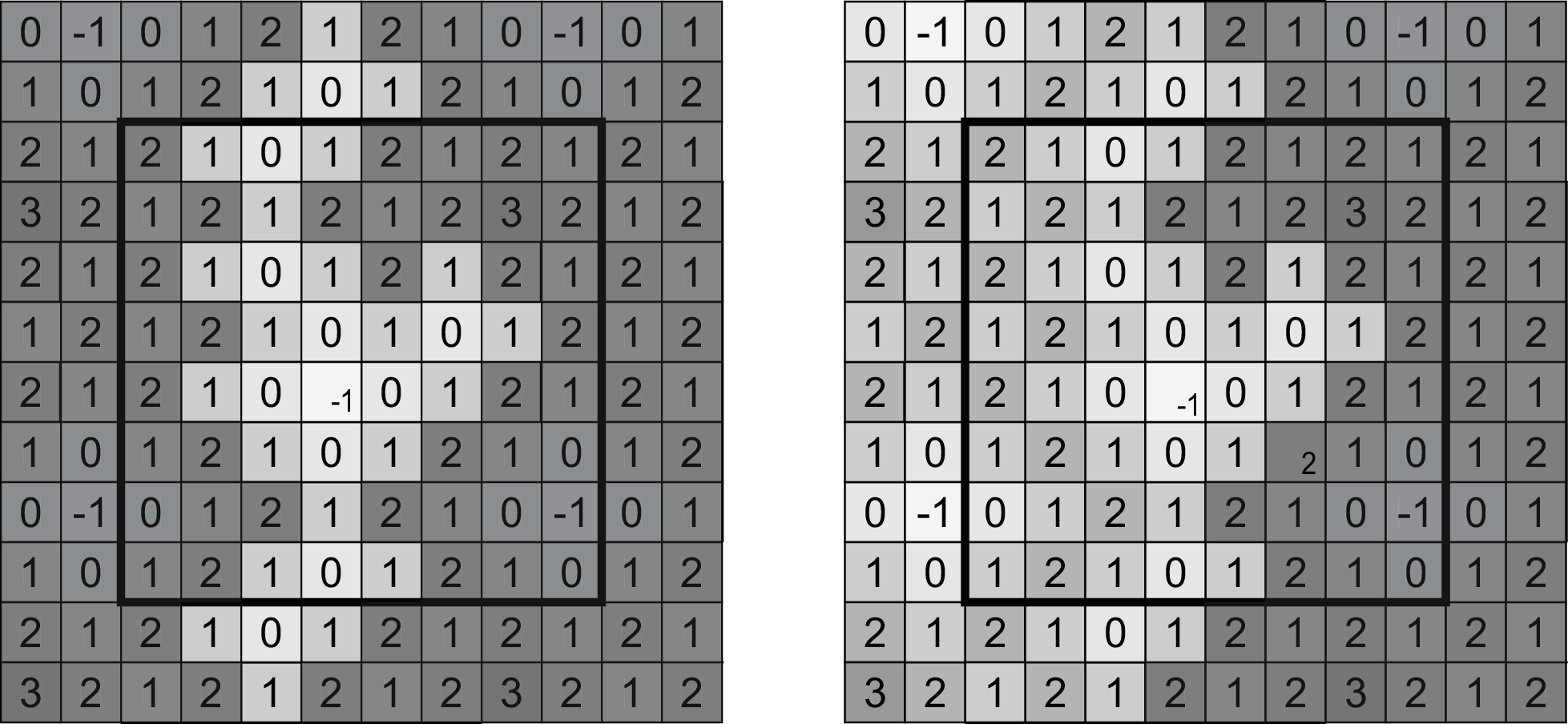}\\

 \rput(-3.4,3.1){$u$}
 \rput(2.83,3.1){$u$}
 \rput(3.7,2.65){$v$}
\caption{An illustration of sublevel components for a certain
periodic $h\in \hm(\ZZ)$, with respect to the two vertices
$u,v\in\ZZ$. On the left - a portion of $\LL_h^{1+}(u)$ is
highlighted. On the right - a portion of $\LC_h^{1+}(u,v)$. Observe
that $\LC_h^{1+}(u,v)$ is \bicon{ }while $\LL_h^{1+}(u)$ is not. }
\label{fig: LC}
\end{figure}

In the rest of the section we prove structure theorems for sublevel
components of HHFs, mainly on $\ZZ$. In Section~\ref{ss:sublevel
components and basic prop} we establish several basic properties of
sublevel components. In Section~\ref{subs: Level components of
height functions} we show that sublevel components on $\ZZ$ are
\bicon\ and boundary disjoint so that they satisfy the conditions of
the pair-trichotomy (Theorem~\ref{thm: pair_trich}). In
Section~\ref{subs: Level components determine height}, we give a
formula for computing the height difference between two vertices in
terms of the sublevel components separating them. In
Section~\ref{subs: Level components of QP functions} we show that
sublevel components of quasi-periodic HHFs are translation
respecting and hence satisfy the conditions of Theorem~\ref{thm:
main trichotomy} and can be assigned a type. We conclude there that
when $m\neq 0$, any HHF in $\QP_m$ has type-$0$ sublevel components.
In Section~\ref{subs: Decreasing Level components} we introduce
superlevel components and discuss their relationships with sublevel
components. Finally, Section~\ref{subs: Level components of two
height functions} gives a condition for two HHFs to share the same
sublevel component.

\subsection{Basic properties of sublevel components}\label{ss:sublevel components and basic prop}
Let $G$ be either $\ZZ$ or $\TT$. Let $h\in\hm(G)$ and suppose
$u,v\in G$ satisfy \eqref{eq:h_u_k_h_v}. Let
\begin{equation*}
  U:=\LC_h^{k+}(u,v).
\end{equation*}
The next proposition collects several basic properties of sublevel
components of $h$.
\begin{propos}\label{prop: basic LC prop} The
sublevel component $U$ satisfies:
\begin{enumerate}
\item $u\in U$ and $v\notin U$.
\item $h(x)=k$ for all $x\in\intb U$, and $h(x)=k+1$ for all $x\in\extb
    U$. In particular, $U$ is odd.
\item $U$ is \bicon.
\item $\intb U \subseteq \LL^{k+}_h(u)\subseteq U.$
\end{enumerate}
\end{propos}
All of these properties are straightforward from the definition and we omit
their proof.

In view of the second item of the proposition, we write, with a
slight abuse of notation, $h(\intb U)$ and $h(\extb U)$ for the
common height of all vertices in $\intb U$ and $\extb U$,
respectively.

In the next corollary, we give useful criteria for containment relations between a connected set
in $G$ and a sublevel component.
\begin{cor}\label{cor: outside LC crit}
The sublevel component $U$ satisfies:
\begin{itemize}
\item If $V\subseteq V(G)$ is connected and satisfies $v\in V$, $u\notin V$ and
    $h(w)>k$ for all $w\in\intb V$, then $V\subseteq U^c$.
\item If $V\subseteq V(G)$ is connected and satisfies $V\cap U\neq \emptyset$,
    $\extb U\subseteq V^c$, then $V\subseteq U$.
\end{itemize}
\end{cor}
\begin{proof}
To get the first item, observe that, by definition of the
$k$-sublevel set $\LC_h^{k+}(u)$ and the fact that an HHF changes by
one between neighbors, $\intb V\subset \LC_h^{k+}(u)^c$. As $\intb
V$ separates $V$ from $u$, we have $V\subset \LC_h^{k+}(u)^c$.
Together with the fact that $V$ is a connected set containing $v$,
the first item follows. The second item is straightforward and we
omit its proof.
\end{proof}

\subsection{Sublevel components on $\ZZ$}\label{subs: Level components of height
functions} Until the end of Section~\ref{sec: QP properties} we
discuss the structure of the set of sublevel components of a single
HHF on $\ZZ$. Throughout the rest of Section~\ref{sec: QP
properties}, we denote by $h$ an arbitrary function in $\hm(\ZZ)$.
In the beginning of Section~\ref{subs: Level components of QP
functions} we shall impose additional restrictions on $h$. Note that
dependence on $h$ will often be implicit in our notation.

\heading{Boundary disjointness} The following proposition implies
that sublevel components on $\ZZ$ satisfy the conditions of the pair
trichotomy, Theorem~\ref{thm: pair_trich}.
\begin{propos}\label{prop: LC boundaries are disjoint}
Distinct sublevel components of a function $h\in\hm(\ZZ)$ are boundary
disjoint.
\end{propos}
\begin{proof}
Consider $U:=\LC_h^{k+}(u,v)$ and $V:=\LC_h^{\ell+}(x,y)$, where
$k,\ell\in\Z$ and $u,v,x,y\in\ZZ$ satisfy $h(u)\le k<h(v)$ and
$h(x)\le \ell<h(y)$. Observe that if $k\neq\ell$, the proposition
holds trivially, by the second item of Proposition~\ref{prop: basic
LC prop} and \eqref{eq: odd disjoint}. We thus assume $k=\ell$.
Suppose $U$ and $V$ are not boundary disjoint and let us show that
this implies them being equal. From the second item of
Proposition~\ref{prop: basic LC prop}, and using \eqref{eq: odd
disjoint}, we get that there exists $e=(w_1,w_2)\in\partial
U\cap\partial V$, such that $w_1\in\intb U\cap\intb V$. By the
fourth item of Proposition~\ref{prop: basic LC prop} we have
$w_1\in\LL_h^{k+}(u)\cap\LL_h^{k+}(x)$ and thus $\LL_h^{k+}(u) =
\LL_h^{k+}(x)$, by the definition of sublevel sets. Since $w_2$ is
in the connected component of both $v$ and $y$ in $\ZZ \setminus
\LL_h^{k+}(u)$, then these connected components are equal and we get
$\LC_h^{k+}(u,v)=\LC_h^{k+}(x,y)$, as required.
\end{proof}

From Proposition~\ref{prop: LC boundaries are disjoint} we derive
the following corollary.

\begin{cor}\label{cor: LC neighbors}
Every edge $(u,v)\in \ZZ$ is contained in the boundary of a unique sublevel
component.
\end{cor}
\begin{proof}
Assume WLOG that $h(v)=h(u)+1$. By definition,
$(u,v)\in\partial\LC_h^{+}(u,v)$. By Proposition~\ref{prop: LC boundaries are
disjoint} no other sublevel component has $(u,v)$ in its edge boundary.
\end{proof}

The next proposition shows that in $\ZZ$, the fact that $A$ is a
sublevel component of $h$ depends only on a certain neighborhood of
the boundary of $A$.

\begin{propos}\label{prop: two HHFs, same LC 1}
Let $h_1,h_2\in\hm(\ZZ)$ be two HHFs. Let $A$ be a sublevel
component of $h_1$ and let $u\in\intb A$. Suppose there exists
$S\supseteq \intb A \cup \extb A$ satisfying that $h_1(w)=h_2(w)$
for all $w\in S$ and that $\LL_{h_1}^{+}(u)\cap S$ is a connected
set. Then $A$ is also a sublevel component of $h_2$.
\end{propos}
\begin{proof}
By our assumption $h_1(w)=h_2(w)$ for all $w\in S$, and by
definition $h_1(w)\le h_1(u)$ for all $w\in\LL_{h_1}^{+}(u)$. We get
that $h_2(w)\le h_2(u)$ for all $w\in \LL_{h_1}^{+}(u)\cap S$. Putting
this together with our assumptions that $u\in \intb A \subseteq S$,
and that $\LL_{h_1}^{+}(u)\cap S$ is connected, we get that
\begin{equation}\label{eq:set_and_S_containment}
\LL_{h_1}^{+}(u)\cap S \subseteq \LL_{h_2}^{+}(u),
\end{equation}
by the definition of sublevel sets.

Next, let $v\in \extb A$ be such that $u\sim v$. Observe that by
Corollary~\ref{cor: LC neighbors}, we have $A=\LC^{+}_{h_1}(u,v)$.
Let $U:=\LC^{+}_{h_2}(u,v)$. We shall show that $A=U$, establishing
the proposition. By the fourth item of Proposition~\ref{prop: basic
LC prop} we have that $\intb A\subseteq \LL_{h_1}^{+}(u)$ so that,
using \eqref{eq:set_and_S_containment} and our assumption that
$\intb A\subseteq S$, we get that $\intb A\subseteq
\LL_{h_2}^{+}(u)$. Thus, using the fourth item of
Proposition~\ref{prop: basic LC prop} again yields that
\begin{equation}\label{eq:A_h_2_relation}
\intb A\subseteq U.
\end{equation}

By our assumptions and Proposition~\ref{prop: basic LC prop}, $A^c$
is connected and satisfies $v\in A^c$, $u\notin A^c$ and $h_2(\extb
A)=h_1(\extb A)=h_2(u)+1$. Thus, the first item of Corollary~\ref{cor:
outside LC crit} implies that $A^c\subseteq U^c$. Thus, using
\eqref{eq:A_h_2_relation} and the fact that $U^c$ is connected by
Proposition~\ref{prop: basic LC prop}, shows that $A^c=U^c$. Hence
$U=A$ as we wanted to show.
\end{proof}

\subsection{Expressing height differences in terms of sublevel components}\label{subs: Level components determine height}

In this section we develop a formula expressing the difference
between the height assigned to a pair of vertices $u$ and $v$ in
terms of sublevel components. The formula is similar to the
Newton-Leibniz formula in that it expresses the global height
difference in terms of local increments. A visual depiction of this
similarity is given in Figure~\ref{fig: hightformula}.

\begin{figure}[htb!]
\centering%
\includegraphics[scale=0.3]{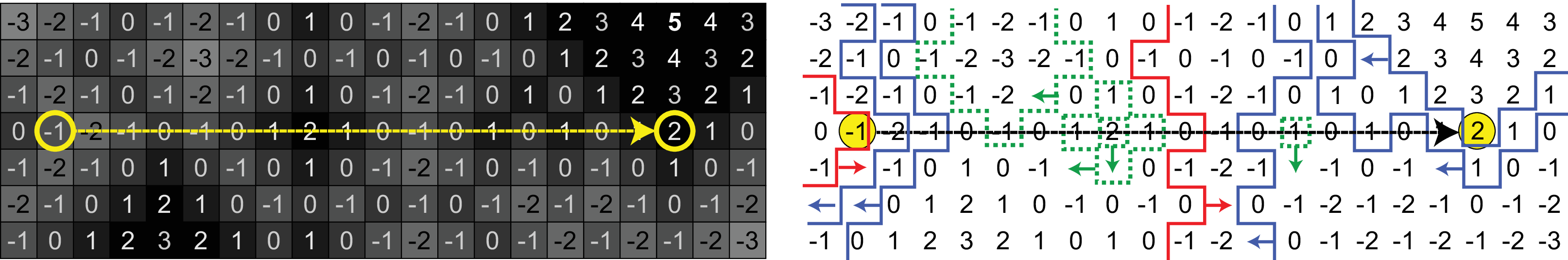}\\
\caption{Left: an HHF on $\ZZ$ and a path between two vertices $u$
and $v$ with height difference 3. Right: the same HHF with the
boundaries of all sublevel components (the arrow on each boundary
points towards its sublevel component) intersecting the path colored
according to which endpoints of the path they contain. Red - the two
sublevel components containing $v$ and not containing $u$. Blue -
the five sublevel components containing $u$ and not containing $v$.
Dashed green - sublevel components either containing both $u$ and
$v$ or neither. Proposition~\ref{prop: formula for hight diff.}
shows that the height difference between $u$ and $v$ equals the
number of blue boundaries minus the number of red boundaries.}
\label{fig: hightformula}
\end{figure}

Let $u,v\in\ZZ$. We define the set of sublevel components \emph{separating}
$u$ from $v$ by
\begin{equation}\label{eq:L_u_v_def}
\L_{(u,v)}:=\{ A \ :\ \exists u',v',k \ \ \text{\bf s.t.}\ h(u')\le
k<h(v')\text{ and }A=LC_h^{k+}(u',v')\text{ satisfies }u\in A,\,
v\notin A\}.
\end{equation}

\begin{propos}\label{prop: formula for hight diff.}
Let $u,v \in \ZZ$. $\L_{(u,v)}$ is finite and ordered by inclusion.
Furthermore, the following formula holds:
\[ h(v)-h(u)= \big |\L_{(u,v)}\big | - \big |\L_{(v,u)}\big |. \]
\end{propos}
\begin{proof}
Let $U,V$ be distinct elements of $\L_{(u,v)}$. We begin by showing
that $\L_{(u,v)}$ is ordered by inclusion. By Proposition~\ref{prop:
basic LC prop}, $U$ and $V$ are {\bicon } and by
Proposition~\ref{prop: LC boundaries are disjoint} they are boundary
disjoint. Thus, $U$ and $V$ satisfy the conditions of
Theorem~\ref{thm: pair_trich}. By the definition of $\L_{(u,v)}$, we
have $u\in U\cap V$ and $v\in U^c \cap V^c$. We deduce that either
$U\subseteq V$ or $V\subseteq U$. As containment relations are
transitive we deduce that $\L_{(u,v)}$ is ordered by inclusion.

To prove the remaining claims we use induction on the distance
between $u$ and $v$. Indeed, the case $u=v$ is trivial. Assume that the
proposition holds for every pair of vertices exactly at distance
$\rho$ and suppose $u,v$ satisfy $\dist(u,v)=\rho+1$. Next, let $w$
be a vertex satisfying $w\sim u$ and $\dist(w,v)=\rho$. By our
assumption
\[ h(v)-h(w)= \big |\L_{(w,v)}\big | - \big |\L_{(v,w)}\big |, \]
and thus
\begin{equation}\label{eq: v-u induction through w}
h(v)-h(u)= \big |\L_{(w,v)}\big | - \big |\L_{(v,w)}\big | + h(w) -
h(u).
\end{equation}

Suppose that $h(w)=h(u)+1$. Thus $U=\LC_h^+(u,w)$ is well defined. By
Corollary~\ref{cor: LC neighbors}, it is the only sublevel component
containing $u$ and not containing $w$, and there is no sublevel component
which contains $w$ and does not contain $u$. If $v\in U$, we get that
$\L_{(u,v)}=\L_{(w,v)}$ and that $\L_{(v,w)}=\L_{(v,u)}\uplus \{U\}$. If
$v\not\in U$, we get that $\L_{(u,v)}=\L_{(w,v)}\uplus \{U\}$ and that
$\L_{(v,u)}=\L_{(v,w)}$. In either case, by \eqref{eq: v-u induction through
w},
\[ h(v)-h(u)= \big|\L_{(u,v)}\big| - \big|\L_{(v,u)}\big|. \]
The case $h(u)=h(w)+1$ follows similar lines.
\end{proof}

\subsection{Sublevel components of quasi-periodic HHFs on $\ZZ$}\label{subs: Level components
of QP functions}

In this subsection we impose the further requirement that $h$ is quasi-periodic, that is, that $h\in \QP_m$ for some
$m\in\Z^d$. We show that sublevel components of such functions are translation
respecting and are thus classified into types according to the translation trichotomy, Theorem~\ref{thm: main trichotomy}.
We conclude that when $m\neq 0$, any such function has a type-$0$ sublevel component.

The first property we observe is that the set of sublevel components of $h$
is invariant under translations in $n\ZZ$.

\begin{propos}\label{prop: quasi per LC}
Let $k\in\Z$ and $u,v\in\ZZ$ be such that $h(u)\le k<h(v)$. For any $x\in
n\ZZ$ we have $\LC_h^{(k+\delta_x)+}(u+x,v+x)=\LC_h^{k+}(u,v)+x$ where
$\delta_x:=h(x)-h(\zero)$.
\end{propos}

The proposition follows directly from the definition of sublevel component
and quasi-periodic function and we omit its proof. A consequence of this proposition
is the following.

\begin{cor}\label{cor: LC trans.rep.}
Every sublevel component of $h$ is translation respecting.
\end{cor}

To see this recall that sublevel components are {\bicon} by the
third item of Proposition~\ref{prop: basic LC prop}, and apply Proposition~\ref{prop:
LC boundaries are disjoint} together with Proposition~\ref{prop: quasi per LC}.

Corollary~\ref{cor: LC trans.rep.} tells us that sublevel components
of quasi-periodic HHFs may be assigned a type, as in
Section~\ref{subs: trich: Topology}. We remark that it is possible
that a sublevel component $A$ will be invariant under all
translations in $n\ZZ$, in which case we follow the convention of
Section~\ref{sec: structure of sets} by assigning to it both type
$-1$ and type $1$. However, we note that this cannot happen when the
slope $m$ of the quasi-periodic function $h$ is non-zero, the case
of most interest to us, as follows from Proposition~\ref{prop: quasi
per LC} and the second item of Proposition~\ref{prop: basic LC
prop}.

The following corollary provides a formula for the height difference
between translates of a type~0 sublevel component.

\begin{cor}\label{cor: type 0 height formula}
Let $U$ be a type 0 sublevel component of $h$ with minimal
translation $\Delta$ and write $\delta:=h(\intb U+\Delta)-h(\intb
U)$. Then for any $z\in n\Z^d$ we have
$$h(\intb (U+z))-h(\intb U)= \delta \cdot o_{U}(U+z),$$ where $o_{U}$ is the order function on translates of
$U$, given by Theorem~\ref{thm: main trichotomy}.
\end{cor}
This is an immediate consequence of Theorem~\ref{thm: main trichotomy} and the fact that $h$ is quasi-periodic.

The next proposition establishes a duality between $\L_{(u,v)}$ and
$\L_{(v,u)}$ when $u-v\in n\ZZ$.

\begin{propos}\label{prop: L duality and structure}
Let $u,z\in \ZZ$ with $z\neq\zero$. If $A\in \L_{(u,u+n z)}$ has $\Ty(A)\neq
0$ then the type of $A$ is uniquely defined and $$A + \Ty(A)\cdot n z \in \L_{(u+n z ,u)}.$$
\end{propos}
\begin{proof}
Let $A\in \L_{(u, u+ nz)}$ for $u,z\in \ZZ$, and observe
that since $u\in  A$  but $u+ nz\notin A$ we have that $A$ is not invariant under translations
in $n\ZZ$ and hence $\Ty(A)$ is uniquely defined.
Suppose that $\Ty(A)\neq 0$, i.e.,
$\Ty(A)\in\{-1,1\}$. Recall that by definition, $u\in A$ and
$u+nz\notin A$. Since $A$ is a sublevel component then, by
Proposition~\ref{prop: quasi per LC}, $A \pm n z$ are also sublevel
components. Both are distinct from $A$ since $u+nz \in A+nz$ and $u
\notin A-nz$. If $\Ty(A)=1$, then by the trichotomy of
Theorem~\ref{thm: main trichotomy}, $u \in A$ implies that $u \notin
A+nz$. Similarly if $\Ty(A)=-1$, then by the same trichotomy $u+nz
\notin A$ implies $u+nz \in A-nz$. In either case the proposition
holds.
\end{proof}

An important corollary of the above proposition is the following:
\begin{cor}\label{cor: there is slope}
If $h\in \QP_m$ for $m=(m_1,\dots, m_d)$ satisfying $m_1>0$, then there exists a
sublevel component of type 0 which contains $\zero$ and does not
contain $ne_1$.
\end{cor}
\begin{proof}
Suppose to the contrary that every sublevel component in
$\L_{(\zero,n e_1)}$ is either of type $1$ or of type $-1$. By
Proposition~\ref{prop: L duality and structure} we get that
$|\L_{(\zero, n e_1)}|\le|\L_{( n e_1,\zero)}|$. By Proposition~\ref{prop:
formula for hight diff.} this implies $h(ne_1)\le h(\zero)$, in
contradiction to our premise. Here, we have also used the fact that
the type of a sublevel component is preserved under translation, thus
distinct sublevel components $A\in\L_{(\zero, ne_1)}$ are mapped to
distinct sublevel components in $\L_{(ne_1, \zero)}$ by the mapping
$A\mapsto A+\Ty(A)\cdot ne_1$.
\end{proof}

\subsection{Superlevel components and sublevel components of type $0$}\label{subs: Decreasing Level components}

In the construction of our embedding (in Section~\ref{sec: proof of
2.1}) we make use of superlevel components. These are counterparts
of sublevel components, in which the role of the sublevel set is
replaced by a superlevel set. The main reason that superlevel components are necessary
for our construction is that in order to guarantee invertibility of the mapping $\Psi_m$,
we wish to define it through an exploration process in a region which is left unchanged by the mapping. Exploration
in one direction is done by finding sublevel components while exploration in the other direction is done
through superlevel components.

While superlevel components could be defined in an
analogous way to that of sublevel components, as given at the
beginning of Section~\ref{sec: QP properties}, we rather define them
through a duality.
\begin{defin}\label{def: dec dual}
For any $u,v\in\ZZ$ and $k\in \Z$ satisfying $h(v)<k\le h(u)$, we
define
\begin{equation*}
  \LC_h^{k-}(u,v):=\LC_{-h}^{(-k)+}(u,v).
\end{equation*}
\end{defin}

This definition allows us to apply propositions dealing with
sublevel components to superlevel components. For instance,
combining the definition with Corollary~\ref{cor: LC trans.rep.} and
Theorem~\ref{thm: main trichotomy} we can assign a type to every
superlevel component. In addition, by Proposition~\ref{prop: basic
LC prop}, a superlevel component $U=\LC_h^{k-}(u,v)$ satisfies
$h(x)=k$ for all $x\in\intb U$, and $h(x)=k-1$ for all $x\in\extb
U$. However, to avoid confusion, we remark that the complement of a
superlevel component is not necessarily a sublevel component.

The next lemma shows that certain sublevel and superlevel components
which are ``sandwiched'' between two type 0 sublevel components must
also be of type 0.

\begin{lemma}\label{lem: increasing decreasing containment}
Let $U\subsetneq W$ be a pair of type $0$ sublevel components, such that
$h(\extb U)< h(\extb W)$ and let $u \in \intb U$, $w\in \intb W$ and
$k\in\Z$. Then:
\begin{itemize}
\item If $h(u)\le k< h(w)$ then $V_+:=\LC_h^{k+}(u,w)$ is a sublevel
    component of type $0$, satisfying $U\subseteq V_+ \subsetneq W$.
\item If $h(u)< k\le h(w)$ then $V_-:=(\LC_h^{k-}(w,u))^c$ satisfies that
    $(V_-)^c$ is a superlevel component of type $0$ and $U\subseteq V_-
    \subsetneq W$.
\end{itemize}
\end{lemma}
\begin{proof}
We start by proving the first item and let $V_+$ be as in the
lemma. We first show that
\begin{equation}\label{eq:U_LC_u_w}
U = \LC_{h}^+(u,w).
\end{equation}
By our assumptions, $U = \LC_{h}^{h(u)+}(u',v')$ for some $u',v'$.
By the fourth item of Proposition~\ref{prop: basic LC prop} we have
$\LC_{h}^{+}(u) = \LC_{h}^{h(u)+}(u')$. Next, $w\notin U$ since
$U\subsetneq W$ and $U$ and $W$ are boundary disjoint by
Proposition~\ref{prop: LC boundaries are disjoint}. Hence
\eqref{eq:U_LC_u_w} follows.

Now observe that by applying \eqref{eq:U_LC_u_w},
Proposition~\ref{prop: basic LC prop} and the first item of
Corollary \ref{cor: outside LC crit} to $U$ and $(V_+)^c$, we get
that $(V_+)^c\subseteq U^c$, i.e., $U\subseteq V_+$. Similarly, by
Proposition~\ref{prop: Timar II},
\begin{equation}\label{eq: w in win + wout}
\text{$\intb W \cup \extb W$ is a connected set containing $w$, whose vertices are of height greater
than $k$,}
\end{equation}
and hence $u\notin \intb W \cup \extb W$. Thus, applying \eqref{eq:
w in win + wout} and the first item of Corollary~\ref{cor: outside LC
crit}, we deduce that $(\intb W \cup \extb W)\subseteq V_+^c$. We
can now use the second item of Corollary~\ref{cor: outside LC crit}
to deduce that $V_+\subseteq W$. Consequently, $U\subseteq
V_+\subsetneq W$, where we have used also that $w\in W\setminus
V_+$. It remains to show that $V_+$ is of type $0$. All that we need
in order to draw this conclusion from Proposition~\ref{prop: Type
rules} is to show that $|T_W|,|T_{V_+}|, |T_U|>1$. To see this first
observe that since $\Ty(U)=\Ty(W)=0$, we have by definition
$|T_W|,|T_U|>1$. By Proposition~\ref{prop: type 0 union is all}
there exists some $\Delta\in n\ZZ$ satisfying $(U+\Delta)\cap
(V_+)^c \neq \emptyset$ while $U+\Delta\subseteq V_++\Delta$. We
deduce that $|T_{V_+}|>1$, so that $V_+$ is of type $0$.

The second item is proved similarly. Let $V_-$ be as in the
lemma. By the definition of superlevel components and
Proposition~\ref{prop: basic LC prop}, we have that $(V_-)^c$ is
connected, $u\notin (V_-)^c$, $w\in (V_-)^c$ and $h(\extb V_-)
> h(u)$. Applying \eqref{eq:U_LC_u_w} and the first item of Corollary~\ref{cor: outside LC
crit} to $(V_-)^c$ we deduce that $(V_-)^c\subseteq U^c$, i.e., $U
\subseteq V_-$.

Applying \eqref{eq: w in win + wout}, the definition of a superlevel
component, and the fourth item of Proposition~\ref{prop: basic LC prop} we
get that $\intb W \cup \extb W\subseteq (V_-)^c$, as it is contained in the
corresponding superlevel set. We deduce that $V_-$ is a connected set
satisfying $u\in V_-$ and $\extb W\subseteq (V_-)^c$. Therefore by the second
item of Corollary~\ref{cor: outside LC crit}, we have $V_-\subseteq W$.
Consequently, $U\subseteq V_-\subsetneq W$, where we have used also that
$w\in W\setminus V_-$. It remains to show that $V_-$ is of type $0$. All that
we need in order to draw this conclusion from Proposition~\ref{prop: Type
rules} is to show that $|T_{V_-}|>1$. This is done in exactly the same way as
in the proof of the first part of the lemma.
\end{proof}

\subsection{Locality property of sublevel components}\label{subs: Level components of two height
functions}
We conclude Section~\ref{sec: QP properties} with a useful criterion for applying Proposition~\ref{prop:
two HHFs, same LC 1}, to show that two HHFs in $\hm(\ZZ)$ share the same sublevel component.

\begin{propos}\label{prop: two HHFs, same LC 2}
Let $h_1,h_2\in\hm(\ZZ)$ be two HHFs and let $A$ be a sublevel
component of $h_1$. Suppose that
\begin{equation}\label{eq:h_1_h_2_equality_assumption}
h_1(w)=h_2(w)\text{ for all $w\in A^+\setminus B^-$},
\end{equation}
for some $B\subsetneq A$ which is either a sublevel component of
$h_1$ or the complement of a superlevel component of $h_1$. Then $A$
is also a sublevel component of $h_2$.
\end{propos}
\begin{proof}
Let $u\in \intb A$. Let $v\in \extb A$ be such that $u\sim v$. By
Corollary~\ref{cor: LC neighbors},
\begin{equation}\label{eq:A_h_1_component}
A = \LC_{h_1}^+(u,v).
\end{equation}
Let us show that $u\notin B$. Suppose to the contrary that $u\in B$.
Hence $u\in\intb B$ by our assumption that $B\subsetneq A$. Then, by
Proposition~\ref{prop: basic LC prop} and the definition of
superlevel component, $B^c$ is a connected set satisfying $v\in B^c$
and satisfying $h_1(\extb B)=h_1(u)+1>h_1(\intb A)$. Thus, by the
first item of Corollary~\ref{cor: outside LC crit}, we have that
$B^c\subseteq A^c$. However, this contradicts the fact that
$B\subsetneq A$.

We continue by considering separately two cases. First, assume that
\begin{equation}\label{eq:first_case_h_1_h_2}
  \text{either $h_1(\intb B)>h_1(u)$ or $h_1(\extb B)>h_1(u)$}.
\end{equation}
Since $u\notin B$, the definition of $\LL_{h_1}^{+}(u)$ and the
assumption~\eqref{eq:first_case_h_1_h_2} imply that
$\LL_{h_1}^{+}(u)\cap B = \emptyset$. Now, Proposition~\ref{prop:
basic LC prop} and \eqref{eq:A_h_1_component} imply that
$\LL_{h_1}^{+}(u)\subseteq A$. Thus, by
\eqref{eq:h_1_h_2_equality_assumption}, $h_1(w)=h_2(w)$ for all
$w\in(\LL_{h_1}^{+}(u))^+$. Hence the definition of sublevel set
yields that $\LL_{h_1}^{+}(u)=\LL_{h_2}^{+}(u)$, which, in turn,
implies that $\LL_{h_1}^{+}(u,v)=\LL_{h_2}^{+}(u,v)$. Thus,
recalling \eqref{eq:A_h_1_component}, $A$ is also a sublevel
component of $h_2$.

Second, let us assume that \eqref{eq:first_case_h_1_h_2} does not
hold. That is, that
\begin{equation}\label{eq:reverse_assumption_h_1_h_2}
  h_1(\intb B)\le h_1(u)\text{ and }h_1(\extb B)\le
h_1(u).
\end{equation}
Denote $S:=A^+\setminus B^-$. Recalling
\eqref{eq:h_1_h_2_equality_assumption} and observing that
$$A^+\setminus A^-=\intb A \cup \extb A \subseteq S,$$ all that we
need to show in order to apply Proposition~\ref{prop: two HHFs, same
LC 1} and derive the proposition, is that
\begin{equation}\label{eq: trunk LL conn.}
\LL_{h_1}^{+}(u)\cap S\text{ is connected.}
\end{equation}

Observe that, as $\LL_{h_1}^{+}(u)\subseteq\LC_{h_1}^{+}(u,v)= A$ by Proposition~\ref{prop: basic LC prop}, we have
\begin{equation*}
\LL_{h_1}^{+}(u)\cap S = \LL_{h_1}^{+}(u)\setminus B^-.
\end{equation*}
Let $H_0\uplus H_1$ be a non-trivial partition of
$\LL_{h_1}^{+}(u)\setminus B^-$. Assume for the sake of obtaining a
contradiction that there is no edge in $\ZZ$ connecting $H_0$ and
$H_1$ (that is an edge between a vertex in $H_0$ and a vertex in
$H_1$). Since $H_0\uplus H_1\uplus(\LL_{h_1}^{+}(u)\cap B^-) =
\LL_{h_1}^{+}(u)$, and $\LL_{h_1}^{+}(u)$ is a connected set, there
must be an edge of $\ZZ$ connecting $H_0$ and $\LL_{h_1}^{+}(u)\cap
B^-$, and an edge of $\ZZ$ connecting $H_1$ and
$\LL_{h_1}^{+}(u)\cap B^-$. The existence of these edges implies
that
\begin{equation}\label{eq: H nonempty}
\begin{split}
&(B^+\setminus B^-)\cap H_0\neq \emptyset\quad\text{and}\\
&(B^+\setminus B^-)\cap H_1\neq \emptyset.
\end{split}
\end{equation}
In particular,
\begin{equation}\label{eq:B_+_B_-_LL}
(B^+\setminus B^-)\cap( \LL_{h_1}^{+}(u)\setminus B^-)\neq\emptyset.
\end{equation}

By Proposition~\ref{prop: Timar II}, we have that
\begin{equation}\label{eq:B_boundaries are connected}
\text{$B^+\setminus B^-$ is a connected set.}
\end{equation}
Observe that $\LL_{h_1}^{+}(u)$ is a connected component of
$\{w\,:\,h_1(w)\le h_1(u)\}$, and, by
\eqref{eq:reverse_assumption_h_1_h_2}, $B^+\setminus
B^-\subseteq\{w\,:\,h_1(w)\le h_1(u)\}$. Thus, using
\eqref{eq:B_+_B_-_LL} and \eqref{eq:B_boundaries are connected} we
may deduce that
\begin{equation}\label{eq:B_boundary_is_in}
(B^+\setminus B^-) \subseteq \LL_{h_1}^{+}(u)\setminus B^- = H_0\cup H_1.
\end{equation}
Putting together \eqref{eq:B_boundary_is_in} and \eqref{eq: H
nonempty} we get that $H_0\uplus H_1$ induces a non-trivial
partition on $B^+\setminus B^-$ that is not crossed by any edge.
Since this contradicts \eqref{eq:B_boundaries are connected}, we
deduce that \eqref{eq: trunk LL conn.} holds.
\end{proof}

\section{Proof of the Embedding Theorem}\label{sec: proof of 2.1}

In this section we use the theory developed in the previous sections
to prove Theorem~\ref{thm: QPm QP0 embedding}. In Section~\ref{subs:
Flattening the slope} we present a one-to-one mapping from $\QP_m$,
the set of quasi-periodic HHFs with slope $m$, to $\QP_\zero$, the
set of periodic HHFs. In Section~\ref{sec: embedding proof} we prove
Theorem~\ref{thm: QPm QP0 embedding} using a probabilistic bound
taken from \cite{PHom} and an auxiliary lemma. This lemma, which
relates the boundaries of sublevel components in $\QP_{\zero}$ with
the boundaries of sublevel components of HHFs in $\hm(\TT)$, is then
proved in Section~\ref{subs: Projection of LC boundary}.

\subsection{Mapping quasi-periodic to periodic}\label{subs: Flattening the slope}

Throughout this section we fix some $m=(m_1,\dots,m_d)\in 6\ZZ$ such
that
\begin{equation*}
  \text{$m_1 > 0$ and $\QP_m\neq\emptyset$}.
\end{equation*}
We also fix $h\in\QP_m$. With the structural results of
Sections~\ref{sec: structure of sets} and~\ref{sec: QP properties}
in our toolkit, we are ready to construct $\Psi_m$, our one-to-one
mapping from $\QP_m$ into $\QP_{\zero}$. We start by defining three
sets, $U_0, V_0$ and $W_0$. The definition relies on the fact that
by Corollary~\ref{cor: LC trans.rep.}, sublevel and superlevel
components of $h$ are translation respecting and can therefore be
assigned a type by Theorem~\ref{thm: main trichotomy}. The first and
the third sets will be used to construct $\Psi_m$. The second set
will be used in Section~\ref{sec: embedding proof} to show that the
image of $\Psi_m$ is small. Proposition~\ref{prop: U and V exist}
below shows that the three sets are well defined.

In the following definition, and throughout the entire section, we
say that a set $S\subset\ZZ$ is the \emph{minimal set} with a given
property, if $S$ is contained in every other set with that property.

\begin{itemize}
\item $W_0=W_0(h)$ is the minimal type $0$ \emph{sublevel} component
    satisfying
    \begin{equation}\label{eq:W_0_prop}
      \zero \in W_0\text{ and }n e_1 \notin W_0.
    \end{equation}
    We let $\Delta$ be a minimal translation of $W_0$ as in Theorem~\ref{thm: main
    trichotomy}. We choose $\Delta$ in some prescribed manner, e.g., as the minimal
    translation which is first in lexicographic order among the
    minimal translations with smallest $\ell_1$ norm. Write
    \begin{equation*}
      \delta := h(\Delta)-h(\zero)=h(\Delta).
    \end{equation*}
\item $V_0=V_0(h)$ is the maximal type $0$ \emph{sublevel} component
    satisfying
    \begin{equation}\label{eq:V_0_prop}
      \text{$h(\intb V_0)=h(\intb W_0)-1$,\, $W_0-\Delta\subseteq V_0\subseteq W_0$,\, $\zero
    \notin V_0$ and $-n e_1 \in V_0$}.
    \end{equation}
\item $U_0=U_0(h)$ is defined by the property that its complement
    $U_0^c$ is the minimal type $0$ \emph{superlevel} component
    such that
    \begin{equation}\label{eq:U_0_prop}
      \text{$h(\intb U_0)= h(\extb W_0)-\delta/2$,\, $W_0-\Delta\subseteq U_0\subseteq W_0$,\, $\zero \notin U_0$
    and $-n e_1 \in U_0$}.
    \end{equation}
\end{itemize}

We remark that the third and fourth properties
in~\eqref{eq:V_0_prop} and~\eqref{eq:U_0_prop} in fact follow from
the first two properties. Nonetheless, to simplify our arguments we
include them as part of the definition. The sets $U_0$, $V_0$ and
$W_0$ of a certain $h\in \QP_{(6,0)}$ are illustrated in
Figure~\ref{fig: uvw}.

\begin{figure}[htb!]
\centering%
\includegraphics[scale=0.3]{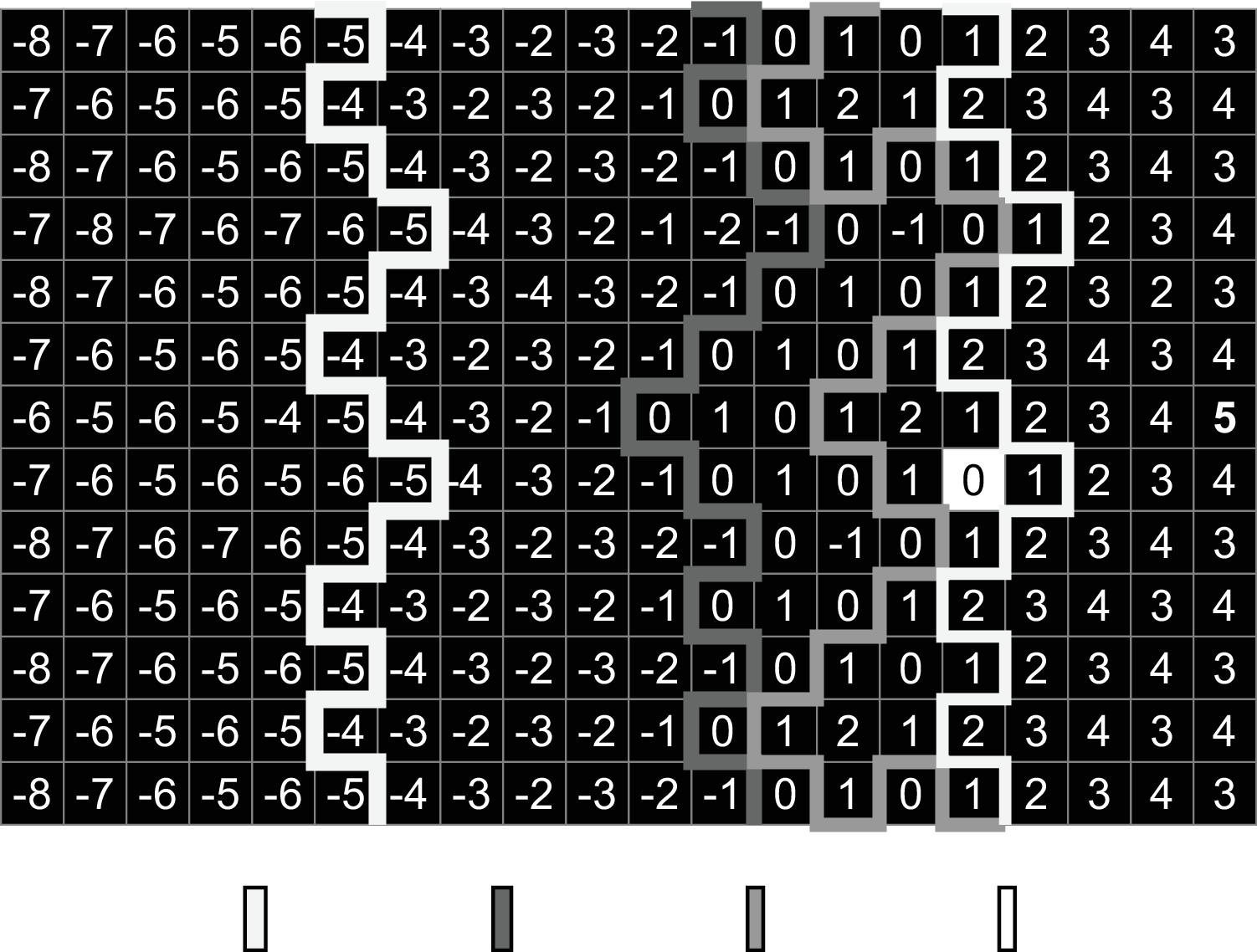}\\
    \rput(-1.58,0.75){$U_0=$}
    \rput(0,0.75){$V_0=$}
    \rput(1.58,0.75){$W_0=$}
    \rput(-3.16,0.75){$W_{-1}=$}
\caption{The boundaries of $U_0$, $V_0$, $W_0$
and $W_{-1}=W_0-\Delta$ for $\Delta=n e_1$ and $\delta=6$. The sets themselves are in all cases to the
left of the boundary. $\zero$ is marked in white.} \label{fig: uvw}
\end{figure}

\begin{figure}[htb!]
\centering%
\includegraphics[scale=0.3]{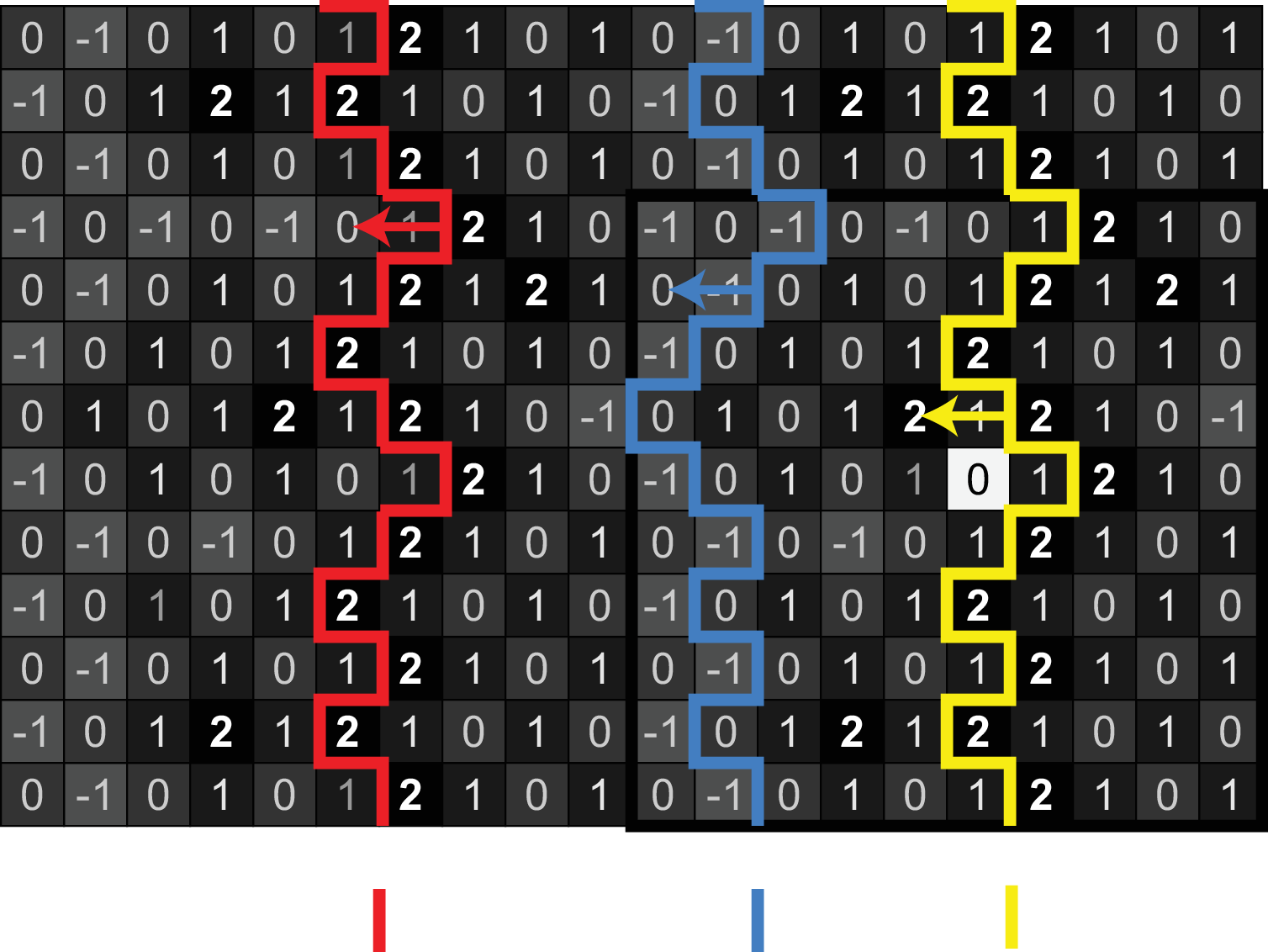}\\
    \rput(0,0.75){$U_0=$}
    \rput(1.58,0.75){$W_0=$}
    \rput(-2.41,0.75){$W_{-1}=$}
\caption{The image through $\Psi$ of the HHF illustrated in figure \ref{fig: uvw}. The boundaries of $U_0$, $W_0$ and
$W_{-1}$ are highlighted to allow the reader to follow the behavior of $\Psi$ in different regions.
$\zero$ is marked in white.} \label{fig: uvw-ref}
\end{figure}

\begin{propos}\label{prop: U and V exist}
$W_0$, $V_0$ and $U_0$ are well-defined, and satisfy
\begin{equation}\label{eq:W_U_V_containment_relations}
W_0-\Delta\subsetneq U_0 \subsetneq V_0 \subsetneq W_0.
\end{equation}
\end{propos}
\begin{proof}
For brevity we write $U$, $V$ and $W$, for $U_0$, $V_0$ and $W_0$
respectively. We begin by showing that $W$ is well defined. Write
$\W$ for the set of type $0$ \emph{sublevel} components which
contain $\zero$ and do not contain $n e_1$. Recalling
\eqref{eq:L_u_v_def} we observe that $\W\subseteq \L_{(0,n e_1)}$.
Thus, by Proposition~\ref{prop: formula for hight diff.}, $\W$ is
ordered by inclusion and finite. By Corollary~\ref{cor: there is
slope}, $\W\neq\emptyset$, and thus $W$, the minimal element of
$\W$, is well defined.

Next, towards showing that $V$ is well defined, we write $\V$ for the set of
type $0$ \emph{sublevel} components $V'$ satisfying $h(\intb V')=h(\intb
W)-1$, $W-\Delta \subsetneq V'\subsetneq W$, $\zero\notin V'$ and $-n e_1 \in
V'$. We observe that $\V\subseteq \L_{(-n e_1,0)}$, and thus by
Proposition~\ref{prop: formula for hight diff.}, $\V$ is ordered by inclusion
and finite. To derive the existence of $V$, all that remains is to show that
$\V\neq \emptyset$.

To see that $\V\neq \emptyset$, we make some observations about $\Delta$ and
$\delta$. Since $h\in\QP_m$, $m\in6\ZZ$ and $\Delta\in n\ZZ$, it follows that
\begin{equation}\label{eq:delta_div_6}
\delta\equiv0\pmod6.
\end{equation}
Since $W$ is of type $0$, $\zero\in W$ and $ne_1\notin W$ we get
that $W\subsetneq W+ne_1$ and therefore, by Theorem~\ref{thm: main
trichotomy},
\begin{equation}\label{eq:W_ne1_W_k_relation}
  \text{$W+ne_1=W+k\Delta$ for some positive $k$.}
\end{equation}
We deduce, using Proposition~\ref{prop: quasi per LC}, that $h(\intb
W+ne_1)=h(\intb W)+h(ne_1)=h(\intb W)+m_1$, and therefore that
$m_1=k\delta$. In particular, since $m_1>0$, we see that
\begin{equation}\label{eq:delta_at_least_6}
  \delta\ge6.
\end{equation}
By subtracting $ne_1$ and $k\Delta$ from both sides of
\eqref{eq:W_ne1_W_k_relation} we have that $W-ne_1 = W-k\Delta$.
Thus, recalling that $0\in W$ and $W-k\Delta\subseteq W-\Delta$, we
obtain that
\begin{equation}\label{eq:ne_1_W_Delta}
  -ne_1 \in W-\Delta.
\end{equation}

By Proposition~\ref{prop: quasi per LC} and
\eqref{eq:delta_at_least_6} we get that $W-\Delta$ is a sublevel
component satisfying $h(\extb (W-\Delta))=h(\extb W)-\delta\le
h(\extb W) - 6$. Thus, the first item of Lemma~\ref{lem: increasing
decreasing containment} guarantees the existence of a type $0$
sublevel component $V'$ satisfying $h(\intb V')=h(\intb W)-1$,
$W-\Delta \subsetneq V'\subsetneq W$. Since $-n e_1\in W-\Delta$ by
\eqref{eq:ne_1_W_Delta} we get that $-n e_1 \in V'$. By the
minimality of $W$, we get that $\zero \notin V'$ implying that
$V'\in\V$ so that $\V\neq\emptyset$.

To show that $U$ is well defined, we write $\U$ for the set
containing all $U'$ such that $(U')^c$ is a type $0$
\emph{superlevel} component such that $h(\intb U')= h(\extb
W)-\delta/2$, $W-\Delta\subseteq U'\subseteq W$, $\zero \notin U'$
and $-n e_1 \in U'$. Recalling Definition~\ref{def: dec dual} of
superlevel sets we use Proposition~\ref{prop: formula for hight
diff.} to deduce that the set of superlevel sets containing $\zero$
and not containing $-n e_1$ is finite and ordered by inclusion, and
therefore $\U$ is also finite and ordered by inclusion. All that
remains in order to deduce the existence of $U$ is to show that
$\U\neq\emptyset$.

This time we apply \eqref{eq:delta_at_least_6} and the second item
of Lemma~\ref{lem: increasing decreasing containment}, to $h$, $V$
and $W-\Delta$, to show the existence of $U'$ satisfying that
$(U')^c$ is a superlevel component of type $0$, $W-\Delta\subseteq
U'\subsetneq V$ and $h(\intb U')= h(\extb W)-\delta/2$. Since
$\zero\notin V$ by definition and $-n e_1\in W-\Delta$ by
\eqref{eq:ne_1_W_Delta} we get that $\zero \notin U'$ and $-n e_1
\in U'$. Thus $U'\in \U$, $\U\neq\emptyset$ so that $U$ is well
defined.

To conclude the proof we must show that $U\subsetneq V$ as with the definitions of $U$ and $V$
this will imply \eqref{eq:W_U_V_containment_relations}.
Let $u\in \intb U$, $w\in\extb W$ and write $V'=\LC_h^{(h(\intb W)-1)+}(u,w)$. Our goal is to show that $V'\in\V$.
By the first item of Corollary~\ref{cor: outside LC crit} applied to $V'$ and $W^c$ we have $V'\subseteq W$.
Since $U$ is \bicon, by Proposition~\ref{prop: Timar II} we have that $\intb U\cup \extb U$ is a connected set of vertices. Moreover, $h(\intb U),h(\extb U) \le h(\extb W)-\delta/2+1 \le h(\intb W)-1$. Therefore $\intb U\cup \extb U\subseteq \LC_h^{(h(\intb W)-1)+}(u)$ and since $w\notin U$
we have $U\subseteq V'$ and hence that $-ne_1\in V'$. Since $W-\Delta\subsetneq V' \subsetneq W$ and since $W$ is of type $0$, we have by Proposition~\ref{prop: Type rules}
that $V'$ is of type $0$. By the minimality of $W$ we get that $0\notin V'$. Thus $V'\in\V$. Since $U\subsetneq V'$, and since $V'\subseteq V$ by the maximality of $V$, we obtain \eqref{eq:W_U_V_containment_relations} as required.
\end{proof}

For $i\in \Z$, we write
\begin{equation}\label{eq:U_i_V_i_W_i_def}
U_i:=U_0+i\Delta,\, V_i:=V_0+i\Delta\,\text{ and
}\,W_i:=W_0+i\Delta.
\end{equation}
\begin{propos}\label{prop: absolute order}
For every $z\in n\ZZ$ and $i\in\mathbb{Z}$ the following are equivalent:
\begin{itemize}
\item $U_0+z=U_i$,
\item $V_0+z=V_i$,
\item $W_0+z=W_i$.
\end{itemize}
\end{propos}
\begin{proof}
We begin by showing that $U_0,V_0$ and $W_0$ all have $\Delta$ as a minimal
translation. For $W_0$, this is the case by the definition of $\Delta$. We
now show this for $V_0$. The proof for $U_0$ is similar. Let $\Delta_V$ be a
minimal translation of $V_0$. Since
$$V_0-\Delta \subsetneq W_0-\Delta \subsetneq V_0\subsetneq W_0$$
by \eqref{eq:W_U_V_containment_relations}, we have $V_0-k\Delta_V=V_0-\Delta$
for some integer $k\ge 1$. We need to show that $k=1$. By Proposition~\ref{prop: dist prop}, we have
$$\dist(V_0-\Delta_V,W_0^c)>\dist(V_0,W_0^c)=\dist(V_0-\Delta_V,W_0^c-\Delta_V).$$
We deduce that $W_0-\Delta_V\subsetneq W_0$, and thus $W_0-\Delta_V\subseteq
W_0-\Delta$ (by the minimality of $\Delta$). Suppose to the contrary that
$W_0-\Delta_V\subsetneq W_0-\Delta$. Since $\Delta$ is a minimal translation
of $W_0$, we get that
$$V_0-\Delta_V \subsetneq W_0-\Delta_V \subseteq W_0-2\Delta \subsetneq V_0-\Delta,$$
contradicting  the minimality of $\Delta_V$. We conclude that
$W_0-\Delta=W_0-\Delta_V$. From this, using \eqref{eq:W_U_V_containment_relations} again, we have that
\begin{equation*}
  V_0 - 2\Delta_V\subsetneq W_0-2\Delta_V = W_0-2\Delta \subsetneq V_0 - \Delta = V_0 - k\Delta_V,
\end{equation*}
so that $k\le 1$, implying that $k=1$ as we wanted to show.

Fix $z\in n\ZZ$. Since $U_0, V_0$ and $W_0$ are of type $0$ with $\Delta$ as
a minimal translation, there exist $i,j,k$ for which $U_0+z=U_i$,
$V_0+z=V_j$, $W_0+z=W_k$. Translating \eqref{eq:W_U_V_containment_relations}
by $z$, we have
\begin{equation}\label{eq:W_U_V_i_j_k}
W_{k-1} \subsetneq  U_i \subsetneq V_j \subsetneq W_k.
\end{equation}
However, \eqref{eq:W_U_V_containment_relations} and
\eqref{eq:U_i_V_i_W_i_def} imply that
\begin{equation*}
  W_{-1}\subsetneq U_0 \subsetneq V_0\subsetneq W_0\subsetneq
  U_1\subsetneq V_1.
\end{equation*}
Hence we conclude from \eqref{eq:W_U_V_i_j_k} and the fact that
$(U_i), (V_i)$ and $(W_i)$ are ordered by inclusion that
$$ k-1 < i \le j \le k$$ and therefore that $i=j=k$.
\end{proof}

We define the mapping $\Psi_m\colon\QP_m\to\QP_{\zero}$ by
\begin{equation}\label{eq:Psi_m_def}
\Psi_m(h)(v):=\begin{cases}
h(v-i\Delta)=h(v)-i\delta,                 & v\in W_i\setminus U_i\text{ for some $i\in\Z$}\\
2h(\extb W_0)-h(v-i\Delta)=2h(\extb W_0)-h(v)+i\delta,     & v\in
U_{i+1}\setminus W_i\text{ for some $i\in\Z$}
\end{cases}.
\end{equation}
The remainder of the section is dedicated to showing that $\Psi_m$
is well defined and has the required properties.

By Theorem~\ref{thm: main trichotomy}, for every $i\in\Z$ we have
$W_{i}\subsetneq W_{i+1}$. Thus, applying Proposition~\ref{prop:
type 0 union is all} to $W_0$, we have that every $v\in\ZZ$ belongs
to exactly one set of the form $W_{i+1}\setminus W_i$. Hence
$\Psi_m(h)(v)$ is defined for every $v\in\ZZ$. The image through
$\Psi$ of the HHF illustrated in Figure~\ref{fig: uvw} is depicted
in Figure~\ref{fig: uvw-ref}.

By definition, $\Psi_m(h)$ is $\Delta$-periodic , i.e., it satisfies
$\Psi_m(h)(v)=\Psi_m(h)(v+\Delta)$ for every $v\in\ZZ$. Thus to
understand $\Psi_m(h)$ it suffices to understand its values on $v\in
W_0\setminus W_{-1}$. As a first step to this end we point out that
on the region $W_0\setminus U_0$, $\Psi_m$ is the identity while on
the region $U_0\setminus W_{-1}$ it is a reflection with respect to
height $h(\extb W_0)-\delta/2 = h(\intb U_{0})$.

\begin{propos}\label{prop: psi_m one to one}
$\Psi_m$ is a one-to-one mapping from $\QP_m$ to $\QP_{\zero}$.
\end{propos}
\begin{proof}
Write $t:=\Psi_m(h)$. We need to show that $t$ is periodic in $n
e_i$ for every $1\le i\le d$, that it is a height function, and that
$\Psi_m$ is one-to-one.

\heading{$t$ is Periodic} First we show that for every $\Delta'\in
n\ZZ,\, a\in\Z$ such that $W_0+\Delta'=W_a$, we have
\begin{equation}\label{eq: zero deltas}
h(v)=h(v+\Delta'-a\Delta)\text{ for all }v\in\ZZ.
\end{equation}
By quasi-periodicity, for all $v\in\ZZ$, we have
$h(v+\Delta'-a\Delta)=h(v) + (h(\Delta'-a\Delta)-h(\zero))$. Hence
it suffices to prove \eqref{eq: zero deltas} for a single $v\in\ZZ$.
Next, note that since $W_0=W_a-\Delta'=W_0+a\Delta-\Delta'$ we have
that if $v\in\intb W_0$, then $v+\Delta'-a\Delta$ is also a member
of $\intb W_0$, implying, by the definition of $W_0$ that $h(v)=h(v+\Delta'-a\Delta)$. This
establishes \eqref{eq: zero deltas}.

Now, let $1\le j\le d$, and suppose that $o_{W_0}(W_0+n e_j)=a\in\Z$
where $o_{W_0}$ is the order function of $W_0$ given by
Theorem~\ref{thm: main trichotomy}. Observe that $W_0+ne_j=W_a$.
Note that if $v\in W_i\setminus U_i$ then, by Proposition~\ref{prop:
absolute order}, $v+ne_j\in W_{i+a}\setminus U_{i+a}$. Thus, using
\eqref{eq: zero deltas}, if $v\in W_i\setminus U_i$ then
$$t(v)= h(v-i\Delta) = h(v+n e_j-(i+a)\Delta) = t(v+ne_j).$$
Similarly, if $v\in U_{i+1}\setminus W_i$ then, using
Proposition~\ref{prop: absolute order}, we have
$$t(v)=2h(\extb W_0)-h(v-i \Delta)=2h(\extb W_0)-h(v+ne_j-(i+a) \Delta)=t(v+ne_j).$$

\heading{$t$ is an HHF} We claim that $t\in\hm(\ZZ)$, i.e., that the values
which $t$ assigns to adjacent vertices differ by exactly $1$. Let $u,v$ be
adjacent vertices in $\ZZ$. We need to show that
\begin{equation}
\label{eq: being HHF}|t(u)-t(v)|=1.
\end{equation}
Since $t$ is $\Delta$-periodic, for every vertex $w\in\ZZ$ there
exists $j\in \Z$ such that $w+j\Delta\in U_1\setminus U_0$ and
$t(w)=t(w+j\Delta)$. We may therefore assume WLOG $u\in U_1\setminus
U_0$, and $v\in U_1$. We consider three cases separately.

First, if both $u,v\in U_1\setminus W_0$ or both $u,v\in
W_0\setminus U_0$ then \eqref{eq: being HHF} follows directly from
the definition of $\Psi_m$.

Second, note that
\begin{equation*}
t(\extb W_0)-t(\intb W_0)=2h(\extb W_0)-h(\extb W_0)-h(\intb
W_0)=h(\extb W_0)-h(\intb W_0)=1.
\end{equation*}
Hence  \eqref{eq: being HHF} holds if either $u\in\extb W_0$ and
$v\in\intb W_0$ or vice versa.

Third,
\begin{equation*}
t(\extb U_0)-t(\intb U_0)= h(\extb U_0) - (2h(\extb W_0) - h(\intb
U_0)-\delta),
\end{equation*}
and plugging the relation $h(\extb W_0)=h(\intb U_0)+\delta/2$ from
\eqref{eq:U_0_prop} yields
\begin{equation*}
t(\extb U_0)-t(\intb U_0)=h(\extb U_0)-h(\intb U_0)=1.
\end{equation*}
Thus \eqref{eq: being HHF} holds if $u\in\extb U_0$ and $v\in\intb
U_0$.

\heading{$\Psi_m$ is one-to-one} To show that $\Psi_m$ is
one-to-one, we explain how to construct an inverse for it. Suppose
that we are able to recover $U_0, W_0, \Delta$ and $\delta$ from $t$
and $m$. Then we may define $U_i=U_0+i\Delta$, $W_i=W_0+i\Delta$ and
the mapping
\begin{equation*}
\Psi^{-1}_m(t)(v):=\begin{cases}
t(v)+i\delta,                 & v\in W_i\setminus U_i\text{ for some $i\in\Z$}\\
2t(\extb W_0)-t(v)+i\delta,     & v\in U_{i+1}\setminus W_i\text{
for some $i\in\Z$}
\end{cases}.
\end{equation*}
It is simple to check that this $\Psi_m^{-1}$ is indeed an inverse
to $\Psi_m$. It is therefore sufficient to show that $U_0, W_0,
\Delta$ and $\delta$ may be recovered from $t$ and $m$.

We begin by recovering $W_0$. To do this we follow the lines of the
proof of proposition~\ref{prop: U and V exist}. Write $\W_t$ for the
set of type $0$ \emph{sublevel} components of $t$ which contain
$\zero$ and do not contain $n e_1$. Again we recall
\eqref{eq:L_u_v_def} and observe that $\W_t\subseteq \L_{(0,n
e_1)}$, where $\L$ is defined with respect to $t$. Thus, by
Proposition~\ref{prop: formula for hight diff.}, $\W_t$ is ordered
by inclusion and finite. We now argue that $\W_t$ is a non-empty set
whose minimal element is $W_0$.

The definition \eqref{eq:Psi_m_def} of $\Psi_m$ and the relation
$h(\extb W_0)=h(\intb U_0)+\delta/2$ from \eqref{eq:U_0_prop} imply
that
\begin{equation}\label{eq:t_h_equality}
  t(x)=h(x)\quad\text{ for $x\in W_0^+\setminus U_0^-$}.
\end{equation}
We can therefore apply Proposition~\ref{prop: two HHFs, same LC 2}
with $h_1=h$, $h_2=t$, $A=W_0$ and $B = U_0$ to get that
\begin{equation}\label{eq:W_0_in_W_t}
W_0\in \W_t.
\end{equation}
 Applying the same proposition with $A=V_0$ yields that
\begin{equation}\label{eq: V_0 is LC t}
\text{$V_0$ is a sublevel component of $t$.}
\end{equation}
Let us write $W_t$ for the minimal element of $\W_t$. Since $W_0\in
\W_t$ we conclude that
\begin{equation}\label{eq:W_t_W_0_inclusion}
  W_t\subseteq W_0.
\end{equation}
To obtain the opposite inclusion we now show that $W_t$ is also a
sublevel component of $h$. Observe that since $W_t$ is of type 0,
and since $\zero \in W_t$ and $ne_1\notin W_t$ we have by
Theorem~\ref{thm: main trichotomy} that $W_t - ne_1\subsetneq W_t$.
We deduce that $-ne_1\in W_t\cap V_0$. In addition, our definitions
imply that $ne_1\in (W_t)^c\cap (V_0)^c$ and $\zero \in
(W_t\setminus V_0)$. By Theorem~\ref{thm: pair_trich}, using that
distinct sublevel components of $t$ are boundary disjoint by
Proposition~\ref{prop: LC boundaries are disjoint}, we deduce that
$V_0\subsetneq W_t$. Applying Proposition~\ref{prop: two HHFs, same
LC 2} with $h_1=t$, $h_2=h$, $A=W_t$ and $B=V_0$, using
\eqref{eq:t_h_equality} and \eqref{eq:W_t_W_0_inclusion} to check
the condition \eqref{eq:h_1_h_2_equality_assumption}, we get that
$W_t$ is a sublevel component of $h$. Together with
\eqref{eq:W_t_W_0_inclusion}, the minimality of $W_0$ now implies
that $W_t=W_0$, allowing the recovery of $W_0$ from $t$. After
recovering $W_0$, we can recover $\Delta$ and $\delta$ using the
fact that $\Delta$ is a minimal translation of $W_0$ chosen in a
prescribed manner and the fact that by Corollary~\ref{cor: type 0 height formula} we have
$\delta \cdot o_{W_0}(W_0+ne_1)= m_1$, where $o_{W_0}$ is the order function on translations of
$W_0$, given by Theorem~\ref{thm: main trichotomy}.

All that remains is to recover $U_0$. Following again the lines of
the proof of Proposition~\ref{prop: U and V exist}, we write $\U_t$
for the set containing all $U'$ such that $(U')^c$ is a type $0$
\emph{superlevel} component of $t$ and $t(\intb U')= t(\extb
W_0)-\delta/2$, $W_0-\Delta\subseteq U'\subseteq W_0$, $\zero \notin
U'$ and $-n e_1 \in U'$. Recalling Definition~\ref{def: dec dual} of
superlevel sets we again use Proposition~\ref{prop: formula for
hight diff.} to deduce that the set of superlevel components
containing $\zero$ and not containing $-n e_1$ is finite and ordered
by inclusion, implying that $\U_t$ is also finite and ordered by
inclusion. We now use \eqref{eq:t_h_equality} and
Proposition~\ref{prop: two HHFs, same LC 2}, with $h_1=-h$,
$h_2=-t$, $A=(U_0)^c$ and $B =(W_0)^c$, to get that $U_0^c$ is a
superlevel component of $t$ (again, using Definition~\ref{def: dec
dual} of superlevel components). It follows from
\eqref{eq:t_h_equality} that $U_0\in \U_t$. Write $U_t$ for the
maximal element of $\U_t$, i.e., the complement of the minimal
element amongst complements of elements in $\U_t$. Since $U_0\in
\U_t$ we conclude that
\begin{equation}\label{eq:U_t_U_0_inclusion}
  U_t^c\subseteq U_0^c.
\end{equation}
Recall that, by the definition of $\U_t$, we have $(W_0)^c\subsetneq
(U_t)^c$ and that, by \eqref{eq:W_0_in_W_t}, $W_0$ is also a
sublevel component of $t$. Applying Proposition~\ref{prop: two HHFs,
same LC 2} to $h_1=-t$, $h_2=-h$, $A=(U_t)^c$ and $B =(W_0)^c$,
using \eqref{eq:t_h_equality} and \eqref{eq:U_t_U_0_inclusion} to
check the condition \eqref{eq:h_1_h_2_equality_assumption}, we get
that $U_t^c$ is also a superlevel component of $h$. We also have
$h(\intb U_t)= h(\extb W_0)-\delta/2$ by \eqref{eq:t_h_equality}.
Thus, together with \eqref{eq:U_t_U_0_inclusion}, the minimality of
$U_0^c$ now implies that $U_0=U_t$. As $W_0, U_0, \Delta$ and
$\delta$ can be recovered from $t$ and $m$, we deduce that $\Psi_m$
is one-to-one.
\end{proof}

\subsection{Proof of Theorem~\ref{thm: QPm QP0 embedding}}\label{sec: embedding
proof} In this section we prove Theorem~\ref{thm: QPm QP0 embedding}
using a bound on the probability for a uniformly chosen HHF on the
torus to have a sublevel component with long boundary. Here, for the
first time, we use sublevel components on $\TT$ (defined in
Section~\ref{sec: QP properties}). To clarify our proof we will
always denote HHFs in $\hm(\TT)$ by $r$, HHFs in $\QP_{\zero}$ by
$t$ and HHFs in $\QP_m$, for arbitrary $m=(m_1,\ldots, m_d)$ with
$m_1>0$, by $h$.

Recall that for $u\in\TT$ we denoted by $\hm(\TT, u)$ the set of all
homomorphism height functions on $\TT$ which are zero at $u$. We use
the following theorem of \cite{PHom} to derive the estimates of
Theorem~\ref{thm: QPm QP0 embedding}.

\begin{theorem}[{\cite[special case of Theorem~2.8]{PHom}}]\label{thm: long level set estimate}
There exist $c>0$ and $d_0$ such that in all dimensions $d\ge d_0$,
for all even $n$, all $u,v\in\TT$ and all $L\ge 1$, if $h$ is
uniformly sampled from $\hm(\TT,u)$ then
\begin{equation*}
\P\left(|\partial \LC_h^{0+}(u,v)|\ge L\right)\le d\exp \left(
-\frac{cL}{d\log^2 d} \right),
\end{equation*}
where we mean that $\LC_h^{0+}(u,v)=\emptyset$ if $h(v)\le 0$.
\end{theorem}

We adapt Theorem~\ref{thm: long level set estimate} to our setting
through the following corollary.
\begin{cor}\label{cor: long level set estimate}
There exist $c>0$ and $d_0$ such that in all dimensions $d\ge d_0$,
for all even $n$ and all $L\ge 1$, denoting
$$A:=\left\{r\in\hm(\TT)\ :\ \text{ there exists a sublevel component $A$ such that $|\partial A|\ge L$} \right\},$$ the following holds,
\begin{equation*}
\frac{\left|A\right|}{\left|\hm(\TT)\right|}\le 2d^2 n^d
\exp\left(-\frac{cL}{d\log^2 d}\right).
\end{equation*}
\end{cor}
\begin{proof}
Fix $L\ge 1$ and let $B:=\{r\in \hm(\TT)\ : \ \exists v\in\TT, v\sim
\zero,\text{ s.t. }|\partial \LC_r^+(\zero,v)|\ge L\}$. By Theorem~\ref{thm:
long level set estimate} with $u=\zero$, and using a union bound on all
$v\sim \zero$, we have
$$|B|\le 2d \cdot d\exp\left(-\frac{cL}{d\log^2 d}\right)\big|\hm(\TT)\big|\quad\text{ for all $d$ greater then some fixed $d_0$.}$$
Now, for every $w\in \TT$ define the mapping $\eta_w :\hm(\TT)\to
\hm(\TT)$ by
$$\eta_w(r)(v):=r(v+w)-r(w).$$
It is not difficult to check that this mapping is well defined and
is a bijection. Moreover, for every $r\in A$ there exists a
$w\in\TT$ such that $\eta_w(r)\in B$. The corollary follows.
\end{proof}

In order to apply Corollary~\ref{cor: long level set estimate}, we
must show that HHFs in the image of $\Psi_m$, when projected to the
torus, contain a sublevel component with a long boundary. We proceed
in two steps. First, we claim that the projection of the boundary of
the set $V_0$ from Proposition~\ref{prop: U and V exist} is
contained in the boundary of a sublevel component of the projection
of $\Psi_m(h)$. Then we claim that this boundary is long. Recall
that $\pi$ was defined in Section~\ref{subs: predef} to be the
natural projection from $\ZZ$ to $\TT$. Here we use also the natural
extension of $\pi$ to edges of $\ZZ$.

\begin{lem}\label{lem: projection of LC boundary}
Let $h\in \QP_m$ for $m\in6\ZZ$ satisfying $m_1 > 0$. Let
$r=\pi\circ \Psi_m(h)$ and $V_0$ be as in Proposition~\ref{prop: U
and V exist}. There exists a sublevel component $R$ of $r$ such that
$\pi (\partial V_0)\subseteq
\partial R$.
\end{lem}
We delay the proof of this lemma to Section~\ref{subs: Projection of LC
boundary}.

At last we are ready to prove the theorem.
\begin{proof}[Proof of Theorem \ref{thm: QPm QP0 embedding}]
Let $m\in6\ZZ\setminus\{\zero\}$. Using an appropriate rotation we may
assume without loss of generality that $m_1>0$. Fixing $h\in \QP_m$ and
applying Lemma~\ref{lem: projection of LC boundary} and Lemma~\ref{lem: type 0 steep has long boundary in increasing direction}, we obtain the existence of a
sublevel component $R$ of $\pi\circ \Psi_m(h)$ such that $|\partial R|\ge
n^{d-1}$. Thus
$$\pi (\Psi_m(\QP_m)) \subset \{r\in\hm(\TT)\,:\,\text{there exists a sublevel component $A$ of $r$ such that $|\partial A|\ge
n^{d-1}$ }\}.$$ Recall that $\pi$ is a bijection from $\QP_\zero$ to
$\hm(\TT)$. Thus, applying Corollary~\ref{cor: long level set
estimate}, we get that for large enough $d$,
$$|\Psi_m(\QP_m)|\le 2d^2 n^d\exp\left(-\frac{cn^{d-1}}{d\log^2 d}\right) |\hm(\TT)|
\le \exp\left(-\frac{c'n^{d-1}}{d\log^2 d}\right)|\QP_\zero|$$ for some
$c,c'>0$. Thus, since $\Psi_m$ is one-to-one, the theorem follows.
\end{proof}

\subsection{Projecting type $0$ sublevel components to the torus}\label{subs: Projection of LC
boundary} In this section we prove Lemma~\ref{lem: projection of LC boundary}
connecting sublevel components on $\QP_{\zero}$ with those on $\hm(\TT)$. While
the relation between sublevel \emph{components} of HHFs on the integer
lattice and those of HHFs on the torus is non-trivial, the relation between
sublevel \emph{sets} of the two spaces is much simpler. In particular,
\begin{equation}\label{eq: LC(1) Property}
\pi(\LL_{t}^+(u))=\LL_{\pi(t)}^+(\pi(u))\text{ for all } t\in\QP_{\zero}\text{ and }u\in\ZZ.
\end{equation}
This can be easily verified from the definition of sublevel sets.

Next, we prove a proposition relating the boundaries of sublevel
components on $\ZZ$ to those of sublevel components on $\TT$. We then show that
this proposition applies to the set $V_0$ from
Proposition~\ref{prop: U and V exist}, and use this fact to prove
Lemma~\ref{lem: projection of LC boundary}. We remind the reader
that $A^+$ and $A^{++}$ were introduced in Section~\ref{subs:
predef}.

\begin{propos}\label{prop: Torus plain bdry relation}
Let $t\in\QP_{\zero}$ and $r=\pi(t)\in\hm(\TT)$. Suppose
$V:=\LC_{t}^+(u,v)$ for adjacent vertices $u,v\in\ZZ$ satisfying
$t(v)=t(u)+1$. If
\begin{equation}\label{eq: condition of projections}
\pi(V^{++}\setminus V) \cap \pi(\LL^+_t(u))=\emptyset
\end{equation}
then
\begin{equation}\label{eq: pi boundary}
\pi(\partial V)\subseteq \partial
\LC^+_{r}(\pi(u),\pi(v)).
\end{equation}

\end{propos}
\begin{proof}
Let $R := \LC^+_{r}(\pi(u),\pi(v))$. We first note that \eqref{eq: pi
boundary} follows from the following two claims,
\begin{align}
\pi(\intb V)&\subseteq R,\label{eq:intb is ok}\\
\pi(\extb V)&\subseteq R^c \label{eq:extb
is ok}.
\end{align}
We begin by showing \eqref{eq:intb is ok}. Indeed, we have:
$$\pi(\intb V)\subseteq\pi(\LL_t^+(u)) = \LL_{r}^+(\pi(u)) \subseteq R,$$
where the equality follows from \eqref{eq: LC(1) Property}, and the two
containment relations follow from Proposition~\ref{prop: basic LC prop}.

Next we show \eqref{eq:extb is ok}. By Proposition \ref{prop: Timar II},
using the fact that $V$ is \bicon{ }by Proposition~\ref{prop: basic LC prop},
we get that $\pi(V^{++}\setminus V)$ is a connected set which contains $\pi(v)$
(recall that $u\sim v$). By \eqref{eq: LC(1) Property} and \eqref{eq:
condition of projections}, $\pi(V^{++}\setminus V)$ is disjoint from
$\LL_{r}^+(\pi(u))$. By the definition of sublevel component this implies
that $\pi(V^{++}\setminus V)\subseteq R^c$. Since $\pi(\extb V)\subseteq
\pi(V^{++}\setminus V)$, we deduce \eqref{eq:extb is ok}.
\end{proof}

At last, we prove Lemma~\ref{lem: projection of LC boundary}. Let
$h\in\QP_m$ for $m\in6\ZZ$ satisfying $m_1 > \zero$. Let $U=U_0$,
$V=V_0$, $W=W_0$ and $\Delta$ be as in Proposition~\ref{prop: U and
V exist}. Let also $t:=\Psi_m(h)$ and $r := \pi(t)$. Our goal is to
show that $V$ satisfies the conditions of Proposition~\ref{prop:
Torus plain bdry relation}, from which Lemma~\ref{lem: projection of
LC boundary} will follow.

Write $\T$ for the set of type $0$ sublevel components $T'$
satisfying $h(\intb T')=h(\extb W)-\delta+1$ and $W-\Delta
\subsetneq T' \subsetneq V$. Recall that $W-\Delta\subsetneq V$ by
\eqref{eq:W_U_V_containment_relations}, $h(\intb (W-\Delta))=h(\intb
W)-\delta$, $h(\intb V) = h(\intb W)-1$ by \eqref{eq:V_0_prop} and
that $\delta\ge 6$ by \eqref{eq:delta_at_least_6}. Hence, by
Lemma~\ref{lem: increasing decreasing containment}, we conclude that
$\T$ is non-empty. Write $T$ for the minimal element of $\T$.

Let us show that $T \subseteq U$. By Lemma~\ref{lem: increasing
decreasing containment} applied to $T\subsetneq V$, using that
$h(\intb T)=h(\extb W)-\delta+1$ and $h(\intb V) = h(\intb W)-1$,
there exists a $U'$ satisfying that $(U')^c$ is a type $0$
\emph{superlevel} component such that $h(\intb U')= h(\extb
W)-\delta/2$ and $T\subseteq U'\subsetneq V$. Next, observe that
$\zero \notin U'$, since $0\notin V$ by \eqref{eq:V_0_prop}, and
that $-n e_1 \in U'$, since $-ne_1\in W-\Delta\subsetneq T$ by
\eqref{eq:ne_1_W_Delta}. Thus,
\eqref{eq:W_U_V_containment_relations} and the definition of $U$ (in
particular, the fact that $U^c$ is minimal), imply that $U'\subseteq
U$. We conclude that
\begin{equation}\label{eq:W_T_U_relation}
  W-\Delta\subsetneq T\subseteq U.
\end{equation}

Next, the definition \eqref{eq:Psi_m_def} of $\Psi_m$,
\eqref{eq:W_T_U_relation} and the definition of $T$ imply that
$$t(\intb T)=2h(\extb W)-h(\intb T)-\delta = h(\extb W) - 1.$$
Now, since $U\subsetneq V\subsetneq W$ by
\eqref{eq:W_U_V_containment_relations}, the definition of $\Psi_m$
implies that
\begin{equation*}
  h(\extb V) = t(\extb V).
\end{equation*}
Thus, by \eqref{eq:V_0_prop},
\begin{equation}\label{eq:T_t_boundary}
  t(\intb T) = t(\extb V).
\end{equation}

We now check that $V$ satisfies the conditions of
Proposition~\ref{prop: Torus plain bdry relation}. Recall that by
\eqref{eq: V_0 is LC t}, $V$ is a sublevel component of $t$. Let
$u\in\intb V$, $v\in\extb V$ be two adjacent vertices. By
Corollary~\ref{cor: LC neighbors} we have $V=\LC_{t}^+(u,v)$.
Observe that the condition \eqref{eq: condition of projections} is
equivalent to
\begin{equation*}
  ((V^{++}\setminus V) + z) \cap \LL^+_t(u)=\emptyset\quad\text{for
  all $z\in n\Z^d$}.
\end{equation*}
Since $(V^{++}\setminus V) + z = (V^{++}+z)\setminus (V+z)$ and
since $V$ is of type $0$ having, by Proposition~\ref{prop: absolute
order}, $\Delta$ as a minimal translation, this is equivalent to
\begin{equation}\label{eq:V++_cond}
  ((V^{++} + k\Delta)\setminus (V+ k\Delta)) \cap \LL^+_t(u)=\emptyset\quad\text{for
  all $k\in \Z$}.
\end{equation}
We note that $T\subsetneq V$ by the definition of $T$. It follows
from \eqref{eq:T_t_boundary} that the set $S:=V\setminus T$
satisfies $t(s)=t(\extb V)$ for all $s\in \extb S$. As $u\notin T$, this implies
that $\LC_t^+(u)\subseteq S$. Thus, to check
condition~\eqref{eq:V++_cond} it suffices to show that
\begin{equation*}
  ((V^{++} + k\Delta)\setminus (V+ k\Delta)) \cap S=\emptyset\quad\text{for
  all $k\in \Z$},
\end{equation*}
which, since $S=V\setminus T$, is itself implied by
\begin{equation}\label{eq:V++_second_cond}
\begin{aligned}
&V^{++} + k\Delta\subseteq T&&\text{for all $k\le -1$,}\\
&V + k\Delta\supseteq V&&\text{for all $k\ge 0$.}
\end{aligned}
\end{equation}
Since $\Delta$ is a minimal translation for $V$, the second part of
\eqref{eq:V++_second_cond} follows trivially and it suffices to
check the first part for $k=-1$. Finally, the condition that
$(V^{++} - \Delta)\subseteq T$ follows from the fact that $V-\Delta,W-\Delta$ and $T$ are boundary disjoint and satisfy $V-\Delta
\subsetneq W-\Delta \subsetneq T$. This is a consequence of
\eqref{eq:W_U_V_containment_relations}, Proposition~\ref{prop: LC boundaries are disjoint} and the definition of $T$.
Thus the condition of Proposition~\ref{prop: Torus plain bdry
relation} is satisfied. Lemma~\ref{lem: projection of LC boundary}
follows from \eqref{eq: pi boundary}.\qed

\section{Steep slopes are extremely unlikely}\label{sec: Steep slopes}
In this section we detail how to modify the proof of
Theorem~\ref{thm: QPm QP0 embedding} to prove Theorem~\ref{thm: QP steep QP0 embedding}.

A main ingredient in the proof of the theorem is the following
proposition which is a consequence of the results of \cite{PHom}.

\begin{propos}[Enhanced version of Corollary~\ref{cor: long level set estimate}]\label{propos: steep long level set estimate}
There exist $c>0$ and $d_0$ such that in all dimensions $d\ge d_0$,
for all even $n$, all integer $k\ge 1$ and $L\ge n^{d-1}$, denoting
$$A:=\left\{r\in\hm(\TT)\colon\text{$\exists u,v\in\TT$ s.t. $r(v)\ge r(u)+ k$ and $\sum_{j=0}^{k-1}|\partial\LC^{(r(u)+j)+}_r(u,v)|\ge L$} \right\},$$ the following holds,
\begin{equation*}
\frac{\left|A\right|}{\left|\hm(\TT)\right|}\le
\exp\left(-\frac{cL}{d\log^2 d}\right).
\end{equation*}
\end{propos}
We make the assumption that $L\ge n^{d-1}$ in order to simplify the
proof of the proposition and as it suffices for our purposes in this
section but we remark that similar estimates may be established for
all $L\ge 1$ using additional arguments (e.g., the isoperimetric
estimates of \cite[Theorem~5.1]{PHom}).
\begin{proof}[Proof of Proposition~\ref{propos: steep long level set
estimate}] Define
\begin{equation*}
B:=\left\{r\in\hm(\TT)\ :\ \exists v\in\TT\text{ s.t. $r(v)\ge k$
and $\sum_{j=0}^{k-1}|\partial\LC^{j+}_r(\zero,v)|\ge L$} \right\}.
\end{equation*}
Using a similar argument as in the proof of Corollary~\ref{cor: long
level set estimate} (defining the mapping $\eta_w$) it suffices to
show that
\begin{equation*}
  \frac{\left|B\right|}{\left|\hm(\TT)\right|}\le n^{-d}\exp\left(-\frac{cL}{d\log^2 d}\right).
\end{equation*}

Equation~(72) in \cite[Proof of Proposition~5.15]{PHom} implies that
for each $(L_j)$, $0\le j< k$, with $L_j\ge 1$ we have
\begin{equation*}
   \frac{\left|\left\{r\in\hm(\TT)\,:\,\exists v\in\TT\text{ s.t. $|\partial\LC^{j+}_r(\zero,v)|=L_j$ for all $0\le j< k$}\right\}\right|}{\left|\hm(\TT)\right|}\le n^d d^{k}\exp\left(-\frac{c'\sum_{j=0}^{k-1} L_j}{d\log^2 d}\right),
\end{equation*}
for some $c'>0$. We deduce that
\begin{align*}
\frac{\left|B\right|}{\left|\hm(\TT)\right|}
&\le \sum_{\bar{L}=L}^\infty n^d d^{k}\exp\left(-\frac{c'\bar{L}}{d\log^2 d}\right)\big|\{(L_0,\ldots, L_{k-1})\colon L_j\ge 1, L_0+\cdots+L_{k-1}=\bar{L}\}\big|\\
& = n^d d^{k}\sum_{\bar{L}=L}^\infty
\binom{\bar{L}-1}{k-1}\exp\left(-\frac{c'\bar{L}}{d\log^2
d}\right)\le n^d d^{k}\sum_{\bar{L}=L}^\infty
\bar{L}^k\exp\left(-\frac{c'\bar{L}}{d\log^2 d}\right).
\end{align*}
Next observe that if $B$ is non-empty then $k$ is at most the
diameter of $\TT$. Thus we assume without loss of generality that
$k\le dn$. Recalling also our assumption that $L\ge n^{d-1}$, we see
that the ratio of consecutive terms in the last sum equals
\begin{equation*}
  \left(1+\frac{1}{\bar{L}}\right)^k
\exp\hspace{-1pt}\left(-\frac{c'}{d\log^2 d}\right)\le
\exp\hspace{-1pt}\left(-\frac{c'}{d\log^2 d}+\frac{k}{\bar{L}}\right)\le
\exp\hspace{-1pt}\left(-\frac{c'}{d\log^2 d}+\frac{d}{n^{d-2}}\right)\le
\exp\hspace{-1pt}\left(-\frac{c'}{2d\log^2 d}\right)
\end{equation*}
for large enough $d$. Thus we conclude that
\begin{align*}
  \frac{\left|B\right|}{\left|\hm(\TT)\right|}
\le n^d (dL)^{k} \exp\left(-\frac{c'L}{d\log^2 d}\right)
\sum_{\bar{L}=L}^\infty \exp\left(-\frac{c'(\bar{L}-L)}{2d\log^2
d}\right)\\
\le n^d (dL)^{k} d^2\exp\left(-\frac{c'L}{d\log^2 d}\right)\le
n^{-d} \exp\left(-\frac{c'L}{2d\log^2 d}\right)
\end{align*}
for large enough $d$, where in the last step we used again that
$k\le dn$ and $L\ge n^{d-1}$.
\end{proof}

We begin by formulating various properties of functions in $\QP_m$
which are needed for the proof of Theorem~\ref{thm: QP steep QP0
embedding}. Fix $m\in6\ZZ\setminus\{\zero\}$. Assume that
$\QP_m\neq\emptyset$ as the theorem is trivial otherwise. Assume
also, without loss of generality, that the first coordinate of $m$
is positive and is the largest in absolute value among all
coordinates, and write
\begin{equation*}
\sigma:=\frac{m_1}6 = \frac{h(ne_1)}{6}.
\end{equation*}
Fix $h\in\QP_m$ and recall our definition of $W_0$, $\Delta$,
$\delta$ and $\Psi_m$ from Section~\ref{subs: Flattening the slope}.
We write
\begin{equation}\label{eq:p_ell_def}
  p:=\frac{\delta}6\quad\text{and}\quad\ell:=\frac{\sigma}{p}=\frac{h(ne_1)}{\delta}.
\end{equation}
We remark that it is possible to show that $\delta, p$ and $\ell$
depend only on $m$, but since this fact will not be of use to us, we
do not prove it. Observe that $\sigma$ is a positive integer as
$m\in6\ZZ$ and $m_1>0$. In addition, by \eqref{eq:delta_div_6},
\eqref{eq:W_ne1_W_k_relation}, \eqref{eq:delta_at_least_6} and the
argument in the paragraph there, we have that $p$ and $\ell$ are
also positive integers and that
\begin{equation}\label{eq: what is w}
W_0 + n e_1= W_0 + \ell \Delta.
\end{equation}

We wish to find a single pair of vertices separated by $p$ sublevel
components of $\pi\circ\Psi_m(h)$, each with boundary size at least
$\ell n^{d-1}$. We remark that more components may be found, at
least as many as $\frac\delta2-2$, but this is not required for our
results.
We proceed by defining additional type~$0$ sublevel components of
$h$, whose roles are similar to the role of $V_0$ in the proof of
Theorem~\ref{thm: QPm QP0 embedding} (Section~\ref{sec: proof of 2.1}).

Define $V_0^{p-1}$, as the maximal type $0$ sublevel
component of $h$ satisfying
\begin{equation*}
\text{$h(\intb V^{p-1}_0)=h(\intb W_0)-3p+2$,\,
$W_0-\Delta\subseteq V^{p-1}_0\subseteq W_0$,\, $\zero
    \notin V_0^{p-1}$ and $-n e_1 \in V^{p-1}_0$}.
\end{equation*}
Further define $V_0^i$, $0\le i< p-1$, as the minimal type $0$ sublevel
component of $h$ satisfying
\begin{equation*}
\text{$h(\intb V^i_0)=h(\intb W_0)-3i-1$,\,
$V_0^{p-1}\subseteq V^i_0\subseteq W_0$}.
\end{equation*}

We recall the set $U_0$ defined in Section~\ref{subs: Flattening the slope}. The following proposition is a
generalization of the part of Proposition~\ref{prop: U and V exist}
which pertains to $V_0$.

\begin{propos}\label{prop: U and V steep exist}
$\{V^i_0\}_{i=0}^{p-1}$ are well defined, and satisfy
\begin{equation}\label{eq:W_U_V_steep containment_relations}
U_0 \subsetneq V^{p-1}_0 \subsetneq\dots \subsetneq V^0_0 \subsetneq
W_0.
    \end{equation}
\end{propos}
\begin{proof}
Following the proof of Proposition~\ref{prop: U and V exist} we write $\V^{p-1}$ for the set of
type $0$ \emph{sublevel} components $V'$ satisfying $h(\intb
V')=h(\intb W_0)-3p+2$, $W_0-\Delta \subsetneq V'\subsetneq W_0$,
$\zero\notin V'$ and $-n e_1 \in V'$. We observe that $\V^{p-1}\subseteq
\L_{(-n e_1,0)}$, and thus, by Proposition~\ref{prop: formula for
hight diff.}, $\V^{p-1}$ is ordered by inclusion and finite. Moreover,
applying the first part of Lemma~\ref{lem: increasing decreasing
containment} with $h$, $W_0-\Delta$ and $W_0$ we obtain the existence of
a type $0$ sublevel component $V'_{p-1}$  which satisfies $h(\intb
V'_{p-1})=h(\intb W_0)-3p-2$ and $W_0-\Delta \subsetneq V'_{p-1}\subsetneq W_0$. By the minimality of $W_0$ we have
$\zero \notin V'_{p-1}$. As $-ne_1\in W_0-\Delta$ by \eqref{eq:ne_1_W_Delta}, we have $-ne_1\in V'_{p-1}$ and hence $V'_{p-1}\in\V^{p-1}$ and
$\V^{p-1}$ is not empty. We deduce that $V^{p-1}_0$ is well defined.

For $0\le i <p-1$ we now define $\V^{i}$ as the set of type $0$
\emph{sublevel} components $V'$ satisfying $h(\intb V')=h(\intb
W_0)-3i-1$, $V_0^{p-1} \subsetneq V'\subsetneq W_0$. Using the same
arguments as above we observe that $\V^{i}$ is ordered by inclusion
and finite. Applying once again the first part of Lemma~\ref{lem:
increasing decreasing containment}, now with $h$, $V_0^{p-1}$ and
$W_0$ we obtain the existence of a a type $0$ sublevel component
$V'_{i}$ which satisfies $h(\intb V'_{i})=h(\intb W_0)-3i-1$ and
$V_0^{p-1} \subsetneq V'_{i}\subsetneq W_0$. Thus $\V^{i}$ is not
empty and $V^{i}_0$ is well defined.

We proceed to show the inclusion relations \eqref{eq:W_U_V_steep
containment_relations}. We have $V_0^0\subsetneq W_0$ and
$V_0^{p-1}\subsetneq V_0^{p-2}$ (when $p>1$) by the definition of
$V_0^0$ and $V_0^{p-2}$. Next, for each $0<i<p-1$ the first part of
Lemma~\ref{lem: increasing decreasing containment} applied with $h$,
$V_0^{p-1}$ and $V_0^{i-1}$ shows that there exists an element
$V_i''\in\V^i$ satisfying $V_0^{p-1} \subsetneq V''_{i}\subsetneq
V_0^{i-1}$, whence the inclusion $V_0^i\subsetneq V_0^{i-1}$ follows
from the minimality of $V_0^i$. It remains to show that $U_0
\subsetneq V^{p-1}_0$ and this is done next by exhibiting an element
of $\V^{p-1}$ which strictly contains $U_0$.

To do so we repeat the arguments in the proof of
Proposition~\ref{prop: U and V exist} concerning the fact that
$U_0\subsetneq V_0$. Let $u\in \intb U_0$, $w\in\extb W_0$ and write
$V''_{p-1}=\LC_h^{(h(\intb W_0)-3p+2)+}(u,w)$. Our goal is to show
that $V''_{p-1}\in\V^{p-1}$. By the first item of
Corollary~\ref{cor: outside LC crit} applied to $V''_{p-1}$ and
$W_0^c$ we have $V''_{p-1}\subseteq W_0$. Since $U_0$ is \bicon, by
Proposition~\ref{prop: Timar II} we have that $\intb U_0\cup \extb
U_0$ is a connected set of vertices of height less or equal to
$h(\extb W)-\delta/2+1= h(\intb W_0)-3p+2$. Therefore $\intb U_0\cup
\extb U_0\subseteq \LC_h^{(h(\intb W_0)-3p+2)+}(u)$ and since
$w\notin U_0$ we have $U_0\subseteq V''_{p-1}$ and hence $-ne_1\in
V''_{p-1}$. Since $W_0-\Delta\subsetneq U_0\subsetneq V''_{p-1}
\subsetneq W_0$ and since $W_0$ is of type $0$, we get from
Proposition~\ref{prop: Type rules} that $V''_{p-1}$ is of type $0$.
By the minimality of $W_0$ we get that $0\notin V''_{p-1}$. Thus
$V''_{p-1}\in\V^{p-1}$, and by the maximality of $V^{p-1}_0$ we get
that $U_0\subsetneq V''_{p-1}\subseteq V^{p-1}_0$ as required.
\end{proof}

For each $0\le i\le p-1$, repeating the proof of Proposition~\ref{prop: absolute order} with
$V_0^i+j\Delta$ in the role of $V_j$ we get that $\Delta$ is a
minimal translation of $V_0^i$, and
that, moreover, for every $z\in n\Z^d$ and integer $k$,
\begin{equation}\label{eq:V_i_0_and_W_0_share_translations}
  V^i_0 + z = V^i_0 + k \Delta\text{ if and only if }W_0 + z = W_0 + k \Delta.
\end{equation}
In particular, \eqref{eq: what is w} implies that
\begin{equation*}
V^i_0 + n e_1= V_0^i + \ell \Delta\text{ for all $0\le i\le p-1$}.
\end{equation*}
Consequently, by Lemma~\ref{lem: type 0 steep has long boundary in
increasing direction},
\begin{equation*}
|\pi(\partial {V^i_0})|\ge \ell n^{d-1} \text{ for all
}0\le i\le p-1.
\end{equation*}
In particular, by the definition~\eqref{eq:p_ell_def} of $p$ and
$\ell$,
\begin{equation}\label{eq:long_sum_of_level_component_boundaries}
  \sum_{i=0}^{p-1} |\pi(\partial {V^i_0})| \ge p\cdot\ell n^{d-1}=\sigma n^{d-1}.
\end{equation}
Aiming to use Proposition~\ref{propos: steep long level set
estimate}, we proceed to find two vertices on the torus such that
the sublevel components of $\pi\circ\Psi_m(h)$ between these
vertices contain the sets $\pi(\partial {V^i_0})$ in their
boundaries.

\begin{propos}\label{prop:boundary_of_level_component_containment}
  Let $u\in \intb V^{p-1}_0$ and $v\in \extb V^0_0$. Then, for each $0\le i\le p-1$,
\begin{equation}\label{eq:V_i_0_as_sublevel_component}
  V^i_0 = \LC_h^{(h(u)+3(p-1-i))+}(u,v)
\end{equation}
and
\begin{equation}\label{eq:V_i_0_boundary_containment}
  \pi(\partial {V^i_0})\subseteq \partial \LC_{\pi\circ \Psi_m(h)}^{(h(u)+3(p-1-i))+}(\pi(u),\pi(v)).
\end{equation}
\end{propos}
\begin{proof}
As $u\in \intb V^{p-1}_0$, we have $h(u) = h(\intb W_0)-3p+2$ and as $v\in \extb V^0_0$ we have $h(v) = h(\intb W_0)$. Denote
\begin{equation*}
A^i:=\LC_h^{(h(\intb W_0)-3i-1)+}(u,v) = \LC_h^{(h(u)+3(p-1-i))+}(u,v),\quad 0\le i\le p-1.
\end{equation*}
We show that $V^i_0 = A^i$ for all $i$, proving
\eqref{eq:V_i_0_as_sublevel_component}. First, as $u\in \intb
V^{p-1}_0$ and $h(v)>h(u)$, we immediately have that $A^{p-1} =
V^{p-1}_0$ by Corollary~\ref{cor: LC neighbors}. Now fix $0\le i<
p-1$. It follows also that $A^i\supseteq V^{p-1}_0$. Moreover, as
$h(\intb V^{i}_0) = h(\intb A^i)$ and, by \eqref{eq:W_U_V_steep
containment_relations}, $u\in V^i_0$, it follows that
$\LL_h^{h(\intb A^i)+}(u) \subseteq V^i_0$. From this, as $v\notin
V^i_0$ by \eqref{eq:W_U_V_steep containment_relations} and
$(V^i_0)^c$ is connected, we conclude that $A^i\subseteq V^i_0$.
Since $A^i$ both contains and is contained in sublevel components of
type $0$, we conclude from Proposition~\ref{prop: Type rules} that
$A^i$ is of type $0$. The minimality of $V^i_0$ now implies that
$V^i_0\subseteq A^i$, leading to the equality $V^i_0 = A^i$ that we
wanted to prove.

We now proceed to prove \eqref{eq:V_i_0_boundary_containment}. Fix $0\le i\le p-1$. Let us prove first that
\begin{equation}\label{eq:V_i_boundary_in_r_level_set}
  \pi(\intb V^i_0)\subseteq \LL_{\pi\circ \Psi_m(h)}^{h(\intb V^i_0)+}(\pi(u)).
\end{equation}
To this end, it suffices to show that
\begin{equation}\label{eq:V_i_boundary_in_t_level_set}
  \intb V^i_0 \subseteq \LL_{\Psi_m(h)}^{h(\intb V^i_0)+}(u).
\end{equation}
Denote $B^i:=\LL_h^{h(\intb V^i_0)+}(u)$. By the definition \eqref{eq:Psi_m_def} of $\Psi_m$ we have (as in \eqref{eq:t_h_equality})
\begin{equation}\label{eq:h_t_equality}
  \text{$\Psi_m(h)(w) = h(w)$ for all $w\in W_0^+\setminus U_0^-$}.
\end{equation}
Therefore, as $\intb V^i_0 \subseteq B^i$ by
\eqref{eq:V_i_0_as_sublevel_component} and part 4 of
Proposition~\ref{prop: basic LC prop}, as $\intb V^i_0\subseteq
W_0^+\setminus U_0^-$ by \eqref{eq:W_U_V_steep
containment_relations} and as $u\in B^i\cap(W_0^+\setminus U_0^-)$,
the containment \eqref{eq:V_i_boundary_in_t_level_set} will follow
once we show that $B^i \cap (W_0^+\setminus U_0^-)$ is connected. To
see this, note that as $h(\intb W_0)>h(\intb V^i_0)$ and $u\in W_0$,
it follows from the definition of $B^i$ that $B^i\subseteq W_0$.
Moreover, as $\intb U_0\cup\extb U_0$ is connected by
Proposition~\ref{prop: Timar II} and as $h(\intb V^i_0)
> h(\intb U_0)$, it follows from the definition of $B^i$ that if
$B^i\cap U_0\neq\emptyset$ then $\intb U_0\cup\extb U_0\subseteq
B^i$. Since $B^i$ is connected by its definition, we conclude that
$B^i\setminus U_0^-$ is also connected. It follows that $B^i \cap
(W_0^+\setminus U_0^-)$ is connected, as we wanted to show, implying
\eqref{eq:V_i_boundary_in_t_level_set} and
\eqref{eq:V_i_boundary_in_r_level_set}.

To prove \eqref{eq:V_i_0_boundary_containment}, it remains to show
that
\begin{equation}\label{eq:V_i_ext_boundary_notin_r_level_set}
  \pi(\extb V^i_0)\subseteq \LL_{\pi\circ \Psi_m(h)}^{h(\intb V^i_0)+}(\pi(u),\pi(v))^c.
\end{equation}
Write $C^i:=\LL_{\pi\circ \Psi_m(h)}^{h(\intb V^i_0)+}(\pi(u))$. As
both $v\in W_0^+\setminus V^i_0$ and $\extb V^i_0 \subseteq
W_0^+\setminus V^i_0$ by \eqref{eq:W_U_V_steep
containment_relations}, the containment
\eqref{eq:V_i_ext_boundary_notin_r_level_set} is a consequence of
\begin{align}
  &\pi(W_0^+\setminus V^i_0)\text{ is connected},\label{eq:W_0_minus_V_i_0_connected}\\
  &C^i\cap\pi(W_0^+\setminus V^i_0) = \emptyset.\label{eq:C_i_no_intersection}
\end{align}
Let us show that $W_0^+\setminus V^i_0$ is connected, which will
imply \eqref{eq:W_0_minus_V_i_0_connected}. Indeed, the facts that
$V^i_0\subseteq W_0$ by \eqref{eq:W_U_V_steep containment_relations}
and $h(\extb W_0)\ge h(\intb V^i_0)+2$ imply that
$(V^i_0)^{++}\setminus V^i_0\subseteq W_0^+\setminus V^i_0$.
Proposition~\ref{prop: Timar II} shows that $(V^i_0)^{++}\setminus
V^i_0$ is connected. Thus, as $W_0^+$ is connected we obtain that
$W_0^+\setminus V^i_0$ remains connected.

Let us now prove \eqref{eq:C_i_no_intersection}. Since, by
\eqref{eq:h_t_equality} and \eqref{eq:W_U_V_steep
containment_relations}, we have the inequalities $\Psi_m(h)(\extb W_0)>h(\intb
V^i_0)$ and $\Psi_m(h)(\extb V^i_0)>h(\intb V^i_0)$, it suffices to
show that
\begin{equation*}
  \pi(u)\notin \pi(W_0^+\setminus V^i_0).
\end{equation*}
Suppose, in order to obtain a contradiction, that $u+z\in W_0^+\setminus V^i_0$ for some $z\in n\Z^d$. Recalling \eqref{eq:V_i_0_and_W_0_share_translations}, let $k$ be the integer satisfying $V^i_0 + z = V^i_0 + k \Delta$ and $W_0 + z = W_0 + k \Delta$. Necessarily $k\neq 0$ as otherwise, since $u\in V^i_0$ by \eqref{eq:W_U_V_steep containment_relations}, also $u + z \in V^i_0 + z = V^i_0$ and hence $u+z\notin W_0^+\setminus V^i_0$. For $k\neq 0$, as $u\in W_0\setminus ((W_0 - \Delta)^+)$ by \eqref{eq:W_U_V_steep containment_relations}, Proposition~\ref{prop: U and V exist} and the fact that $h(u)=h(\intb W_0)-\frac{\delta}{2}+2>h(\intb(W_0-\Delta))=h(\intb W_0)-\delta$, it follows that $u+z\in (W_0 + k \Delta)\setminus ((W_0 - (k-1)\Delta)^+)$. This contradicts our assumption that $u+z\in W_0^+\setminus V^i_0$ as $(W_0 + k \Delta)\setminus ((W_0 - (k-1)\Delta)^+)$ is disjoint from $W_0^+\setminus (W_0 - \Delta)$ and $(W_0 - \Delta)\subseteq V^i_0$ by \eqref{eq:W_U_V_steep containment_relations} and Proposition~\ref{prop: U and V exist}. This finishes the proof of \eqref{eq:C_i_no_intersection} and hence the proof of the proposition.
\end{proof}
We may now deduce Theorem~\ref{thm: QP steep QP0 embedding} in a
straightforward manner.
\begin{proof}[Proof of Theorem~\ref{thm: QP steep QP0 embedding}]
Let $m\in6\ZZ\setminus\{\zero\}$. Assume, without loss of generality (as $|\QP_m|$ does not change when permuting the coordinates of $m$ and replacing $m$ with $-m$), that the first coordinate of $m$ is positive and is the largest in absolute value among all
coordinates. Recall from \eqref{eq:p_ell_def} that for each $h\in\QP_m$ we may define $p = p(h)$ and $\ell = \ell(h)$. As both $p(h)$ and $\ell(h)$ are positive integers it follows from \eqref{eq:p_ell_def} that $1\le p(h)\le \frac{m_1}{6}$. Denote $\QP_{m,p}:=\{h \in \QP_m\,:\, p(h)=p\}$.

Proposition~\ref{prop:boundary_of_level_component_containment} together with \eqref{eq:long_sum_of_level_component_boundaries} show that for each $h\in\QP_{m,p}$, denoting $r:=\pi\circ \Psi_m(h)$, there exist $u,v\in\TT$ satisfying $r(v) = r(u) + 3p-2$ and $\sum_{j=0}^{3p-3}|\partial \LC_{r}^{(r(u)+j)+}(r(u),r(v))|\ge \frac{m_1}{6} n^{d-1}$. This allows us to apply Proposition~\ref{propos: steep long level set estimate}, using also that $\pi\circ\Psi_m$ is
one-to-one and that $|\hm(\TT)|=|\QP_{\zero}|$, to deduce that in high dimensions,
\begin{multline*}
  \frac{|\QP_m|}{|\QP_\zero|} = \frac{|(\pi\circ\Psi_m)(\QP_m)|}{|\hm(\TT)|}=\frac{\sum_{p=1}^{m_1/6} |(\pi\circ\Psi_m)(\QP_{m,p})|}{|\hm(\TT)|}\\
  \le \frac{m_1}{6} \exp\left(-\frac{c m_1 n^{d-1}}{d\log^2 d}\right)\le \exp\left(-\frac{c' m_1 n^{d-1}}{d\log^2 d}\right)
\end{multline*}
for some $c,c'>0$, as we wanted to show.
\end{proof}

\section{Near optimality of the bound}\label{sec:near_optimality}
In this section we prove
Proposition~\ref{prop:rigidity_lower_bound}.

Assume that the dimension $d$ is sufficiently large for the
following calculations and fix an even integer $n$. Let $f$ be a uniformly
chosen proper $3$-coloring of $\TT$. Define the events
\begin{equation*}
A_x:=\{f(y) = 0\text{ for all $y\in\TT$ with $\dist(x,y)=2$}\},\quad
x\in\TT.
\end{equation*}
The following claim is a consequence of our main theorem,
Theorem~\ref{thm: main}.
\begin{lemma}\label{lem:A_x_prob}
  For each $x\in\TT$ we have
  \begin{equation*}
    \P(A_x) \ge \frac{1}{7}.
  \end{equation*}
\end{lemma}
\begin{proof}
  Observe that Theorem~\ref{thm: main} and
  Markov's inequality imply that there exists a constant $c>0$ for
  which
  \begin{equation*}
    \P\bigg(\min_{i\in\{0,1\}} \cp_{i,k}(f)\ge \exp\left(-\frac{cd}{\log^2
    d}\right)\bigg) \le \exp\left(-\frac{cd}{\log^2 d}\right),\quad
    k\in\{0,1,2\}.
  \end{equation*}
  By the union bound,
  \begin{equation}\label{eq:E_event_small_probability}
    \P(E) \le 3\exp\left(-\frac{cd}{\log^2 d}\right),
  \end{equation}
  where
  \begin{equation*}
    E:=\left\{\exists k\in\{0,1,2\}, \min_{i\in\{0,1\}} \cp_{i,k}(f)\ge \exp\left(-\frac{cd}{\log^2
    d}\right)\right\}.
  \end{equation*}
  As the coloring $f$ is proper, on the event $E^c$ we necessarily have some
  random $k_0\in\{0,1,2\}$ and $i_0\in\{0,1\}$ such that
  \begin{equation*}
    \cp_{i_0,k_0}(f)\ge 1 - 2\exp\left(-\frac{cd}{\log^2
    d}\right).
  \end{equation*}
  By the homogeneity of the torus, it follows that for each $y\in
  V^{i_0}$,
  \begin{equation*}
    \P(f(y) \neq k_0\,|\, E^c, i_0) \le 2\exp\left(-\frac{cd}{\log^2
    d}\right).
  \end{equation*}
  We conclude that for each $x\in V^{i_0}$,
  \begin{equation*}
    \P(f(y) = k_0\text{ for all $y\in\TT$ with
    $\dist(x,y)=2$}\,|\, E^c, i_0) \ge 1 -
    8d^2 \exp\left(-\frac{cd}{\log^2
    d}\right).
  \end{equation*}
  Using the homogeneity of the torus again and symmetry between the
  $3$ colors, for each $x\in\TT$,
  \begin{equation*}
    \P(A_x\,|\, E^c) \ge \frac{1}{6}\left(1 -
    8d^2 \exp\left(-\frac{cd}{\log^2
    d}\right)\right),
  \end{equation*}
  from which the lemma follows for sufficiently large $d$
  using \eqref{eq:E_event_small_probability}.
\end{proof}
We note for later use that for each $x\in\TT$, by the domain Markov
property,
\begin{equation}\label{eq:prob_of_ones}
  \P(f(x) = 1\,|\, A_x) = \frac{1}{2^{2d}+2}\quad\text{and}\quad\P(f(x+e)=1\,|\, A_x)=\frac{1}{2},
\end{equation}
where $e$ is a standard basis vector.

Let $T\subset V^0$ be such that for any two distinct $x_1,x_2\in T$
we have $\dist(x_1,x_2)\ge 4$ and $|T|\ge \frac{|V^0|}{5d^2}$. Such
a set may be constructed greedily as for any $x\in V^0$, $|\{y\in
V^0\colon \dist(x,y)\le 3\}|\le 5d^2$. Let $T_0\subset T$ be the
collection of $x\in T$ for which the event $A_x$ occurs and write
$S=|T_0|$. We have $\E(S)\ge \frac{1}{7}|T|$ by
Lemma~\ref{lem:A_x_prob}, whence by Markov's inequality for $|T|-S$
we obtain
\begin{equation}\label{eq:S_prob_lower_bound}
 \P\left(S > \frac{1}{8}|T|\right)= 1-\P\left(|T|-S\ge \frac78 |T|\right)\ge 1-\frac{8}{7}\frac{|T|-\E(S)}{|T|} \ge \frac{1}{49}.
\end{equation}

For each $x\in \TT$ define the $1$-ball around $x$ by
\begin{equation*}
  B_x:=\{y\in\TT\colon \dist(x,y)\le 1\}.
\end{equation*}
Let $\mathcal{F}_T$ be the sigma algebra generated by the values
$f(y)$, where $y$ ranges over all vertices in $\TT$ satisfying that
$\dist(x,y)=2$ for some $x\in T$. Observe that $T_0$ is measurable with respect to
$\mathcal{F}_T$. Observe further that conditioned on $\mathcal{F}_T$, by the domain Markov property, the values that $f$ takes on each $B_x$,
$x\in T_0$, are uniformly sampled $3$-colorings of $B_x$ with
zero boundary conditions on $\extb B_x$, and these colorings are independent between the different $x\in T_0$.
Fix an arbitrary (measurable with respect to $\mathcal{F}_T$) partition $T_0^0 \uplus T_0^1 =T_0$ with $|T_0^0| = \lfloor\frac{S}{2}\rfloor$
and $|T_0^1|=\lceil\frac{S}{2}\rceil$.
We then have that, conditioned on $\mathcal{F}_T$,
\begin{equation*}
  \min_{i\in\{0,1\}} \cp_{i,1}(f) \ge \frac{\min(|\{x\in T_0^0\ :\ f(x)=1\}|,|\{x\in T_0^1\ :\ f(x+e_1)=1\}|)}{|V^0|} = \frac{\min(X,Y)}{|V^0|},
\end{equation*}
where $X,Y$ are independent binomial random variables satisfying, by \eqref{eq:prob_of_ones},
$X\sim\text{Bin}\left(\lfloor\frac{S}{2}\rfloor, \frac{1}{2^{2d}+2}\right)$ and
$Y\sim\text{Bin}\left(\lceil\frac{S}{2}\rceil, \frac{1}{2}\right)$.
The fact that $S$ is measurable with respect to $\mathcal{F}_T$, together with \eqref{eq:S_prob_lower_bound}, now allows to
conclude that
\begin{equation*}
  \E\bigg(\min_{i\in\{0,1\}} \cp_{i,1}(f)\bigg)\ge \frac{c'}{d^2}2^{-2d},
\end{equation*}
for some $c'>0$. As the color $k=1$ is arbitrary, this concludes the
proof of the proposition.

\section{Remarks and Open Problems}\label{sec: Rem Op}
In this section we discuss a few open problems and make a remark.
\begin{enumerate}
  \item (Tori with odd side length) In this work we consider a uniformly sampled proper $3$-coloring of a high-dimensional discrete torus with even
  side length. Our main result is that for such a coloring, with high probability, one of
  the two bipartition classes is dominated by a single color. How
  will this result change if we take $n$, the side length of the torus, to be odd?
  since tori with odd side length are no longer bipartite, some change must occur.
  We expect that in this situation, we will find in a typical coloring
  three `pure phase' regions. Each of these regions will have a distinct dominant color coloring one of its bipartition classes
  while the two remaining colors equally dominate the other bipartition class. Every two regions will be separated by a single
  long odd interface (of size roughly $n^{d-1}$), and the vertices on each side of the interface will be colored by
  the dominant color of their region.
    \item (Positive temperature) In physical terminology, a uniformly chosen proper $3$-coloring
  is the zero-temperature case of the antiferromagnetic 3-state
  Potts model. The positive temperature version of this model is
  defined as follows. A $3$-coloring $f$, not necessarily proper, of the
  underlying graph is sampled with probability proportional to
  $\exp(-\beta H(f))$, where $\beta>0$ is a parameter proportional
  to the inverse temperature and $H(f)$ is the number of edges
  $(u,v)$ for which $f(u)=f(v)$. We expect that the analog of
  Theorem~\ref{thm: main} continues to hold when the temperature is small, but
  positive (that is, when $\beta$ is sufficiently large). Proving
  this is complicated by the fact that non-proper $3$-colorings are no longer related to height
  functions.
  \item (Larger number of colors) As explained in Section~\ref{sec: BG},
      it is expected that Theorem~\ref{thm: main} has a natural extension
      to proper colorings of the torus with more than 3 colors.
      Specifically, that for each $q$ there is some $d_0(q)$ such that if
      $d\ge d_0(q)$ then a typical proper $q$-coloring of $\TT$ has the
      property that the $q$ colors split into two sets of sizes $\lfloor
      q/2 \rfloor$ and $\lceil q/2\rceil$ with each bipartition class
      dominated by colors from one of the two sets. Proving this is wide
      open even for the case $q=4$. A result of Vigoda \cite{V00} implies
      that $d_0(q)\ge \frac{3}{11} q$. In \cite[Conjecture 5.3]{EG12} it is conjectured
      that $d_0(q) = q/2$, at least in the sense that certain ``long range influences'' exist if and only if $d\ge q/2$. However, any result showing that
      $d_0(q)<\infty$ will constitute a major advance.
\end{enumerate}

We end with the following remark. Our work extends
certain results from \cite{PHom}. The results in \cite{PHom} were
proven in greater generality than simply for the torus $\TT$. There,
also tori with non-equal side lengths were considered, of the form
$\mathbb{T}^1_{n_1}\times\mathbb{T}^1_{n_2}\times\dots\times
\mathbb{T}^1_{n_d}$. These include, in particular,
``two-dimensional'' tori of the form $\mathbb{T}_n^2\times
\mathbb{T}_2^{d}$ for $d$ a fixed large constant. In our work, for
simplicity, we considered only the case of the torus $\TT$. However,
it seems that our arguments can be adapted with no difficulty to the
more general tori for which results were obtained in \cite{PHom}.

\subsection{Acknowledgments}\label{sec: Ack}

We thank Eugenii Shustin for useful discussions on the relation
between our work and algebraic topology. We also thank Naomi
Feldheim, Wojciech Samotij and Yinon Spinka for many useful comments
which greatly improved the quality of the paper. We are also
grateful to the anonymous referees whose insightful comments
significantly helped to improve the presentation.

\end{document}